\documentclass[generic,noinfoline]{imsart}

\RequirePackage[OT1]{fontenc}
\usepackage[table]{xcolor}
\RequirePackage{amsmath,amssymb,amsthm,mathrsfs,enumerate,bm,makeidx,xcolor,multirow, graphicx, float}
\RequirePackage[colorlinks,citecolor=blue,urlcolor=blue]{hyperref}
\RequirePackage{natbib}
\RequirePackage{lscape} 
\usepackage{booktabs}
\usepackage{ragged2e}
\usepackage{geometry}
\usepackage{subfig}
\usepackage{accents}
\newcommand{\vardbtilde}[1]{\tilde{\raisebox{0pt}[0.85\height]{$\tilde{#1}$}}}

\numberwithin{equation}{section}

\newcommand{\E}{\mathbb{E}}

\newtheorem{thm}{Theorem}
\newtheorem{lem}{Lemma}

\newtheorem*{defn*}{Definition}

\theoremstyle{definition}
\theoremstyle{definition}
\theoremstyle{definition}\newtheorem{assumption}{Assumption}
\theoremstyle{remark}
\theoremstyle{definition}
\theoremstyle{definition}
\theoremstyle{definition}
\theoremstyle{definition}
\theoremstyle{definition}
\theoremstyle{definition}

\RequirePackage{lineno}
\newcommand*\patchAmsMathEnvironmentForLineno[1]{%
	\expandafter\let\csname old#1\expandafter\endcsname\csname #1\endcsname
	\expandafter\let\csname oldend#1\expandafter\endcsname\csname end#1\endcsname
	\renewenvironment{#1}%
	{\linenomath\csname old#1\endcsname}%
	{\csname oldend#1\endcsname\endlinenomath}}%
\newcommand*\patchBothAmsMathEnvironmentsForLineno[1]{%
	\patchAmsMathEnvironmentForLineno{#1}%
	\patchAmsMathEnvironmentForLineno{#1*}}%
\AtBeginDocument{%
	\patchBothAmsMathEnvironmentsForLineno{equation}%
	\patchBothAmsMathEnvironmentsForLineno{align}%
	\patchBothAmsMathEnvironmentsForLineno{flalign}%
	\patchBothAmsMathEnvironmentsForLineno{alignat}%
	\patchBothAmsMathEnvironmentsForLineno{gather}%
	\patchBothAmsMathEnvironmentsForLineno{multline}%
}

\begin{document}

\title{\huge Doubly Robust Semiparametric Difference-in-Differences Estimators with High-Dimensional Data }
\author{Yang Ning \thanks{Department of Statistical Science, Cornell University, yn265@cornell.edu}, Sida Peng\thanks{Microsoft Research, sidapeng@microsoft.com}, Jing Tao \thanks{Department of Economics, University of Washington, jingtao@uw.edu} }
\date{}

\maketitle

	\vspace*{0.1in}

\begin{abstract}
This paper proposes a doubly robust two-stage semiparametric difference-in-difference estimator for estimating heterogeneous treatment effects with high-dimensional data. Our new estimator is robust to model miss-specifications and allows for, but does not require, many more regressors than observations. The first stage allows a general set of machine learning methods to be used to estimate the propensity score. In the second stage, we derive the rates of convergence for both the parametric parameter and the unknown function under a partially linear specification for the outcome equation. We also provide bias correction procedures to allow for valid inference for the heterogeneous treatment effects. We evaluate the finite sample performance with extensive simulation studies. Additionally, a real data analysis on the effect of Fair Minimum Wage Act on the unemployment rate is performed as an illustration of our method. An R package for implementing the proposed method is available on Github.\footnote{\url{https://github.com/psdsam/HDdiffindiff}} \\


\end{abstract}

\maketitle

\newpage 
\section{Introduction}
This paper proposes a doubly robust two-stage semiparametric difference-in-difference estimator for estimating heterogeneous treatment effects conditional on high-dimensional covariates. The difference-in-difference (DiD) design has been widely adopted in policy evaluation from academia to industry when a real experiment is expensive or infeasible. When a policy/feature only affects a fraction of the population, DiD design can identify the average treatment effect on the treated (ATT) based on observational data. It is based on the simple idea of comparing the difference in pre and post-treatment outcome of those individuals who are affected and those who are not affected by the policy/feature of interest.

A key identification assumption for the classical DiD design is the parallel trend assumption. It requires that the outcome variables for treated and non-treated individuals would have followed parallel paths over time in the absence of treatment. However, this assumption ignores the potential selection problem due to individual heterogeneity. For example, a company might want to evaluate the effect of an email marketing campaign (advertisement through email with an embedded promo link). A researcher can compare the customers' conversion rate (whether a purchase was made) before and after the campaign for a group of treated customers (click into the link) and a group of non-treated customers (did not click into the link). If existing customers are more likely to click into the link and also more likely to purchase again even without the campaign intervention, the classical DiD estimator will lead to a positive bias and exaggerate the effect of the campaign. To account for such case, \cite{Abadie2005} proposed a two-stage semiparametric estimator with the so-called conditional parallel trend assumption. In this framework, a propensity score is estimated in the first stage to explicitly account for any observed confounders that may affect both the treatment take-up as well as the outcome growth trend.

While the semiparametric DiD (semi-DiD) estimator is comprehensively used by researchers in academia and industry, three major challenges arise in practice. First, the semi-DiD estimator becomes difficult to implement when there exist too many covariates. Following the previous example, researchers may also observe customers' browsing history and may suspect customers who visited certain (unknown) websites are more likely to click into the link while also more likely to make the purchase. However, the semi-DiD estimator cannot be implemented if the number of attributes exceeds the number of observations. Therefore, researchers may be forced to select covariates based on their intuition or insights and may lead to further biases (\cite{Belloni2017}). Even when the number of observations is larger than the number of covariates, the semi-DiD estimator may still contain a large bias when too many covariates are included (\cite{Matias2019}). Second, the semi-DiD estimator is sensitive to the choice of specification for propensity score estimation. This is a similar problem for the inverse propensity score weighted (IPW) estimator and it becomes a more subtle problem if machine learning methods (e.g. random forrest, neural network, etc.) are used to predict propensity score. Third, conditional or heterogeneous treatment effects on treated (ATT) is often of interest to practitioners. While semi-DiD framework provides a way to estimate conditional or heterogeneous ATT under a vector of low-dimensional covariates, it is not clear how this framework can be extended to the high-dimensional case, (e.g. how to develop estimation and inference methods and theory).

In this paper, we propose a new estimator to solve the above three problems. Our doubly robust DiD (Dr-DiD) estimator is robust to model miss-specifications under high-dimensional covariates. We show the desired rate of convergence of our estimator can be achieved as long as either the propensity score function or the outcome equation can be approximated asymptotically at a moderate rate. Thus, a general set of machine learning methods can be used in our framework. Although diff-in-diff design is an ATT estimator, we show that the semi-DiD estimator can be extended to an augmented inverse propensity score weighted (AIPW) estimator. We show that the extended AIPW form still preserves the doubly robustness property under the parallel trend assumption.

To further incorporate high-dimensional covariates and heterogeneous treatment effects, we consider a partially linear specification in the potential outcome estimation. The partially linear form is composed of a nonparametric specification from a set of low-dimensional covariates as well as a linear parametric specification from a set of high-dimensional covariates, which provides a flexible functional form to model the potential outcome. This is a very useful specification in real world application. For example, researchers may be interested in the nonlinear relationship between the outcome variable and a set of covariates while also facing a large number of indicator variables such as age, gender and region.  

We derive the rate of convergence for our estimator as well as a de-bias procedure for inference. We show that the high-dimensional linear part of the estimator can achieve the oracle rate of convergence, while the nonparametric part maintains the nonparametric rate of convergence. With bias correction, the high-dimensional linear part can achieve normality at $\sqrt{n}$-rate, while the nonparametric part can achieve the normality at the nonparametric rate. Finally, we demonstrate the finite sample performance of our estimator in a simulation study and apply our estimator to study the effect of the Fair Minimum Wage Act on the unemployment rate using the data collect by \cite{CallawayLi2020}. We show that the heterogeneity in the effect of this policy can be explained by variations in demographics. More specifically, counties with larger population and higher median income are more likely to suffer from an increase in the unemployment rate. These findings coincide with the canonical economic theory on unemployment rate. For example, a higher median income level implies a higher substitution cost for workers currently at minimum wage and thus leads to an increase in the unemployment rate when minimum wage rises. On the other hand, regions with larger population sizes have more labor supply and thus a minimum wage raise can also lead to a surplus.

In summary, the main contributions of this work are as follows: first, we propose a doubly robust approach to estimate heterogeneous ATT conditional on covariates for DiD models that allows either the propensity scores or the model for ATT to be misspecified. Second, we propose a regularized two-stage estimation procedure for DiD models that allows (i) suitable machine learing tools to estimate the first-stage propensity socres and (ii) high-dimensional covariates and nonparametric specification for the heterogeneous ATT in the second-stage. Third, we provide a novel approach to simultaneously correct the biases due to both stages and provide a novel statistical inference procedure based on the de-biased estimator. Finally, as a useful byproduct, we derive novel estimation and inference methods for a partially linear model for both the high-dimensional parametric parameter and the nonparametric function. 

\subsection{Related Literature}
This paper is related to the vast literature on robust estimation and inference for treatment effects models; see for example, \cite{robins1994estimation}, \cite{tan2006distributional}, \cite{chen2008semiparametric}, \cite{graham2012inverse}, \cite{OkuiTan2012doubly}, \cite{farrell2015robust}, \cite{vermeulen2015bias}, \cite{ogburn2015doubly}, \cite{Belloni2017}, \cite{lee2017doubly}, \cite{Chernozhukov2016}, \cite{sloczynskiWooldridge2018}, \cite{kennedy2019robust} and \cite{Tan2020} among many others. Our work is particularly closely related to a recent work independently developed by \cite{SantAnnaZhao2019}. Both are based on the seminal framework proposed in \cite{Abadie2005}. Our estimator complements theirs as we focus on estimation and inference for heterogeneous ATT conditional on covariates in a high-dimensional setting while \cite{SantAnnaZhao2019} focus on efficient estimation of ATT  when the dimension of covariates is fixed and is much smaller than the sample size.

This paper also contributes to the literature by connecting the widely used DiD estimator with the machine learning/ high-dimensional statistic literature. The DiD estimator has been an active research field in the economic literature, e.g. \cite{Card1994}, \cite{Abadie2005}, \cite{Athey2006}, \cite{Imai2019}, \cite{Athey2019}, \cite{Callaway2019} among others. Our paper proposes a specific DiD estimator so that high-dimensional/machine learning tools can be applied. This paper also contributes to a set of works that apply machine learning tools to casual inference. This includes \cite{chernozhukov2016locally}, \cite{Belloni2017}, \cite{semenova2017estimation}, \cite{Chernozhukov2016}, \cite{Syrgkanis2019}, \cite{fan2020estimation}, \cite{Tan2020}, etc. Our work distinguish this literature in the following two ways. First, we propose a doubly robust diff-in-diff estimator in the high-dimensional/machine learning setting that has not been studied. Second, to our best knowledge,  the doubly-robust estimators in these papers use various high-dimensional set of covariates and machine learning methods to deal with the selection into treatment. However, the ultimate parameter of interest in the second-stage is a low-dimensional subset of the covariates so traditional nonparametric estimator can apply.  By contrast, our parameter of interest in the second-stage contains both high-dimensional covariates in the parametric part and an unknown function, which brings substantial challenges for estimation and inference. We construct a new Neyman orthogonal moment condition (\cite{chernozhukov2016locally}) and propose de-biased estimators for both the parametric parameters and the nonparametric function in the second-stage to construct valid confidence intervals. 

Moreover, as useful by-products, we provide an inference method for a partially linear model for both the parametric parameter and the nonparametric function when the linear part contains high-dimensional covariates. Thus, this work is related to recent discussion in \cite{muller2015partial}, \cite{MaHuang2016}, \cite{YLC2019}, \cite{ZhuYuGuang2019high}, among others. Our paper departs from the existing papers in the following three aspects. First, our partially linear form is in the second-stage outcome equation so estimation and inference results have to take the first-stage estimators into consideration, while the existing papers focus on a one-stage regression problem. Second, the above papers propose estimators with penalized estimation in functional space. As is pointed out in \cite{shen1997methods}, this approach often leads to undesirable properties of the estimates, such as inconsistency and roughness. Moreover, such an optimization procedure is difficult to implement in practice. Therefore, we consider the extension to sieve estimation in our estimator by approximating the nonparametric function with sieves so that we carry out optimization within a dense subset of the infinite- dimensional space, which is finite-dimensional and therefore easy to work with. 
Finally, the existing literature is concerned with asymptotic theories and inference procedures for the parametric parameters only. The nonparametric function is profiled out as an infinite-dimensional nuisance parameter. This paper considers the joint asymptotic theory and inference methods when the parameters of interest are not only the parametric parameter, but also the nonparametric function. To the best of our knowledge, these results are new to the literature. We show that the parametric parameter converges to a normal distribution with a $\sqrt{n}$-rate. The parametric estimator achieves the semiparametric efficiency bound when the error term is homoskedastic,  while the functional of a nonparametric function converges to a normal distribution with a nonparametric rate. We observe that the marginal asymptotic variance for the nonparametric component is, in general, different from those derived without the high-dimensional parametric parameter, i.e., \cite{Newey1997}, \cite{BCCK2015} and \cite{ChenChristensen2015}. This result may be of independent interest to the readers. \\

\subsection{Organization of the Paper}
The paper is organized as follows. The estimator is proposed in Section 2. Rate of convergence of the estimator and inference theory are developed in Sections 3 and 4, respectively. Section 5 presents extensive simulation results to evaluate the finite sample performance. An empirical study on the effect of the Fair minimum Wage Act on the unemployment rate is presented in Section 6. Section 7 concludes.  We defer the proofs to the Appendices.

\subsection{Notation}
For a vector ${\bf x}=(x_{1},\dots,x_{d})^{\top}\in\mathbb{R}^{d}$ and $1\leq q\leq\infty$, let $\|{\bf x}\|_{q}=\left(\sum_{i=1}^{d}|x_{i}|^{q}\right)^{1/q}$, $\|{\bf x}\|_{\infty}=\max_{1\leq i\leq d}|x_{i}|$, $\|{\bf x}\|_{0}=\vert{\rm supp}({\bf x})\vert,$ where ${\rm supp}({\bf x})=\{j:x_{j}\neq0\}$ and $|a|$ is the cardinality of a set $a$. For a symmetric matrix $A$, let $\Lambda_{\max}(A)$ and $\Lambda_{\min}(A)$ be the maximum and minimum eigenvalues of $A$. For a matrix $B=[B_{jk}]$, let $\Vert B\|_{\max}=\max_{jk}\vert B_{jk}\vert,$ $\|B\|_{1}=\sum_{jk}\vert B_{jk}\vert$, $\|B\|_{2}=\sqrt{\Lambda_{\max}(B^{\top}B)}$ and $\|B\|_{\ell_{\infty}}=\max_{j}\sum_{k}|B_{jk}|$.  For any function $f:\mathcal{Z}\rightarrow\mathbb{R},$ let $\|f\|_{\infty}=\sup_{z\in\mathcal{Z}}\vert f(z)\vert$, $\|f\|_{P,2}=\sqrt{\mathbb{E}f^{2}(z)}$ and $\|f\|_{n}=\sqrt{n^{-1}\sum_{i=1}^{n}f^{2}(Z_{i})}$. We denote $I_{d}$ as the $d\times d$ identity matrix. For a set $S\subseteq\{1,\dots,d\}$,
let ${\bf x}_{S}=\{x_{j}:j\in S\}$ and $S^{c}$ be the complement of $S$. Let $S_0$ be the set of all non-zero components of $\beta_0$ and $s_0=|S_0|$. We use $\bigtriangledown_{S}f({\bf x})$ to denote the gradient of $f({\bf x})$ with respect to ${\bf x}_{S}$. Given $a,b\in\mathbb{R}$, let $a\vee b$ and $a\wedge b$ denote the maximum and minimum of $a$ and $b$. For two positive sequences $a_{n}$ and $b_{n}$, let $a_{n}\asymp b_{n}$ denote $C\leq a_{n}/b_{n}\leq C'$ for some $C,C'>0$; let $a_{n}\lesssim b_{n}$ denote $a_{n}\leq Cb_{n}$ for some constant $C>0$. Also, we write $a_{n}=O(b_{n})$ if $|a_{n}|\leq C|b_{n}|$. We use $X_{n}\rightarrow_{p}a$ for some constant $a$ if a sequence of random variables $X_{n}$ converges in probability to $a$. Similarly, if $X_{n}$ converges weakly to $X$ we write $X_{n}\rightsquigarrow X$ for some random variable $X$. For notational simplicity, we use $C$, $C'$ and $C''$ to denote generic constants, whose values can change from line to line. Let $\mathbb{E}_{n}f=\frac{1}{n}\sum_{i=1}^{n}f(X_i)$ and $\mathbb{G}_{n}f=\mathbb{E}_{n}f-\mathbb{E}f$.

A random variable $X$ is called sub-exponential if there exists some positive constant $K_{1}$ such that $\mathbb{P}(|X|>t)\leq\exp(1-t/K_{1})$ for all $t\geq0$. The sub-exponential norm of $X$ is defined as $\|X\|_{\psi_{1}}=\sup_{q\geq1}q^{-1}(\mathbb{E}|X|^{q})^{1/q}$. Similarly, a random variable $X$ is called sub-Gaussian if there exists some positive constant $K_{2}$ such that $\mathbb{P}(|X|>t)\leq\exp(1-t^{2}/K_{2}^{2})$ for all $t\geq0$. And the sub-Gaussian norm of $X$ is defined as $\|X\|_{\psi_{2}}=\sup_{q\geq1}q^{-1/2}\left(\mathbb{E}|X|^{q}\right)^{1/q}.$

\section{Doubly Robust DiD Estimator}
Denote $Y^0(i, t)$ as the potential outcome of individual $i$ at time $t$ being not treated and $Y^1(i, t)$ as the potential outcome of individual $i$ at time $t$ being treated. We cannot observe both $Y^0(i,t)$ and $Y^1(i,t)$ for the same individual, but we observe the realized outcome for individual $i$ at time $t$ as
\[Y(i, t) = D_iY^1(i, t) + (1-D_i)Y^0(i, t)\] 
\noindent where $D_i$ is the treatment status at time $t=1$. For some observed covariates $W_i = (X_i, Z_i)$, we want to learn the heterogeneous treatment effect on the treated conditional on the covariates $W_i$ such that
\begin{equation}
ATT(W_{i})=\tau_{0}(W_i):=\mathbb{E}[Y^1(i, 1) - Y^0(i, 1)|W_i, D_i = 1]. \label{def:tau_0}
\end{equation}

Our parameter of interest is different from ATT (e.g. \cite{SantAnnaZhao2019}), which is defined as
\begin{equation}
ATT=\mathbb{E}\left[Y^{1}(i,1)-Y^{0}(i,1)|D_{i}=1\right]. \label{def:tau_0_1}
\end{equation}
While it is useful to know (\ref{def:tau_0_1}), a doubly robust estimator for (\ref{def:tau_0}) conditional on the covariates could also be relevant and important in empirical applications when the parameter of interest is the heterogeneous treatment effects conditional on the covariates.

As pointed out in \cite{Abadie2005}, the conventional DiD estimator is based on the strong assumption that outcomes for treated and non-treated groups or individuals would have followed parallel paths over time in the absence of treatment. That assumption can be easily violated when differences in observed characteristics create non-parallel outcome dynamics between treated and non-treated populations.  \cite{Abadie2005} generalizes this assumption by allowing the parallel trend assumption to hold after conditioning on the covariates as follows: 
\begin{assumption} \label{Ass:parallel}
\[\mathbb{E}[Y^0(i, 1) - Y^0(i, 0)|W_i, D_i = 1] = \mathbb{E}[Y^0(i, 1) - Y^0(i, 0)|W_i, D_i = 0].\]
\end{assumption}
In addition, a full support assumption will guarantee the existence of the propensity score function.  
\begin{assumption} \label{Ass:nonzero}
With probability approaching 1, there exits a constant $c>0$ such that $\mathbb{E}[D_i=1|W_i]>c$ and $\mathbb{E}[D_i=1|W_i]<1-c$.
\end{assumption}
Together with Assumptions \ref{Ass:parallel} and \ref{Ass:nonzero}, the \cite{Abadie2005} estimand can be defined as 

\begin{equation}\label{abadie}
\mathbb{E}\left[\frac{(D_i - \mathbb{E}(D_i=1|W_i))}{\mathbb{E}(D_i=1|W_i)\mathbb{E}(D_i=0|W_i)}(Y(i, 1) - Y(i, 0)) \Big| W_i   \right].	
\end{equation}
Defining $\Delta Y_{i} := Y(i, 1) - Y(i, 0)$, we then have
\begin{equation}\label{abadie2}
\begin{split}
\mathbb{E}\left[Y^1(i, 1) -  Y^0(i, 1) |W_i, D_i = 1\right]& = \mathbb{E}\left[\frac{(D_i - \mathbb{E}(D_i=1|W_i))}{\mathbb{E}(D_i=1|W_i)\mathbb{E}(D_i=0|W_i)}\Delta Y(i) \Big| W_i   \right]\\
& =  \mathbb{E}\left[\frac{D_i\Delta Y(i)}{\mathbb{E}(D_i=1|W_i)} \Big| W_i \right] - \mathbb{E}\left[\frac{(1-D_i)\Delta Y(i)}{\mathbb{E}(D_i=0|W_i)} \Big| W_i \right].
\end{split}
\end{equation}
It is easy to see that Equation (\ref{abadie2}) is in the form of Horvitz-Thompson estimator (\cite{Horvitz1952}). As a natural extension to the IPW form estimator, we study whether a doubly robust form exists under the DiD setting and this leads to our parameters of interest as follows. Define 
\[ \Delta Y_{1i} :=Y^1(i, 1) - Y^0(i, 0), \qquad  \Phi_1(W_i) := \mathbb{E}[\Delta Y_{1i}|W_i, D_i=1],  \]
\[\Delta Y_{0 i}  : = Y^0(i, 1) - Y^0(i, 0), \qquad \Phi_0(W_i):= \mathbb{E}[\Delta Y_{0i}|W_i, D_i=0].\]
Let \[\rho_0 = \frac{D_i - \mathbb{ P}(D_i=1|W_i)}{\mathbb{ P}(D_i=1|W_i)\mathbb{E}(D_i=0|W_i)},\] 
our doubly robust estimand is defined as 
\begin{equation}\label{estimator}
\tau_{0}(W_i)=\mathbb{E}[\rho_0\left(\Delta Y_{i} - \mathbb{E}(D_i=0|W_i)\Phi_1(W_i) - \mathbb{E}(D_i=1|W_i)\Phi_0(W_i)\right) |W_i],
\end{equation}
where $\mathbb{E}[D_i=0|W_i]$, $\mathbb{E}[D_i=1|W_i]$, $\Phi_0(W_i)$ and $\Phi_1(W_i)$ are nuisance functions to be estimated from the first-stage.
\begin{lem}\label{doublerobust}
Under Assumptions \ref{Ass:parallel} and \ref{Ass:nonzero}, (i) the estimand defined in Equation (\ref{estimator}) is doubly robust in the sense that
\[\mathbb{E}[Y^1(i, 1) - Y^0(i, 1)|W_i, D_i = 1]=\mathbb{E}[\rho_0\left(\Delta Y_{i} - \mathbb{E}(D_i=0|W_i)\Phi_1(W_i) - \mathbb{E}(D_i=1|W_i)\Phi_0(W_i)\right) |W_i]\] holds provided that one of the two conditions (a) or (b) holds, even if both do not hold simultaneously:
(a) specifications $\Phi_1(W_i)$ and $\Phi_0(W_i)$ are correct,
(b) specification of $\mathbb{P}(D_i=1|W_i)$ is correct.

(ii) Let $\alpha_{0}=(\Phi_{1}(\cdot),\Phi_{0}(\cdot),\pi(\cdot))$, $\pi(W)=\mathbb{P}(D=1|W)$ and 
\[\varUpsilon(W;\alpha_{0})=\rho_0\left[\Delta Y- \left(1-\pi(W)\right)\Phi_1(W) - \pi(W)\Phi_0(W)\right].\] 
Then the moment condition $\mathbb{E}\left[\Upsilon(W_i;\alpha_{0})-\tau_{0}(W_i)|W_i=w\right]=0$
holds and the following Neyman orthogonality condition holds:
\[
\left.\partial_{r}\mathbb{E}\left[\Upsilon\left(W_i;\alpha_{0}+r(\alpha-\alpha_{0})\right)-\tau_0(W_i)|W_i=w\right]\right|_{r=0}=0.
\]

\end{lem}

Lemma \ref{doublerobust} shows that, with Assumptions \ref{Ass:parallel} and \ref{Ass:nonzero}, we can have a doubly robust estimator for $\tau_0(W_i)$ when either the regression models $\Phi_0(\cdot)$ and $\Phi_1(\cdot)$ are misspecified or the propensity score $\mathbb{P}(D_i=1|X_i)$ is misspecified. 

To model $\tau_0(\cdot)$, we consider a class of flexible high-dimensional partially linear model such that
\begin{equation}
\mathbb{E}[Y^1(i, 1) - Y^0(i, 1)|W_i, D_i = 1]=X_i^\top\beta_0+f_0(Z_i),\label{eq:pop_model}
\end{equation} 
where the linear part contains the parametric Euclidean vector $\beta_0\in\mathcal{B}\subseteq \mathbb{R}^{p}$ with $p>n$, and the nonparametric part contains an unknown function $f(\cdot): \mathcal{Z}\rightarrow \mathbb{R}$, where $\mathcal{Z}$ is a compact subset of $\mathbb{R}^{d_{z}}$. We will assume that the unknown function belongs to a smoothed function class defined in Section \ref{sec:Estimation}. 

Compare with the definition of equation (11) in \cite{Abadie2005}, we define our estimand in equation \eqref{eq:pop_model} in a partial linear form rather than approximate it with a best linear predictor. The semi-parametric structure is slightly stronger as equation (11) in \cite{Abadie2005} is satisfied if we plug in equation \eqref{eq:pop_model} and allow $g(X_k, \theta)$ to admit a partial linear specification.     

On the other hand, the partially linear specification in (\ref{eq:pop_model}) provides a flexible functional form while still allowing us to maintain the Neyman orthogonality condition when designing the estimator under the high dimensionally covariates. Theoretical properties of the semiparametric partially linear model when the dimension of $X$ is fixed and smaller than $n$ have been thoroughly discussed in the econometrics literature (\cite{engle1986semiparametric}, \cite{Robinson1988root}, \cite{ahn1993semiparametric}, \cite{donald1994series}, \cite{linton1995second}, \cite{fan1999root}, to mention only a few; see \cite{LiRacine} for a review). We complement the literature by providing new estimation and inference methods and theory when $X$ is high-dimensional. 

As a result of Lemma \ref{doublerobust} and equation \eqref{eq:pop_model}, if Assumptions \ref{Ass:parallel} and \ref{Ass:nonzero} hold, we have 
\begin{equation}
(\beta_0,f_0)=\arg\min_{(\beta\in\mathcal{B},f\in\mathcal{F})}\mathbb{E}\left[\left\{ X_{i}^{\top}\beta+f(Z_{i})-\rho_{0}\left(\Delta Y_{i}-(1-\pi(W_{i}))\Phi_{1}(W_{i})-\pi(W_{i})\Phi_{0}(W_{i})\right)\right\} ^{2}\right]. \label{eq:pop_crit}
\end{equation}
We are going to construct a two-step estimator of $(\beta_0,f_0)$ based on the sample analogue of (\ref{eq:pop_crit}), where the first-step estimator estimates $\rho_0$ and the second-step estimator estimates $(\beta_0,f_0)$. We allow the propensity score, hence $\rho_0$, to be estimated by any suitable machine learning methods as long as certain conditions in Section \ref{sec:Estimation} are satisfied.

\section{Estimation}\label{sec:Estimation}
Let $\hat{\pi}(\cdot)$, $\hat{\Phi}_1(\cdot)$ and $\hat{\Phi}_0(\cdot)$ be  nonparametric or machine learning estimators of $\pi(W_i)$, $\Phi_1(W_i)$ and $\Phi_0(W_i)$, respectively. We propose the following two-stage estimator such that
\begin{equation}\label{eq:Estimator}
(\hat{\beta}, \hat f) = \arg \min_{\beta\in\mathcal{B}, f_n\in\mathcal{F}_n} \mathbb{E}_n\left[\left(X_i^{\top}\beta + f_n(Z_i) - \hat{\rho}_i \Big(\Delta Y_i -(1-\hat{\pi}_i)\hat{\Phi}_1(W_i) - \hat{\pi}_i\hat{\Phi}_0(W_i) \Big)\right)^2\right]+\lambda \|\beta\|_1,
\end{equation}
where $f_n(\cdot)=\psi^{k_n}(\cdot)^\top\gamma_{n}$ is a sieve approximation of the unknown function $f(\cdot)\in\mathcal{F}$ with 
\[f_{0}(Z_i) = \sum_{j = 1}^{k_n} \psi_j(Z_i)\gamma_{j, n0} + r_{n}(Z_i):=f_{n0}(Z_i) +r_{ni},
\]
where $r_{ni}:=r_{n}(Z_i),i=1,\dots,n$ is a sieve approximation error that depends on the smoothness of $f_{0}$ and the sample size $n$. 

For $\alpha>0$ and  any vector $\tau=(\tau_1,\dots,\tau_{d_z})$ of $d_z$ integers, define the differential operator $D^\tau=\partial^{\tau.}/\partial z^{\tau_1}_1...z^{\tau_{d_z}}_{d_z}$, where $\tau.=\sum_{l=1}^{d_{z}}\tau_l$. For a function $g:\mathcal{Z}\rightarrow \mathbb{R}$, let 
\begin{equation}
\|g\|_{\infty,\alpha}=\max_{\tau.\leq \underline{\alpha}}\sup_z|D^{\tau}g(z)|+\max_{\tau.=\underline{\alpha}}\sup_{z,z'}\frac{|D^{\tau}g(z)-D^{\tau}g(z')|}{\|z-z'\|^{\alpha-\underline{\alpha}}}.\label{def:sup_norm}
\end{equation}
Let $C^{\alpha}_M(\mathcal{Z})$ be the set of all continuous functions $g:\mathcal{Z}\rightarrow \mathbb{R}$ with $\|g\|_{\infty,\alpha}\leq M$. We assume that $\mathcal{F}\subseteq C^{\alpha}_M(\mathcal{Z})$. Let $\psi^{k_{n}}(Z_{i})=\left(\psi_{1}(Z_{i}),\dots,\psi_{k_{n}}(Z_{i})\right)^\top$ be a $k_n\times1$ vector of basis functions, and we use the notation $\mathcal{F}_{n}$ to represent the space of sieve functions. Define the projection of $X_{ij}, i=1,\dots,n$ onto $\mathcal{F}_{n}$ as
\[\Pi_n(X_{ij}|Z)= \arg\min_{h^* \in \mathcal{F}_{n}}\|X_ij-h^*\|^2_{n} =\psi^{k_n}(Z)(\psi^{k_n}(Z)^\top\psi^{k_n}(Z))^{-1}\psi^{k_n}(Z)^\top X_ij, j=1,\dots,p\]
and $\Pi_{n,X_i|Z_i}:=(\Pi_n(X_{i1}|Z_i),\dots, \Pi_n(X_{ip}|Z_i))^{\top}, i=1,\dots,n.$

Next, define 
\[\hat S_i := \hat{\rho}(W_i) \Big(\Delta Y_i -(1-\hat{\pi}(W_i))\hat{\Phi}_1(W_i) - \hat{\pi}(W_i)\hat{\Phi}_0(W_i) \Big),\]
\[S_i := \rho_0(W_i) \Big(\Delta Y_i -(1-\pi(W_i))\Phi_1(W_i) - \pi(W_i)\Phi_0(W_i) \Big),\]
$\Pi_{n, X|Z} :=P_{Z}\bm{X}$, $P_{Z}:=\Psi_{n}\left(\Psi_{n}^{\top}\Psi_{n}\right)^{-1}\Psi_{n}^{\top}$ and $\tilde{\bm{X}} := \bm{X} - \Pi_{n, X|Z} $, where 
 $\Psi_{n}:=\Psi_{n}(\bm{Z})=\left(\psi^{k_{n}}(Z_{1})^{\top},\dots,\psi^{k_{n}}(Z_{n})^{\top}\right)^{\top}$ is a $n\times k_n$ matrix. Let $\bar{\bm{X}}=\bm{X}-\Pi_{X|Z}=(X_1^{\top}-\Pi^{\top}_{X_1|Z_1},...,X_n^{\top}-\Pi^{\top}_{X_n|Z_n})^{\top},$ where $\Pi_{X|Z}=\mathbb{E}[X_i|Z_i]$. 

Define $\bm{\eta}_{n}=\bm{S}-\bm{X}\beta_{0}-\Psi_{n}\gamma_{n0}=\bm{\epsilon}+\bm{r_{n}}$, where $\bm{\epsilon}=(\epsilon_{1},\dots,\epsilon_{i},\dots,\epsilon_{n})^{\top}$, $\epsilon_{i}=\frac{D_{i}}{\pi_{i}}\epsilon_{1i}+\frac{1-D_{i}}{1-\pi_{i}}\epsilon_{0i}$, $\epsilon_{0i}=\Delta Y_{0i}-\Phi_{0}(W_i)$, $\epsilon_{1i}=\Delta Y_{1i}-\Phi_{1}(W_i)$, $\bm{r_{n}} = (r_{n1},\dots, r_{ni}, \dots, r_{nn})^{\top}$. We have the following decomposition:
\[\|\bm{X}\beta_0+f_{n0}\|_{n}^2 = \|\tilde{\bm{X}}\beta_0\|_{n}^2 + \|\Pi_{n,X|Z}\beta_0 + f_{n0}\|_{n}^2.\]

\begin{assumption}\label{assum:data} (i) The data are i.i.d. from the distribution of $(Y^1(i, 1), Y^1(i, 0),D_i,W_i)$ conditional on $t=1$, while conditional on $t=0$, the data are i.i.d. from the distribution of $(Y^0(i, 1), Y^0(i, 0),D_i,W_i)$; (ii) $\mathcal{W}$ is compact with nonempty interior; (iii) $(\beta_0,f_0)\in\mathcal{B}\times\mathcal{F}\subseteq \mathbb{R}^{p}\times \mathcal{C}^{\alpha}_{M}(\mathcal{Z})$ is the only $(\beta,f)$ that satisfies (\ref{eq:pop_model}), where $\alpha\geq d_{z}/2$; (iv) $\mathbb{E}[f_0(Z)|X]$ does not belong to the linear span of $X$.
\end{assumption}

\begin{assumption}\label{assum:error}(i) The error terms $\epsilon_{1i}\in \mathbb{R}$ and $\epsilon_{0i}\in \mathbb{R}$ are independently distributed with $\mathbb{E}[\epsilon_{1i}|W_i]=0$ and $\mathbb{E}[\epsilon_{0i}|W_i]=0$; (ii) $\max_{1\leq i\leq n}\sup_{w\in \mathcal{W}}\mathbb{E}[|\epsilon_{i}|^{r_\epsilon}|W_i=w]\leq \mathcal{C}_\epsilon$ for some $r_\epsilon> 4$ and a positive universal constant $\mathcal{C}_\epsilon$.
\end{assumption}
In Assumption \ref{assum:data}, (i) is Assumption 3.3 in \citep{Abadie2005}. We follow the same sampling scheme to consider repeated cross sections; (ii) can be relaxed if we add a continuous nonnegative weight function in the definition of $\|\cdot\|_{\infty,\alpha}$ in (\ref{def:sup_norm}) \citep{freyberger2019practical}; and (iii) and (iv) are standard identification conditions for a partially linear model. 
Assumption \ref{assum:error} allows the error terms to be non-identically distributed and be conditionally heteroskedastic. One can replace it by the stronger sub-Gaussian assumption often used in the literature. 

As is standard in the literature for high-dimensional data, we introduce the restricted eigenvalue condition for as
\[
\Lambda^{2}_{\bar{X}}(s_0):=\min_{\delta\in\mathbb{R}^{p}\backslash\{0\},\|\delta_{S_{0}^{C}}\|_{1}\leq3\sqrt{s_{0}}\|\delta_{S_{0}}\|_{2}}\frac{\delta^{\top}\mathbb{E}\left[\bar{X}_{i}\bar{X}_{i}^{\top}\right]\delta}{\|\delta_{S_{0}}\|_{2}^{2}}>0
\]
Let $\Sigma_{\bar{X}}=\mathbb{E}[\bar{X}_{i}\bar{X}_{i}^{\top}]$, $\Sigma_{\tilde{X}}=\mathbb{E}[\tilde{X}_{i}\tilde{X}_{i}^{\top}]$ and $\Sigma_{\Pi}=\mathbb{E}[\Pi_{X_{i}|Z_{i}}\Pi_{X_{i}|Z_{i}}^{\top}]$.

\begin{assumption}\label{asseigen}(i) For each $i=1,....,n$, the covariates $X_i=\mathbf{x}$ is a sub-Gaussian vector such that for any vector $v\in\mathbb{R}^p$, $v^{\top}\mathbf{x}$ is sub-Gaussian with $\sup_{v\in\mathbb{R}^{p}:\|v\|=1}\|v^{\top}\mathbf{x}\|_{\psi_{2}}\leq K_{X}$;  
(ii) there exists constant $\mathcal{C}_{\bar{x}}>0$, 
such that $\Lambda^{2}_{\bar{X}}(s_0)>\mathcal{C}_{\bar{x}}$;
(iii) there exist constants $\mathcal{C}_{\Sigma_{\bar{X}}}>0$, and $\mathcal{C}_{\Sigma_{\Pi}}>0$ such that $\mathcal{C}_{\Sigma_{\bar{X}}}<\Lambda_{\min}(\Sigma_{\bar{X}})\leq \Lambda_{\max}(\Sigma_{\bar{X}})<1/\mathcal{C}_{\Sigma_{\bar{X}}}$ and $\mathcal{C}_{\Sigma_{\Pi}}<\Lambda_{\min}(\Sigma_{\Pi})\leq\Lambda_{\max}(\Sigma_{\Pi})<1/\mathcal{C}_{\Sigma_{\Pi}}$.
\end{assumption}

Assumption \ref{asseigen} is standard in the literature: (i) can be relaxed if we replace it by some uniform moment conditions discussed in \cite{caner2018asymptotically}; (ii) and (iii) restrict the eigenvalues. In particular, (ii) is a restricted eigenvalue condition. 

\begin{assumption}\label{approx error} There are finite constants $c_{k_n}$ and $\ell_{k_n}$ such that for each $f\in \mathcal{F}$ and for each $n$ and $k_n$, we have 
\[
\|r_n\|_{P,2}=\sqrt{\int_{z\in \mathcal{Z}}r^{2}_{n}(z)dF(z)}\leq c_{k_n}, \|r_n\|_{\infty}=\sup_{z\in \mathcal{Z}}|r_{n}(z)|\leq \ell_{k_n}c_{k_n}.
\]
\end{assumption}

\begin{assumption}\label{assfeigen}(i) The density of $Z_i$ is bounded and bounded away from zero. 
For every $k_{n}$, there exist a constant $\mathcal{C}_z>0$, which does not depend on $k_n$, such that  $\lambda_{\min}\left(\mathbb{E}[\psi^{k_{n}}(Z_{i})\psi^{k_n}(Z_{i})^{\top}]\right)>\mathcal{C}_z$; (ii) there is a sequence of constant $\xi_0(k_n)$ satisfying that $\sup_{z}\|\psi^{k_n}(z)\| \leq \xi_0(k_n)$, $\xi_{0}(k_{n})^{2}\log k_{n}/n=o_p(1)$, and $k_{n}\xi_{0}(k_{n})^{2}\log p/n=O_{p}(1)$; (iii)$\|\mathbb{E}[\tilde{X}_{i}\tilde{X}_{i}^{\top}-\bar{X}_{i}\bar{X}_{i}^{\top}]\|_{\infty}=O(\sqrt{\log p/n})$.
\end{assumption}

Assumption \ref{approx error} is Assumption A.3 in \cite{BCCK2015}. Note that $\mathcal{F}$ is a set of functions $f$ in $C^{\alpha}_M(\mathcal{Z})$, thus, $\|f\|_{\infty,\alpha}$ is bounded from above uniformly over all $f\in \mathcal{F}$. Then for instance, $c_{k_n}=O(k^{-\alpha/d_{z}})$ for the polynomial series and $c_{k_n}=O(k^{-(\alpha\wedge \alpha_0)/d_{z}})$ for splines with order $\alpha_0$. Assumption \ref{assfeigen} (i) and the first two conditions in (ii) are also standard in the literature. Given this assumption, it is without loss of generality to normalize $\mathbb{E}[\psi^{k_n}(Z_{i})\psi^{k_n}(Z_{i})^{\top}]=I_{k_n}$. The condition $k_{n}\xi_{0}(k_{n})^{2}\log p/n=O_{p}(1)$ in Assumption \ref{assfeigen}(ii) is new. It is a mild condition on the relationship between $p$ and $k_n$. Assumption \ref{assfeigen}(iii) is a smoothness condition on the approximation error of the projection $\Pi_{n,X|Z}$ to $\Pi_{X|Z}$.

\begin{assumption}\label{assum:splitting1} 
\begin{itemize}
	\item[(i)] $\sup_{w}|1/\hat{\pi}(w)-1/\pi(w)|=O_p(1)$;
	\item[(ii)] $\sup_{w}|\hat{\Phi}_1(w)-{\Phi}_1(w)|=O_p(1)$;
	\item[(iii)] $\sup_{w}|\hat{\Phi}_0(w)-{\Phi}_0(w)|=O_p(1)$;
	\item[(iv)]$\mathbb{E}_n((1/\hat{\pi}(W_i)-1/\pi(W_i))^2 \cdot (\hat{\Phi}_1(W_i)-{\Phi}_1(W_i))^2) = O_p(\log p/n \vee k_n\log k_n/n)$;
	\item[(v)] $\mathbb{E}_n((1/\hat{\pi}(W_i)-1/\pi(W_i))^2 \cdot (\hat{\Phi}_0(W_i)-{\Phi}_0(W_i))^2) = O_p(\log p/n \vee k_n\log k_n/n)$.
\end{itemize}
 \end{assumption}

Assumption \ref{assum:splitting1} imposes moderate conditions on the first stage approximations of the nuisance functions ${\pi}(W_i)$, ${\Phi}_1(W_i)$ and ${\Phi}_0(W_i)$. Only the interaction terms between ${\pi}(W_i)$ and ${\Phi}_1(W_i)$ or ${\pi}(W_i)$ and ${\Phi}_0(W_i)$ are required to converge at a mild rate $O_p(\log p/n \vee k_n\log k_n/n)$. This demonstrates the double robustness properties of our estimator such that when either ${\pi}(W_i)$ or $({\Phi}_1(W_i), {\pi}(W_i))$ are correctly specified, the desired rate of convergence in Theorem \ref{thm1} can be achieved. As pointed out in \cite{Chernozhukov2016}, the benefit of using sample-splitting is that it makes the entropy condition become very weak, allowing machine learning methods (e.g. random forest, boosted trees, deep neural nets, and their aggregated and hybrid versions) to be applied to estimate the functions $\hat{\pi}(W_i)$, $\hat{\Phi}_1(W_i)$ and $\hat{\Phi}_0(W_i)$. One can provide more primitive conditions to verify these rates for each given machine learning method of chosen.

\begin{assumption}\label{assum:chooselambda}
We choose $\lambda$, $k_n$, and $R$ satisfying the following: (i) $\lambda\gtrsim \sqrt{\log p/n}$; 
(ii) $2\lambda^{2}s_{0}/\Lambda^{2}_{\bar{X}}(s_0)\lesssim R^{2}\lesssim \lambda$ ; 
and (iii) $R^2 = \min\left(\ell_{k_{n}}^{2}c_{k_{n}}^{2}k_{n}/n, \xi_{0}^{2}(k_{n})c^{2}_{k_n}/n\right) +k_{n}/n$. 
\end{assumption}

\begin{thm}\label{thm1} Suppose that Assumptions \ref{Ass:parallel}-\ref{assum:chooselambda} hold. Then with probability approaching 1,
\[\|\hat{\beta}-\beta_0\|_1=O_p(\lambda s_0) \qquad \text{and} \qquad \|\hat{f}-f_0\|_{P,2}=O_p(R).\]
\end{thm}
Theorem \ref{thm1} establishes the rate of convergence for our estimator. We show that for the parametric estimator, similar to the one for the high-dimensional linear regressors (e.g., Theorem 6.1 in \cite{buhlmann2011statistics}), its convergence rate depends on the rate of the tuning parameter $\lambda$ and the level of sparsity $s_0$. 
When $\lambda=O(\sqrt{\log p/n})$, we have $\|\hat{\beta}-\beta_0\|_{1}=O_p(s_0\sqrt{\log p/n})$, which is the same rate in lasso regression for high-dimensional linear models without unknown functions. For the nonparametric estimator, the convergence rate maintains the same rate as the one obtained in nonparametric regressor models (e.g., Theorem 4.1 in \cite{BCCK2015}), which depends on the order of basis function $k_{n}$ and the approximation error. Unlike the results in the literature of semiparametric partially linear model when the dimension of $X$ is much smaller than sample size $n$, the convergence rate of the parametric estimator is slower than $\sqrt{n}$ due to high dimensionality of the model. It makes the inference problem challenging. As we will show in Section \ref{sec: Inference}, the asymptotic variance of the nonparametric estimator will contain a projection term that reflects the effect of the high-dimensional parametric estimation.

\section{Asymptotic Inference} \label{sec: Inference}
In many applications, practitioners are not only interested in the estimation of the treatment effect but also the uncertainty quantification of the estimated treatment effect. The latter provides the confidence of the treatment effect estimation and is a routine procedure in most causal inference problems. While the inferential properties under high-dimensional linear/generalized  linear models have been extensively investigated in the recent literature  \citep{zhang2011confidence,javanmard2013confidence,van2013asymptotically,belloni2013honest,ning2014general,ning2017likelihood,cai2015confidence,neykov2018unified,gold2020inference}, the asymptotic inference under the DiD design has not been studied, especially in the partially linear model specification. In this section, we consider how to construct confidence intervals for the parametric component $\beta_0$ and the nonparametric component $f_0(z)$ for given $z\in \mathcal{Z}$. 

Consider the inference problem for a linear combination of $\beta_0$, say $\xi^{\top}\beta_0$, for a known vector $\xi\in \mathbb{R}^p$. For instance, if we take $\xi$ as the unit basis vector $e_j=(0,...,0,1,0,...0)$ with the $j$th position being 1 and 0 otherwise, then the linear functional reduces to $\xi^{\top}\beta_0=(\beta_0)_j$, which is the $j$th component of the regression coefficient. Similarly, if we are interested in the prediction for a given test sample $X=x_0$ and $Z=z_0$, then the parameter of interest becomes $x_0^{\top}\beta+f(z_0)$. Thus, the inference problems can be decomposed into two problems: the inference on $x_0^{\top}\beta$ and the inference on $f(z_0)$. The former is again a linear combination of $\beta_0$ with $\xi=x_0$. The inference on $f(z_0)$ will be studied later in this section. 
To construct the confidence intervals for $\xi^{\top}\beta_0$, we extend the de-biasing approach to the DiD design under the partially linear model specification. Given the Lasso estimator $\hat\beta$, we propose the following de-biased Lasso estimator:
\begin{equation}\label{eqdebias1}
\hat T=\xi^{\top}\hat\beta-\hat w^{\top}  \mathbb{E}_n\left\{\left(\hat{\rho}_i \Big(\Delta Y_i -(1-\hat{\pi}_i)\hat{\Phi}_1(W_i) - \hat{\pi}_i\hat{\Phi}_0(W_i) \Big)-X_i^\top\hat\beta - \hat f(Z_i) \right)\tilde X_i\right\},
\end{equation}
where   
\begin{equation}\label{eqhatw}
\hat w=\arg\min \|w\|_1~~\textrm{s.t.}~~\|\xi+ \hat{\Sigma}_{\tilde{X}} w\|_\infty\leq \lambda',
\end{equation} with $\hat{\Sigma}_{\tilde{X}}=\frac{1}{n}\sum_{i=1}^n \tilde X_i \tilde X_i^{\top}$ and $\lambda'$ as a tuning parameter. 
We will show that $\hat w$ is a consistent estimator of $w_{0}=\Sigma^{-1}_{\bar{X}}\xi$. Let $s_{w} = |\{k: w_{0k} \neq 0\}|$ as the size of non-zero elements in $w_{0}$.

\begin{assumption}\label{assum:splitting} 
\begin{itemize}
\item [(i)]$\mathbb{E}_n((1/\hat{\pi}(W_i)-1/\pi(W_i))^2)=o_p(1/{(k_n \log k_n)} \vee 1/{(s_w\log p)})$ and $\sup_{w}|1/\hat{\pi}(w)-1/\pi(w)|=o_p(1)$;
\item [(ii)]$\mathbb{E}_n(\hat{\Phi}_1(W_i)-{\Phi}_1(W_i))^2=o_p(1/{(k_n \log k_n)} \vee 1/{(s_w\log p)} )$ and  $\sup_{w}|\hat{\Phi}_1(w)-{\Phi}_1(w)|=o_p(1)$;
\item[(iii)]$\mathbb{E}_n(\hat{\Phi}_0(W_i)-{\Phi}_0(W_i))^2=o_p(1/{(k_n \log k_n)} \vee 1/{(s_w\log p)})$ and  $\sup_{w}|\hat{\Phi}_0(w)-{\Phi}_0(w)|=o_p(1)$;
\item [(iv)]$\E_n((1/\hat{\pi}(W_i)-1/\pi(W_i))^2 \cdot (\hat{\Phi}_1(W_i)-{\Phi}_1(W_i))^2) = o_p(1/(s_w^2n) \vee 1/(k_nn))$;
\item[(v)] $\E_n((1/\hat{\pi}(W_i)-1/\pi(W_i)) \cdot (\hat{\Phi}_0(W_i)-{\Phi}_0(W_i))) = o_p(1/(s_w^2n) \vee 1/(k_nn))$.
\end{itemize}
\end{assumption}
Assumption \ref{assum:splitting} is a stronger version of Assumption \ref{assum:splitting1}, which is required for constructing the asymptotic normality. 

Let $\sigma^{2}_i=\mathbb{E}[\epsilon^{2}_{i}|X_i]$, $V_\beta=\Sigma^{-1}_{\bar{X}}\Omega_\beta\Sigma^{-1}_{\bar{X}}$ with $\Sigma_{\bar{X}}=\mathbb{E}\left[\bar{X}_{i}\bar{X}_{i}^{\top}\right]$ and $\Omega_{\beta}:=\mathbb{E}\left[\sigma_{i}^{2}\bar{X}_{i}\bar{X}_{i}^{\top}\right]$. Let $\hat{V}_\beta=\hat{w}^{\top}\hat{\Omega}_{\beta}\hat{w}$ with $\hat{\Omega}_{\beta}:=\mathbb{E}_{n}\left[\hat{\sigma}_{i}^{2}\tilde{X}_{i}\tilde{X}_{i}^{\top}\right]$.

\begin{assumption} \label{assum: rate_thminf1} We have (i) $n^{-1/2}(s_{w}^{2}(\log p)^{1/2}\vee s_{w}\log p)=o_{p}(1)$ and $s_w \max_{1\leq j\leq p,1\leq i \leq n}|\tilde{X}_{ij}r_{ni}|=o(n^{-1/2})$; (ii) $s_{w}\mathbb{E}_{n}\left[\epsilon_{i}(\tilde{X}_{i}-\bar{X}_{i})\right]=o_{p}(n^{-1/2})$; (iii) ; (iv) the smallest eigenvalue of $\Omega_{\beta}$ denoted as $\lambda_{\min}(\Omega_{\beta})$ is bounded away from 0 and the biggest eigenvalue denoted as $\lambda_{\max}(\Omega
_{\beta})$ is bounded from above. \end{assumption}
\begin{thm}\label{thminf1}
Suppose that Assumptions 1-7, 9, 10 and 11 hold. let $\lambda' \gtrsim \sqrt{\log p/n}$, we have that
$$
\sqrt{n}(\hat T-\xi^{\top}\beta_0)\rightarrow_d N(0, \xi^{\top}V_\beta \xi).
$$
Furthermore, if $\left(\log(np)\left( \frac{s_{0} \log p}{n}+\sqrt{\frac{\xi_{0}^{2}(k_{n}) k_{n}}{n}}+\ell_{k_{n}}c_{k_{n}}\right) \right)=o(1)$, $\hat{V}_{\beta}\overset{p}{\rightarrow}V_{\beta}$.
\end{thm}

Theorem \ref{thminf1} implies that we can construct an asymptotic $(1-\alpha)$ confidence interval for $\xi^{\top}\beta_0$ as $(\hat T-z_{1-\alpha/2}(\hat w^{\top} \hat{V}_\beta\hat w)^{1/2}, \hat T+z_{1-\alpha/2}(\hat w^{\top} \hat{V}_\beta\hat w)^{1/2}$, where $z_{1-\alpha/2}$ is the $1-\alpha/2$ quantile of a standard normal random variable. Note that by constructing the Neyman orthogonality condition and by using an de-biased estimator, the asymptotic variance of the parametric parameter coincides the one in low-dimensional partially linear models \citep{Robinson1988root}. In particular, with homoskedasticity $\mathbb{E}[\epsilon_{i}^{2}|X_{i}]=\sigma^{2}$, the asymptotic variance achieves the semiparametry efficiency bound $V_{\beta}=\sigma^{2}\mathbb{E}[\tilde{X}_{i}\tilde{X}_{i}^{\top}]^{-1}$. \\

In the following, we extend the de-biasing approach to construct the confidence intervals for $f(z)$ for any given $z\in \mathcal{Z}$, where we assume $d_z$ is much smaller than $n$ to avoid the curse of dimentionality problem for nonparametric estimation. Recall that $f(z)$ can be approximated in the sieve space by $\psi^{k_n}(z)^{\top}\gamma_{n0}$. To construct the confidence interval for $f(z)$, it suffices to apply the debias approach to the parameter $\gamma_{n}$, in which the parameter $\beta$ is treated as a high-dimensional nuisance parameter. To this end, we first derive the score function for $\gamma_n$ as
$$
U_{\gamma_{n}}(\beta,M)=\mathbb{E}_n\left\{\left(\hat{\rho}_i \Big(\Delta Y_i -(1-\hat{\pi}_i)\hat{\Phi}_1(W_i) - \hat{\pi}_i\hat{\Phi}_0(W_i) \Big)-X_i^\top\beta -  \psi^{k_n}(Z_i)^{\top}\gamma_n \right)(\psi^{k_n}(Z_i)- MX_i)\right\},
$$
where $M=\mathbb{E}(\psi^{k_n}(Z_i)X_i^{\top})\{\mathbb{E}X_i^{\otimes 2}\}^{-1}$ is a $k_n\times p$ matrix. One key property of the score function is that $U_{\gamma_{n}}(\beta,M)$ is insensitive to the unknown high-dimensional nuisance parameters $\beta$ and $M$. In fact, we will show below that $U_{\gamma_{n}}(\hat\beta,\hat M)=U_{\gamma_{n}}(\beta,M)+o_p(n^{-1/2})$ for some suitable estimators $\hat\beta$ and $\hat M$ to be defined later. 

Given this score function $U_{\gamma_{n}}(\beta,M)$, we can  define the one-step updated de-biased estimator as $\bar{f}(z):=\psi^{k_{n}}(z)^{\top}\bar{\gamma}_{n}$, where
$$
\bar{\gamma}_{n}:=\hat{\gamma}_{n}-\hat{\Sigma}^{-1}_{f}\mathbb{E}_n\left\{\left(\hat{\rho}_i \Big(\Delta Y_i -(1-\hat{\pi}_i)\hat{\Phi}_1(W_i) - \hat{\pi}_i\hat{\Phi}_0(W_i) \Big)-X_i^\top\hat\beta - \hat f(Z_i) \right)(\psi^{k_n}(Z_i)-\hat MX_i)\right\}
$$
where $\hat{\Sigma}_f = \mathbb{E}_n\left\{(\psi^{k_n}(Z_i)-\hat MX_i)\psi^{k_n}(Z_i)^{\top}\right\}$ and $\hat M=[\hat M_1,...,\hat M_j,...,\hat M_{k_n}]^{\top}$ with
\begin{equation}\label{eqhatm}
\hat M_j=\arg\min \|m\|_1~~\textrm{s.t.}~~\|m^{\top}\mathbb{E}_n(X_i^{\otimes 2})-\mathbb{E}_n(\psi_j^{k_n}(Z_i)X_i^{\top})\|_\infty\leq \lambda'',
\end{equation}

Let $s_{m}=\left|k:\left\{ \left(\mathbb{E}[\psi^{k_{n}}(Z_{i})X_{i}^{\top}]\right)\left(\mathbb{E}\left[X_{i}X_{i}^{\top}\right]^{-1}\right)\right\} _{k}\neq0\right| $ as the size of non-zeros elements in $\left(\mathbb{E}[\psi^{k_{n}}(Z_{i})X_{i}^{\top}]\right)\left(\mathbb{E}\left[X_{i}X_{i}^{\top}\right]^{-1}\right)$. 

Let $\sigma_{z}^{2}=\psi^{k_{n}}(z)^{\top}V_{f}\psi^{k_{n}}(z)$ with $V_{f}=\Sigma_{f}^{-1}\Omega_{f}\Sigma_{f}^{-1}$, $\Sigma_{f}=\mathbb{E}\left[\left(\psi^{k_{n}}(Z_{i})-MX_{i}\right)\psi^{k_{n}}(Z_{i})^{\top}\right]$ and $\Omega_{f}=\mathbb{E}\left[\sigma_{i}^{2}\psi^{k_{n}}(Z_{i})\psi^{k_{n}}(Z_{i})^{\top}\right]-M\mathbb{E}\left[\sigma_{i}^{2}X_{i}X_{i}^{\top}\right]M^{\top}$. 

We define the sample analogs similarly. Let $\hat{\sigma}_{z}^{2}=\psi^{k_{n}}(z)^{\top}\hat{V}_{f}\psi^{k_{n}}(z)$ with $\hat{V}_{f}=\hat{\Sigma}_{f}^{-1}\hat{\Omega}_{f}\hat{\Sigma}_{f}^{-1}$, $\hat{\Sigma}_{f}=\mathbb{E}\left[\left(\psi^{k_{n}}(Z_{i})-\hat{M}X_{i}\right)\psi^{k_{n}}(Z_{i})^{\top}\right]$ and $\hat{\Omega}_{f}= \mathbb{E}_n\left[\hat{\sigma}^{2}_i\psi^{k_n}(Z_i)\psi^{k_n}(Z_i)^{\top}\right] - \hat M\mathbb{E}_n\left[\hat{\sigma}^{2}_iX_iX_i^{\top}\right]\hat M^{\top}$.

\begin{assumption} \label{assum: rate_thminf2}We have (i)$s_{m}^{2}\sqrt{k_{n}\log p/n}=o(1)$, $\sqrt{n}\sigma^{-1}_{z}\|\mathbb{E}_{n}\left[r_{ni}(\psi^{k_{n}}(Z_{i})-MX_{i})\right]=o(1)$ and $s_{m}s_{0}\log p/\sqrt{n}=o(1)$; 
\[\sqrt{n}\sigma_{z}^{-1}\psi^{k_{n}}(z)^{\top}\Sigma_{f}^{-1}\mathbb{E}_{n}\left\{r_{ni}\left(\psi^{k_{n}}(Z_{i})-MX_{i}\right)\right\} = o_p(1)\]
(ii) the smallest eigenvalues of $\Sigma_{f}$ and $\Omega_{f}$ denoted as $\lambda_{\min}(\Sigma_{f})$ and $\lambda_{\min}(\Omega_{f})$, respectively, are bounded away from 0 and the biggest eigenvalues denoted as $\lambda_{\max}(\Sigma_{f})$ and $\lambda_{\max}(\Omega_{f})$, respectively, are bounded from above. \end{assumption}

\begin{thm}\label{thminf2} Suppose that Assumptions 1-7, 9-12 hold and $\sqrt{n}\sigma_{z}^{-1}\max_{1\leq i\leq n}|r_{ni}|=o(1)$. Let $\lambda'' \gtrsim \sqrt{\log p/n}$, we have that
\[
\sqrt{n}\sigma_{z}^{-1/2}(\bar f(z)-f_{0}(z))\rightarrow_d N(0, 1). 
\]
Furthermore, if $\left(\sqrt{\log np}\left(n^{1/r_{\epsilon}}+c_{k_{n}}\ell_{k_{n}}\right)\left(s_{0}\sqrt{\log p/n}+\sqrt{\xi_{0}^{2}(k_{n})k_{n}/n}+c_{k_{n}}\ell_{k_{n}}\right)\right)=o(1)$, we have $\hat{\sigma}^{2}_{z}\overset{p}\rightarrow \sigma^{2}_{z}$. 
\end{thm}
Theorem \ref{thminf2} provides asymptotic theory that can be used to construct the confidence intervals for $f(z)$ for any $z$. Unlike the standard results in the nonparametric literature, we need to construct a de-biased $\bar{f}(z)$ estimator that corrects bias caused by estimating the high-dimensional parametric component of the partially linear model. Since the Lasso estimator of the parametric linear part has a convergence rate that is slower than $\sqrt{n}$, the asymptotic variance of $\bar{f}(z)$ contains a projection term that reflects the effect of the parametric component on the nonparametric component. We can construct an asymptotic $(1-\alpha)$ confidence interval for $f_{0}(z)$ as $\left[ \bar{f}(z)-z_{1-\alpha/2}\hat{\sigma}_{z}, \bar{f}(z)+z_{1-\alpha/2}\hat{\sigma}_{z}\right]$, where $z_{1-\alpha/2}$ is the $1-\alpha/2$ quantile of a standard normal random variable. Combining the results in Theorems \ref{thminf1} and \ref{thminf2}, one can easily construct the confidence interval for the heterogeneous ATT denoted as $\tau_0(w)$. 


\section{Simulation}
We compare the finite sample performance of the doubly robust estimator proposed in this paper with the semiparametric DiD estimator in \cite{Abadie2005} when the latter is applicable. We consider two data generating processes. In the first setting (DGP1), we allow $Y^0(i, 0)$ and $Y^1(i,0)$ to follow standard normal distribution, where $Y^1(i, 1)$ and $Y^0(i, 1)$ are defined as follows:
\begin{align*}
	Y^1(i, 1) &= Y^1(i, 0) + X_i^\top\beta^1 +f(Z_i) +\epsilon_{i1},\\
	Y^0(i, 1) &= Y^0(i, 0) + X_i^\top\beta^0  +\epsilon_{i0},
\end{align*}
where $X_i$ and $Z_i$ are generated from independently standard normal distributions. The errors $\epsilon_{i1}$ and $\epsilon_{i0}$ are independently generated from standard normal distributions. In the second setting (DGP2), we define 
\[Y(i, 0) = \tilde\epsilon_{i}\cdot (1/\sqrt{2} \cdot Z_i + 1/\sqrt{2} \cdot X_{i1})\]
 where $\tilde\epsilon_{i}$ is generated from a standard normal distribution and $X_i \sim N(0, \Sigma)$ where $\Sigma_{jk} = \rho^{|j-k|}$. This allows both heteroskedasticity in the error term as well as correlation among regressors. In both DGP1 and DGP2, we set $\beta^1_i = 2/i$ and $\beta^0_i = 1/i$ for $i \leq 15$ and $f(Z_i) = \exp(Z_i)$. The treatment assignment probability is based on a logistic distribution with 
\[\mathbb{P}(T_i=1) = 1-(1+\exp(X_i^\top\theta_0))^{-1},\]
where $\theta_{0i} = 1/i$ for $i \leq 10$. In both setting, we use 8th degree trigonometric polynomial basis for the non-parametric estimation. 

Table \ref{tab2000} and \ref{tab3000} summarize the results for the two settings. We report the average bias, average standard errors, average mean squared errors, average coverages for a $90\%$ confidence intervals as well as the average lengths for this confidence intervals separately for both the linear coefficients and the nonparametric coefficients. To compare with the parametric part, we report the coverages for the linear combination of the nonparametric coefficients. Divided by the standard error, it also converges to standard normal with the same condition in \ref{thminf2}. The ``Dr-DiD" columns represent the results for the doubly robust diff-in-diff estimator and ``semi-DiD" columns represent the results for the \cite{Abadie2005} estimator. We present results with $n$ varying from 200, 500 and 1000 and the dimension for linear specification $p$ varying from 10, 50 , 500 and 1000. Notice that the \cite{Abadie2005} estimator is infeasible when $n\leq p$ so we omit to report the ``semi-DiD" results when $p =500, 1000$ and denote them as ``-" in the tables. Furthermore, the variance for the Semi-Did estimator becomes large when $n$ is relatively large compared to $p$ (e.g. $n=200$ and $p=50$). Although Semi-DiD estimator can still be computed when $n=1000$ and $p=500$, we choose not to report this result because of its large variance . 

As shown in both tables, the Dr-DiD estimator has a smaller standard error, RMSE and confidence interval length in both linear and nonparametric specifications. When $p=50$, the Semi-DiD estimator becomes too conservative and produces larger standard errors. On the other hand, the Dr-DiD estimator is also more robust comparing with the Semi-DiD estimator when switching from homoskedastic errors to heteroskedastic errors. More importantly, our experiments show that in finite sample, the Dr-DiD estimator can deliver reasonable estimates under high-dimensional settings.

\begin{table}[t!]
\caption{Simulation: homogeneous error with independent confounders}\label{tab2000}
	\begin{center}
		\scalebox{0.90}{
			\begin{tabular}{llccccccccc}
				\toprule 
					
				& p  &&  \multicolumn{2}{c}{10} &  \multicolumn{2}{c}{50}    &  \multicolumn{2}{c}{500}  & \multicolumn{2}{c}{1000} \\
				\cmidrule(lr){4-5} \cmidrule(lr){6-7} \cmidrule(lr){8-9} \cmidrule(r){10-11}  
				&&&  Dr-DiD & Semi-DiD & Dr-DiD & Semi-DiD &Dr-DiD & Semi-DiD &Dr-DiD & Semi-DiD  \\
				\midrule     
\parbox[t]{2mm}{\multirow{12}{*}{\rotatebox[origin=c]{90}{$n=200$}}} & \multicolumn{5}{l}{linear} \\                              
&Bias &&-0.0476&0.0746&-0.0279&-0.1177&0.0008&-&-0.0045&-\\  
&Std Err &&0.5241&1.2641&0.4533&4.3207&0.3732&-&0.4030&-\\   
&RMSE &&0.2816&1.7240&0.2116&21.9360&0.1410&-&0.1677&-\\     
&Coverage &&0.8300&0.9260&0.8680&0.9710&0.8833&-&0.8900&-\\                        
&CI length &&1.4018&4.1499&1.2448&49.9842&1.0778&-&1.0800&-\\
& \multicolumn{5}{l}{nonparametric} \\ 
&Bias &&-0.0165&0.0128&0.0365&-0.1181&0.0260&-&0.0102&-\\    
&Std Err &&1.0612&1.3773&0.8825&8.7693&0.8072&-&0.9131&-\\   
&RMSE &&1.1701&1.9408&0.7846&80.8639&0.6548&-&0.8426&-\\     
&Coverage &&0.8200&0.9056&0.8525&0.9712&0.7988&-&0.8375&-\\  
&CI length &&2.4596&5.2435&2.3056&49.2043&2.1269&-&2.1916&-\\
\hline                                                       
\parbox[t]{2mm}{\multirow{12}{*}{\rotatebox[origin=c]{90}{$n=500$}}} & \multicolumn{5}{l}{linear} \\                              
&Bias &&-0.0258&0.0031&-0.0171&-0.0011&-0.0016&-&-0.0011&-\\ 
&Std Err &&0.2812&0.4972&0.2268&0.5771&0.2086&-&0.1983&-\\   
&RMSE &&0.0828&0.2574&0.0536&0.3369&0.0441&-&0.0398&-\\      
&Coverage &&0.8570&0.8815&0.8656&0.9635&0.8857&-&0.8916&-\\                        
&CI length &&0.7628&1.4949&0.6723&3.9223&0.6488&-&0.6261&-\\ 
& \multicolumn{5}{l}{nonparametric} \\ 
&Bias &&0.0023&-0.0133&0.0028&0.0108&0.0031&-&0.0032&-\\     
&Std Err &&0.3953&0.6337&0.3872&0.7954&0.3715&-&0.3558&-\\   
&RMSE &&0.1564&0.4020&0.1501&0.6348&0.1384&-&0.1277&-\\      
&Coverage &&0.8888&0.8644&0.8688&0.9606&0.8588&-&0.8725&-\\  
&CI length &&1.1900&1.9907&1.1284&5.6116&1.1040&-&1.1037&-\\ 
\hline                                                       
\parbox[t]{2mm}{\multirow{12}{*}{\rotatebox[origin=c]{90}{$n=1000$}}} & \multicolumn{5}{l}{linear} \\                              
&Bias &&-0.0091&0.0044&-0.0068&0.0094&-0.0014&-&-0.0008&-\\  
&Std Err &&0.1842&0.3126&0.1485&0.3347&0.1401&-&0.1399&-\\   
&RMSE &&0.0342&0.0992&0.0226&0.1132&0.0199&-&0.0198&-\\      
&Coverage &&0.8640&0.8720&0.8848&0.9027&0.8910&-&0.8946&-\\                      
&CI length &&0.5432&0.9229&0.4683&1.2354&0.4466&-&0.4482&-\\ 
& \multicolumn{5}{l}{nonparametric} \\   
&Bias &&0.0047&-0.0251&-0.0187&-0.0121&-0.0069&-&-0.0083&-\\ 
&Std Err &&0.2695&0.3950&0.2362&0.4562&0.2268&-&0.2446&-\\   
&RMSE &&0.0728&0.1570&0.0566&0.2087&0.0515&-&0.0599&-\\      
&Coverage &&0.8750&0.8725&0.8812&0.9075&0.8825&-&0.8712&-\\  
&CI length &&0.8223&1.2382&0.7472&1.7548&0.7196&-&0.7247&-\\ 
				\bottomrule
\end{tabular}
}
\end{center}
\begin{flushleft}
	\footnotesize{ This table compare the doubly robust diff-in-diff estimator (denoted as Dr-DiD) with the original semi-parametric diff-in-diff estimator (denoted as Semi-DiD) proposed in \cite{Abadie2005}. $p$ represents the dimension for the linear specification. The nonparametric part is specified by an exponential function and it is approximated by a 8th degree trigonometric polynomial basis in both methods. The nominal coverage is at 90\%.}
\end{flushleft}

\end{table}

\begin{table}[t!]
\caption{Simulation: heterogeneous error with correlated confounders}\label{tab3000}
	\begin{center}
		\scalebox{0.90}{
			\begin{tabular}{llccccccccc}
				\toprule 
					
				& p  &&  \multicolumn{2}{c}{10} &  \multicolumn{2}{c}{50}    &  \multicolumn{2}{c}{500}  & \multicolumn{2}{c}{1000} \\
				\cmidrule(lr){4-5} \cmidrule(lr){6-7} \cmidrule(lr){8-9} \cmidrule(r){10-11}  
				&&&  Dr-DiD & Semi-DiD & Dr-DiD & Semi-DiD &Dr-DiD & Semi-DiD &Dr-DiD & Semi-DiD  \\
				\midrule     
\parbox[t]{2mm}{\multirow{12}{*}{\rotatebox[origin=c]{90}{$n=200$}}} & \multicolumn{5}{l}{linear} \\                                                           
&Bias &&-0.0213&0.0159&0.0122&-0.1277&0.0048&-&0.0027&-\\     
&Std Err &&0.5929&2.2577&0.5333&11.3165&0.4583&-&0.4649&-\\   
&RMSE &&0.3643&5.7073&0.2912&156.7863&0.2129&-&0.2255&-\\     
&Coverage &&0.8390&0.9260&0.8608&0.9360&0.8881&-&0.8912&-\\                         
&CI length &&1.5717&8.8340&1.4061&12.8266&1.2764&-&1.2174&-\\ 
& \multicolumn{5}{l}{nonparametric} \\  
&Bias &&0.0327&0.0407&-0.0200&0.0026&-0.0391&-&0.0183&-\\     
&Std Err &&0.9525&2.7935&1.0049&12.2120&0.9813&-&0.9562&-\\   
&RMSE &&0.9288&8.0217&1.0136&176.1709&0.9779&-&0.9171&-\\     
&Coverage &&0.8275&0.9269&0.8200&0.9450&0.8275&-&0.8088&-\\   
&CI length &&2.4823&10.4079&2.4977&12.4647&2.4471&-&2.3827&-\\
\hline                                                        
\parbox[t]{2mm}{\multirow{12}{*}{\rotatebox[origin=c]{90}{$n=500$}}} & \multicolumn{5}{l}{linear} \\                                                            
&Bias &&0.0148&-0.0486&0.0072&0.0198&0.0006&-&0.0001&-\\      
&Std Err &&0.3691&1.7750&0.2966&2.5093&0.2305&-&0.2285&-\\    
&RMSE &&0.1368&3.8087&0.0888&6.7667&0.0536&-&0.0526&-\\       
&Coverage &&0.8350&0.8550&0.8798&0.9761&0.8901&-&0.8943&-\\                           
&CI length &&0.9454&3.2110&0.8478&15.5631&0.7177&-&0.7059&-\\
& \multicolumn{5}{l}{nonparametric} \\ 
&Bias &&-0.0068&0.0255&-0.0061&0.1500&0.0127&-&-0.0156&-\\    
&Std Err &&0.4750&1.8659&0.4719&2.7917&0.4122&-&0.3913&-\\    
&RMSE &&0.2278&3.7041&0.2237&7.9921&0.1695&-&0.1541&-\\       
&Coverage &&0.8800&0.8569&0.8487&0.9681&0.8225&-&0.8562&-\\   
&CI length &&1.3348&3.5635&1.2741&18.3474&1.1226&-&1.1594&-\\ 
\hline                                                        
\parbox[t]{2mm}{\multirow{12}{*}{\rotatebox[origin=c]{90}{$n=1000$}}} & \multicolumn{5}{l}{linear} \\                                                            
&Bias &&-0.0018&-0.0000&0.0048&-0.0007&0.0004&-&0.0006&-\\    
&Std Err &&0.2733&0.6733&0.2071&1.1275&0.1850&-&0.1706&-\\    
&RMSE &&0.0750&0.4579&0.0433&1.3780&0.0345&-&0.0294&-\\       
&Coverage &&0.8410&0.8650&0.8720&0.9099&0.8943&-&0.8960&-\\                          
&CI length &&0.7548&1.9303&0.6045&3.2445&0.5709&-&0.5378&-\\  
& \multicolumn{5}{l}{nonparametric} \\ 
&Bias &&-0.0116&0.0129&0.0114&-0.1022&0.0023&-&0.0081&-\\     
&Std Err &&0.3303&0.7637&0.2771&1.3032&0.2771&-&0.2642&-\\    
&RMSE &&0.1092&0.5839&0.0781&1.8072&0.0765&-&0.0703&-\\       
&Coverage &&0.8600&0.8750&0.8762&0.8981&0.8588&-&0.8488&-\\   
&CI length &&0.9666&2.1754&0.8233&3.6871&0.8006&-&0.7682&-\\                                 
				\bottomrule
\end{tabular}
}
\end{center}
\begin{flushleft}
	\footnotesize{ This table compare the doubly robust diff-in-diff estimator (denoted as Dr-DiD) with the original semi-parametric diff-in-diff estimator (denoted as Semi-did) proposed in \cite{Abadie2005}. $p$ represents the dimension for the linear specification. The nonparametric part is specified by an exponential function and it is approximated by a 8th degree trigonometric polynomial basis in both methods. The nominal coverage is at 90\%.}
\end{flushleft}

\end{table}  

\section{Empirical Application}
We use the proposed method to study the effect of increasing the minimum wage on unemployment rates at the county level. We use the same dataset collected by \cite{CallawayLi2020}, which contains the county level unemployment rates from 2005 to 2007 before the Fair Minimum Wage Act was enacted in all states on May 25, 2007. Eleven states increased their minimum wage by the first quarter of 2007, while the other states did not increase their minimum wage until the federal minimum wage increased in July of 2007.\footnote{New Hampshire and Pennsylvania are dropped for the same reason  as in \cite{CallawayLi2020} } 

We explore the variation in adopting the minimum wage policy among different states to evaluate its impact on county level unemployment rates. \cite{CallawayLi2020} consider identification and estimation of the quantile treatment effect on the treated under a distributional extension of the mean difference in differences assumption with fixed-dimensional covariates. Differing from the work in \cite{CallawayLi2020}, this work focuses on studying the impact of covariates on heterogeneous ATT in this DiD design. Our proposed method allows us to weaken the parallel trend assumption to the conditional parallel trend assumption by conditioning on a large amount of potential confounders. For example, states with smaller populations may have higher variation in the unemployment rates, thus moving at different trends as compared to states with larger populations. Our method also allows us to derive marginal effect given a specific covariate of interest. As a result, customized policy recommendations can be designed based on those results. 

Figure \ref{meandiff} plots the simple difference for the 2005 to 2007 difference in unemployment rate by median income (panel a) and population (panel b). We separate the counties in the control and treated states by red and blue color. The solid lines on the graphs are the local means for the control and treated groups. There is a general decrease in the unemployment rate from 2005 to 2007 across all counties as the change in unemployment rate is centered below 0. The difference between the red and blue lines is the standard DiD estimator under the unconditional parallel trend assumption. The decrease in the unemployment rate for the treated counties is lower than the control counties at the low income region, while not much difference between treated and control is observed at high income region. On the other hand, the decrease in the unemployment rate for the treated counties is consistently lower than the control counties regardless of population size.\\
Figure \ref{703diff} compares the semi-parametric diff-in-diff estimator (Semi-DiD: blue) with the doubly robust estimator (Dr-DiD: red) on the nonparametric component $f$. Both methods use 4th degree trigonometric polynomial basis to approximate $f$ and have partially linear forms. The Semi-DiD estimator controls only the underlying covariate of interest ($Z_i$), while the Dr-DiD estimator controls not only the underlying covariate of interest ($Z_i$) but also 703 covariates ($X_i$) in a linear additive form. These covariates include 38 county level characteristics, as well as all interactions between them. It is also worth pointing out that when computing the marginal effect for median income, population is used as a confounder in the linear part and vice versa when computing the marginal effect for population. The dashed lines are 95\% confidence intervals for the estimator.

We find that both estimators show that regions with high median income levels and larger population sizes may suffer from an increase in unemployment rate due to the minimum wage policy, while no significant effect of the policy is detected for regions with lower median income levels and smaller population sizes. These findings coincide with the canonical economic theory on unemployment rate. For example, a higher median income level implies a higher substitution cost for a worker currently at minimum wage and thus leads to an increase in unemployment rate when the minimum wage rises. On the other hand, a region with a larger population size means more labor supply and thus a raise in minimum wage can lead to a surplus. The difference in the general direction of the results predicted by Figure \ref{meandiff} and Figure \ref{703diff} indicates the potential severeness of confounding problems in this design.

Next, while both the Semi-DiD estimator and Dr-DiD estimator show no significant impact of the policy at low median income and thin population regions, the Dr-DiD shows a larger impact in regions with a higher median income level and a denser population than the Semi-DiD estimator. The Dr-DiD estimated effects are also more significant in those regions. This is due to the controlling of additional covariates that further alleviates the concern of confoundedness as well as reducing the uncertainly in the model to yield more accurate estimation.

\begin{figure}[t]
\begin{center}
\subfloat[Median Income]{\includegraphics[scale=0.5]{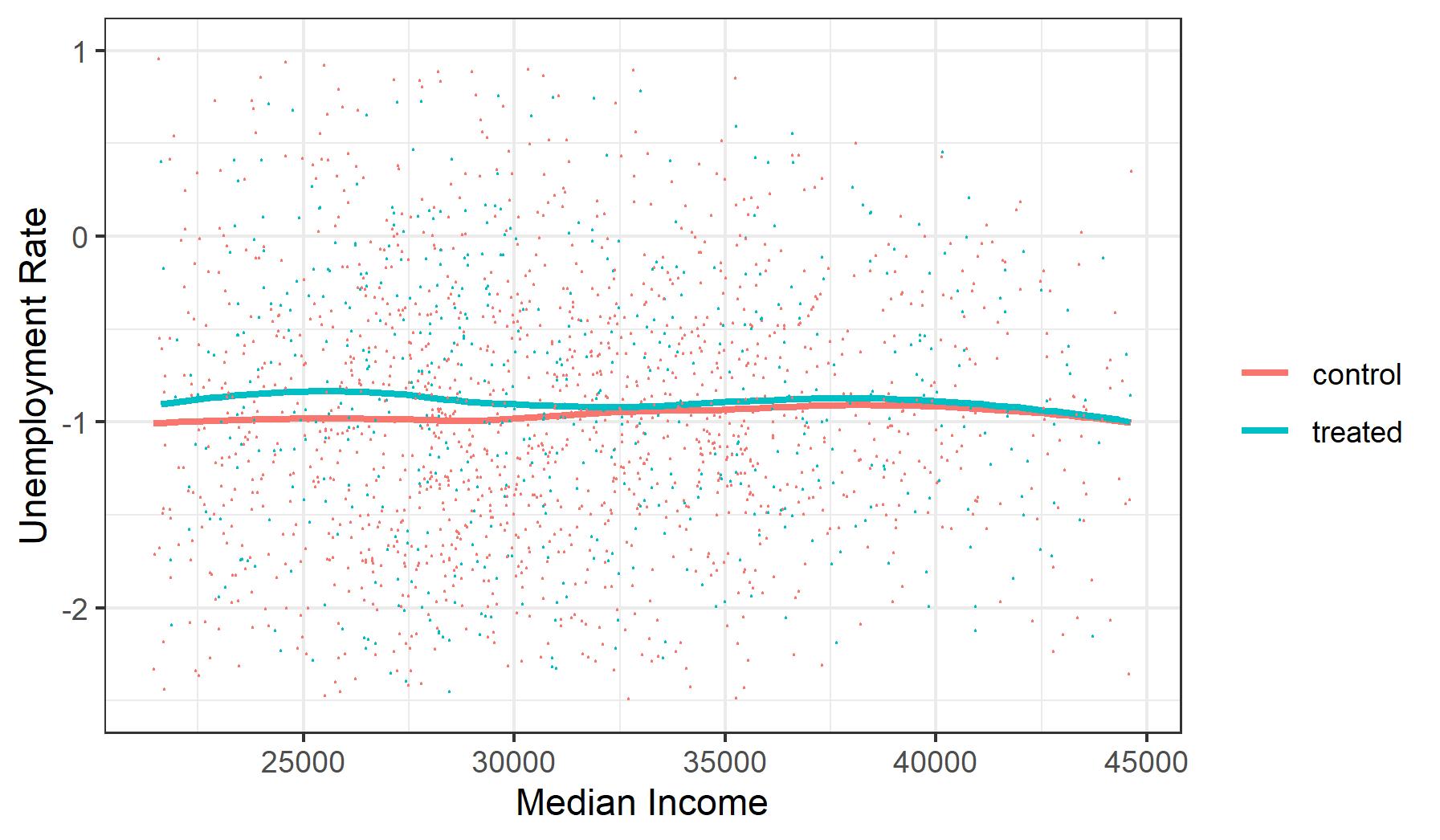}	}
\subfloat[Population]{\includegraphics[scale=0.5]{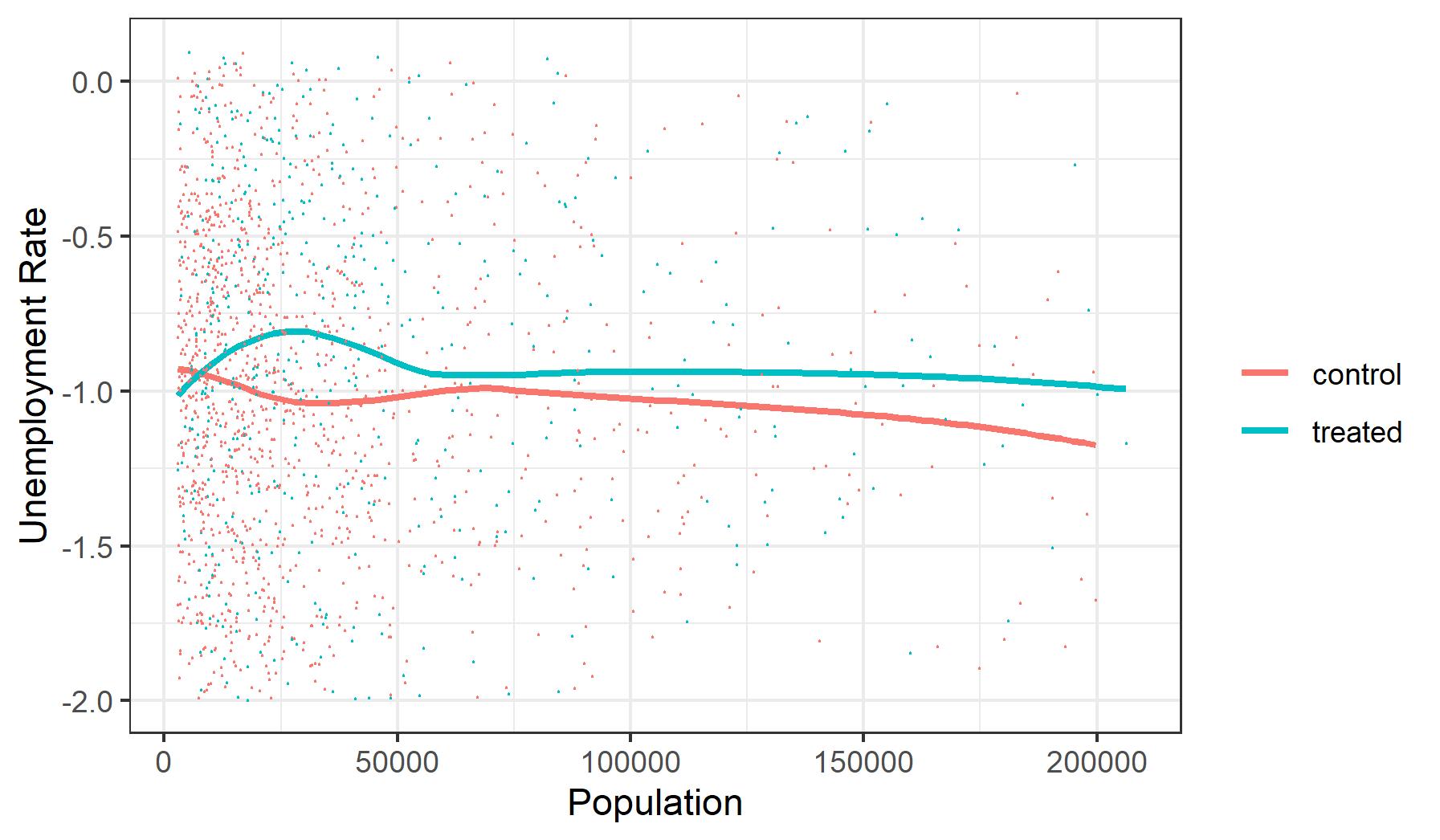}	}
\end{center}	
\caption{Difference in 2005 to 2007 Unemployment Rate \label{meandiff}}
\end{figure}


\begin{figure}[t]
\begin{center}
\subfloat[Median Income]{\includegraphics[trim={0 0 0 0},clip, scale=0.63]{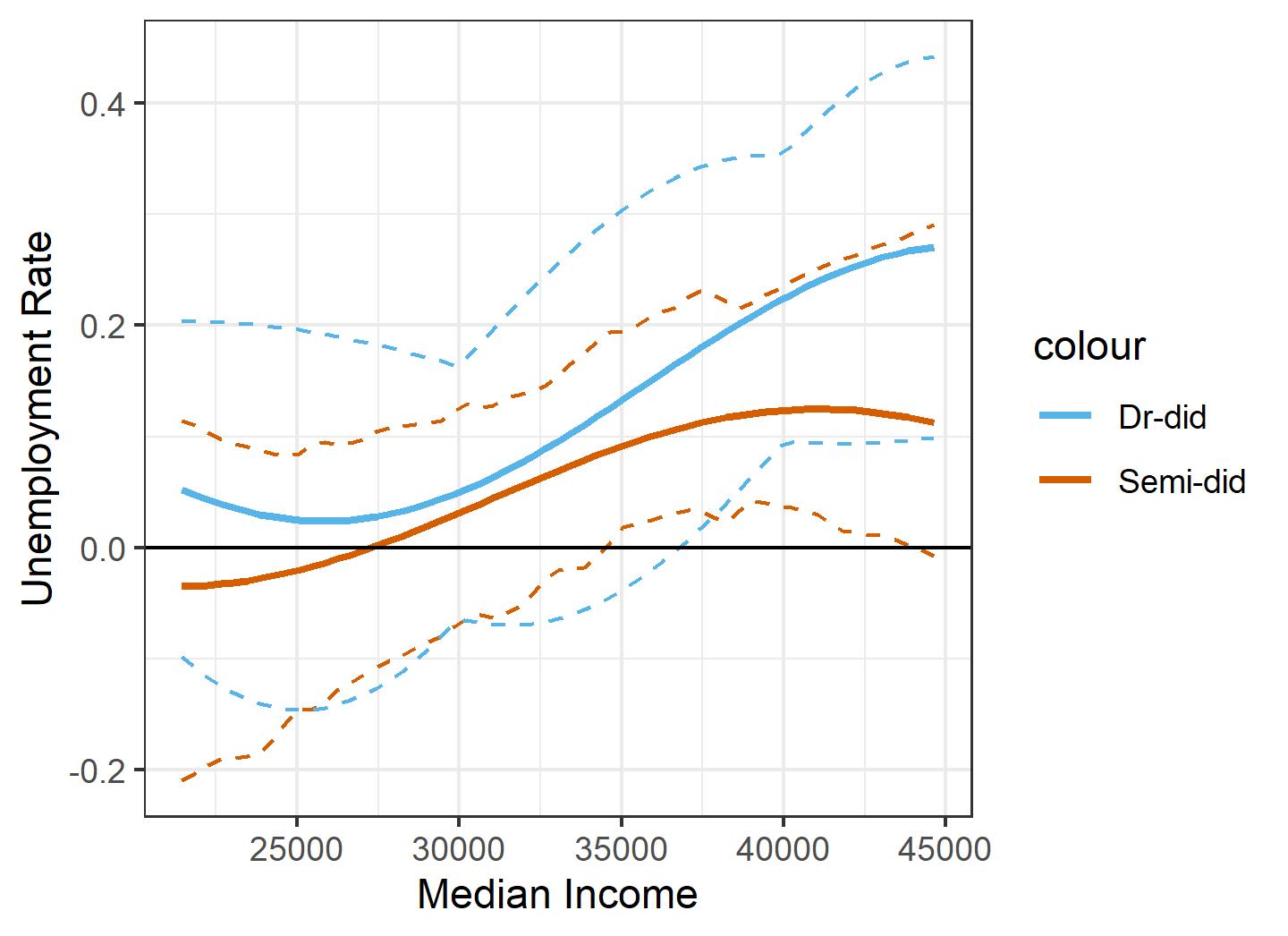}	}
\subfloat[Population]{\includegraphics[trim={0 0 0 0},clip, scale=0.63]{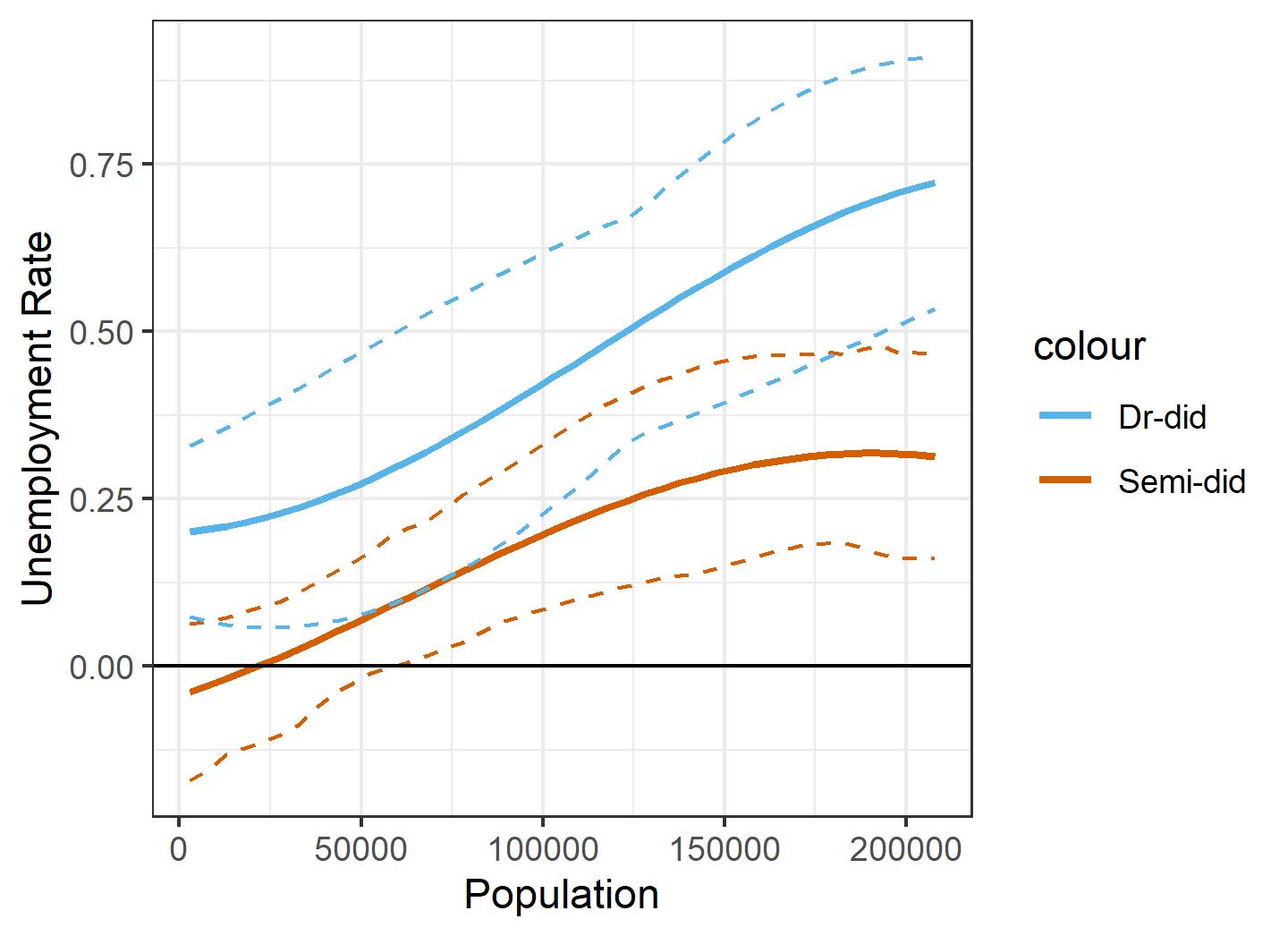}	}
\end{center}	
\caption{Compare Dr-DiD estimator (with 703 covariates) with Ab-did estimator \label{703diff}}
\end{figure}

\section{Discussion}
In this paper, we propose a new doubly robust two-stage difference-in-differences estimator that allows for, but does not require, the number of potential confounding covariates to be greater than the number of observations. Our estimator is robust to model miss-specification and a general set of machine learning tools can be used in our estimation procedure to estimate the propensity score. The outcome equation is modeled as a flexible partially linear form. The rate of convergence is derived for the new estimator and a novel de-bias procedure is proposed for inference. This allows the user to construct confidence intervals for the heterogeneous treatment effects. A simulation study shows promising finite sample performance of our estimator under different data generation processes. Our method is applied to study the effect of the Fair Minimum Wage Act on local unemployment rates and show heterogenous effects could rise due to the differences in demographics. Moreover, an R package for implementing the proposed method is available on Github. More work remains to be done. For example, it will be interesting to consider a similar estimation and inference strategy for panel data, or develop estimators for quantile treatment effect on the treated with high-dimensional covariates. We leave these topics for future studies.

\appendix
\section{Proofs} 
\subsection{Proof of Lemma \ref{doublerobust}}
\begin{proof}
For Part (i):let $\Phi_1(W_i;\theta_{1})$ and $\Phi_{0}(W_i;\theta_{0})$ be postulated models for $\Phi_1(W_i)$ and $\Phi_0(W_i)$. Let $\pi(W_i;\vartheta_{0})$ be a postulated model for the true propensity score $\mathbb{P}(D_i=1|W_i)$.
Since $\rho_0 = \frac{D_i - \mathbb{P}(D_i=1|W_i)}{\mathbb{P}(D_i=1|W_i)\mathbb{P}(D_i=0|W_i)}$, we have
\begin{equation*}
\begin{split}
&\mathbb{E}\left[\rho_0\Big(\Delta Y(i)-\mathbb{P}(D_i=0|W_i)\Phi_1(W_i) -\mathbb{P}(D_i=1|W_i)\Phi_0(W_i) \Big)|W_i \right]\\
& = \mathbb{E}\left[\frac{(D_i - \mathbb{P}(D_i=1|W_i))}{\mathbb{P}(D_i=1|W_i)\mathbb{P}(D_i=0|W_i)}\Delta Y(i) \Big| W_i\right]\\
& - \mathbb{E}\left[\frac{D_i - \mathbb{P}(D_i=1|W_i)}{\mathbb{P}(D_i=1|W_i)\mathbb{P}(D_i=0|W_i)}\Big(\mathbb{P}(D_i=0|W_i)\Phi_1(W_i) + \mathbb{P}(D_i=1|W_i)\Phi_0(W_i)  \Big)\Big| W_i \right]\\
& = \underbrace{\mathbb{E}\left[\frac{D_i\Delta Y(i)}{\mathbb{P}(D_i=1|W_i)} \Big| W_i \right] - \mathbb{E}\left[\frac{D_i - \mathbb{P}(D_i=1|W_i)}{\mathbb{P}(D_i=1|W_i)}\Phi_1(W_i) \Big| W_i \right]}_{PartL1.1}\\
& - \underbrace{\mathbb{E}\left[\frac{(1-D_i)\Delta Y(i)}{\mathbb{P}(D_i=0|W_i)} \Big| W_i \right] - \mathbb{E}\left[\frac{D_i - \mathbb{P}(D_i=1|W_i)}{\mathbb{P}(D_i=0|W_i)}\Phi_0(W_i) \Big| W_i \right].}_{PartL1.2}\\
\end{split}
\end{equation*}

First we consider Part L1.1:
\begin{equation*}
\begin{split}
&\mathbb{E}\left[\frac{D_i\Delta Y(i)}{\mathbb{P}(D_i=1|W_i)} \Big| W_i \right] - \mathbb{E}\left[\frac{D_i -\mathbb{P}(D_i=1|W_i)}{\mathbb{P}(D_i=1|W_i)}\Phi_1(W_i) \Big| W_i \right]\\
& = \mathbb{E}\left[\frac{D_i(Y^1(i, 1) - Y^0(i, 1))}{\mathbb{P}(D_i=1|W_i]} - \frac{D_i - \mathbb{P}(D_i=1|W_i)}{\mathbb{P}(D_i=1|W_i)}\Phi_1(W_i) \Big| W_i   \right]\\
& = \mathbb{E}\left[\Phi_1(W_i)  + \frac{D_i}{\mathbb{P}(D_i=1|W_i)}((Y^1(i, 1) - Y^0(i, 1)) -\Phi_1(W_i)) \Big| W_i   \right]\\
& = \Phi_1(W_i)  + \mathbb{E}\left[\frac{D_i}{\mathbb{P}(D_i=1|W_i)}((Y^1(i, 1) - Y^0(i, 1)) -\Phi_1(W_i))\Big| W_i \right].\\
\end{split}
\end{equation*}

Notice that when $\pi(W_i;\vartheta_{0})$ is misspecified, but $\Phi_1(W_i;\theta_{1})$ is correctly specified so $\Phi_{1}(W_{i})=\Phi_1(W_i;\theta_{1})$,  we have

\begin{equation*}
\begin{split}
&\mathbb{E}\left[\frac{D_i }{\pi(W_i;\vartheta_{0})}((Y^1(i, 1) - Y^0(i, 1)) - \Phi_1(W_i)) \Big| W_i   \right]\\
& = \mathbb{E}\left[((Y^1(i, 1) - Y^0(i, 1)) - \Phi_1(W_i)) \Big| W_i,  D_i = 1  \right]\frac{1-\pi(W_i;\vartheta_{0})}{\pi(W_i;\vartheta_{0})}=0.\\
\end{split}
\end{equation*}

When $\Phi_1(W_i;\theta_{1})$ is misspecified, but the propensity score $\pi(W_i;\vartheta_{0})$ is correctly specified so $\pi(W_i)=\pi(W_i;\vartheta_{0})$, we have

\begin{equation*}
\begin{split}
&\mathbb{E}\left[\frac{D_i\Delta Y(i)}{\mathbb{P}(D_i=1|W_i)} \Big| W_i \right] - \mathbb{E}\left[\frac{D_i - \mathbb{P}(D_i=1|W_i)}{\mathbb{P}(D_i=1|W_i)}\Phi_1(W_i,\theta) \Big| W_i \right]\\
& =\mathbb{E}[Y^1(i, 1) - Y^0(i, 0)|W_i, D_i=1) - \mathbb{E}\left[D_i - \mathbb{ P}(D_i=1|W_i) \Big| W_i \right] \frac{\Phi_1(W_i,\theta_{1})}{\mathbb{P}(D_i=1|W_i)}\\
& = \mathbb{E}[Y^1(i, 1) - Y^0(i, 0)|W_i, D_i=1). 
\end{split}
\end{equation*}

We next consider Part L1.2: 

When $\pi(W_i,\gamma)$ is misspecified, but $\Phi_0(W_i,\theta_{0})$ is correctly specified so $\Phi_{0}(W_i)=\Phi_0(W_i,\theta_{0})$, Part L1.2 is
\begin{equation*}
\begin{split}
&\mathbb{E}\left[\frac{(1-D_i)\Delta Y(i)}{1-\pi(W_i;\vartheta_{0})} \Big| W_i \right] + \mathbb{E}\left[\frac{D_i - \pi(W_i;\vartheta_{0})}{1-\pi(W_i;\vartheta_{0})}\Phi_0(W_i) \Big| W_i \right]\\
& = \mathbb{E}\left[\frac{(1-D_i)(Y^0(i, 1) - Y^0(i, 0))}{1-\pi(W_i;\vartheta_{0})} - \frac{(1-D_i) - \mathbb{P}(D_i=0|W_i)}{1-\pi(W_i;\vartheta_{0})}\Phi_0(W_i) \Big| W_i   \right)\\
&=  \mathbb{E}\left[(Y^0(i, 1) - Y^0(i, 0)) + \frac{(1-D_i) -(1-\pi(W_i;\vartheta_{0}))}{1-\pi(W_i;\vartheta_{0})}((Y^0(i, 1) - Y^0(i, 0)) - \Phi_0(W_i)) \Big| W_i   \right]\\
 &=\Phi_0(W_i) + \mathbb{E}\left[\frac{(1-D_i) }{1-\pi(W_i;\vartheta_{0})}((Y^0(i, 1) - Y^0(i, 0)) - \Phi_0(W_i)) \Big| W_i   \right] = \Phi_0(W_i).\\
\end{split}
\end{equation*}
And when $\Phi_0(W_i;\theta_{0})$ is misspecified, but propensity score $\pi(W_i;\vartheta_{0})$ is correctly specified, 
\begin{equation*}
\begin{split}
&\mathbb{E}\left[\frac{(1-D_i)\Delta Y(i)}{\mathbb{P}(D_i=0|W_i)} \Big| W_i \right]+ \mathbb{E}\left[\frac{D_i - \mathbb{P}(D_i=1|W_i]}{\mathbb{P}(D_i=0|W_i)} \Phi_0(W_i;\theta_{0}) \Big| W_i \right]\\
& = \mathbb{E}\left[(Y^0(i,1) - Y^0(i,0)) \Big| W_i, D_i=0 \right]+ \mathbb{E}\left[\frac{D_i - \mathbb{P}(D_i=1|W_i)}{\mathbb{P}(D_i=0|W_i)}\Phi_0(W_i;\theta_{0}) \Big| W_i \right]\\
& = \mathbb{E}\left[(Y^0(i,1) - Y^0(i,0)) \Big| W_i, D_i=0 \right] + \mathbb{E}\left[D_i - \mathbb{P}(D_i=1|W_i) \Big| W_i \right]\frac{\Phi_0(W_i;\theta_{0})}{\mathbb{P}(D_i=0|W_i)}\\
& = \mathbb{E}\left[(Y^0(i,1) - Y^0(i,0)) \Big| W_i, D_i=1 \right]. \qquad(\text{By Assumption \ref{Ass:parallel}})
\end{split}
\end{equation*}

The result in Part (ii) follows from Part (i) and direct calculation.
\end{proof}

\subsection{Proof of Theorem \ref{thm1}}
Let $B(s_0,p)$ be a set of $p-$ dimensional vectors with at most $s_{0}$ non-zero coordinates. Let $\widehat{Q}_{z}=\frac{1}{n}\sum_{i=1}^{n}\psi^{k_{n}}(Z_{i})\psi^{k_{n}}(Z_{i})^{\top}$. Recall the definition of $S_i = X_i^{\top}\beta_{0}-\psi^{k_n}(Z_i)^{\top}\gamma_{n0} + r_{ni} +\epsilon_{in} $.
Let $\bm{v}_{n}(\hat f, \hat \beta) = \bm{v}_{n}(\bm{Z}, \hat f) + \bm{v}_{n}(\bm{X}, \hat \beta)$, where $\bm{v}_{n}(\bm{Z}, \hat f)=\Psi_{n}(\hat{\gamma}_{n}-\gamma_{n0})$, $\bm{v}_{n}(\bm{X}, \hat \beta)=\bm{X}(\hat{\beta}-\beta_{0})$.
Let $\tilde{\bm{v}}_{n}(\bm{Z}, \hat f)=\Psi_{n}(\hat{\gamma}_{n}-\gamma_{n0})+\Pi_{n,X|Z}(\hat{\beta}-\beta_{0})$
and $\tilde{\bm{v}}_{n}(\bm{X}, \hat \beta)=\tilde{\bm{X}}(\hat{\beta}-\beta_{0}).$ Then $\bm{v}_{n}(\hat f, \hat \beta)=\tilde{\bm{v}}_{n}(\bm{Z}, \hat \gamma_{n})+\tilde{\bm{v}}_{n}(\bm{X}, \hat \beta)$. 
Define $\iota_{ni} = r_{ni} + \epsilon_{ni}$ and  $\bm{\iota}_{n}$ as the vector with $\iota_{in}$ on the $i$th position. Let $\bm{\iota}_{n}^{*}=P_Z\bm{\iota}_{n}$, we have $\bm{\iota}_{n}^{*\top}\tilde{\bm{v}}_{n}(\bm{Z}, \hat f)=\bm{\iota}_{n}^{\top}\tilde{\bm{v}}_{n}(\bm{Z}, \hat f)$ for $\bm{\iota}_{n}=\bm{\epsilon}+\bm{r}_n $. Define the following norm and set:

\begin{align*}
\tau(\beta, f, R) &= R^{-1}\lambda\|\beta\|_1+\|X\beta + f\|_{P,2}\\
\mathcal{M}_1(R)&=\left\{f: \|f-f_{n0}\|_{P,2} \leq 4R, f\in \mathcal{F} \right\}\\
\mathcal{M}_2(R)&=\left\{(\beta,f):\tau(\beta,f)\leq  R,\beta\in B(s_0,p) , f\in \mathcal{M}_1(R) \right\} \\
\mathcal{T}_{1}&=\left\{\sup_{\mathcal{M}_2(R)}|\|X^{\top}\beta+f\|^{2}_{n}-\|X^{\top}\beta+f\|^{2}_{P,2}|\leq  R^{2}/96\right\}\\
\mathcal{T}_{2}& =\{\|\bm{\iota}_{n}^{*}\|_{n}^{2}\leq R^{2}/192\}\\
\mathcal{T}_{3}&=\{|\bm{\iota}_{n}^{\top}\tilde{\bm{v}}_{n}(\bm{X}, \beta)/n|\leq \lambda/8\|{\beta}-\beta_0\|_{1}, (\beta,f) \in \mathcal{M}_2(R)\}\\
\mathcal{T}_{4}&=\{\Lambda_{\tilde{X},n}^{2}(s_{0})\geq\Lambda_{\tilde{X}}^{2}(s_{0})/2\}\\
\mathcal{T}_{5}&=\{|(\hat{\bm{S}} - \bm{S})^{\top}{\bm{v}}_{n}(\bm{X}, \beta)/n|\leq \lambda/8\|{\beta}-\beta_0\|_{1}, (\beta,f) \in \mathcal{M}_2(R)\}\\
\mathcal{T}_{6}&=\{\sqrt{k_n}\cdot \|(\hat{\bm{S}} - \bm{S})^{\top} \Psi_n /n\|_\infty \leq R/384\}\\
\mathcal{T}_{7}&=\{|(\hat{\bm{S}} - \bm{S})^{\top}\tilde{\bm{v}}_{n}(\bm{X}, \beta)/n|\leq \lambda/4\|\hat{\beta}-\beta_0\|_{1}, (\beta,f) \in \mathcal{M}_2(R)\}
\end{align*} 
and 
\[
\Lambda_{\tilde{X},n}^{2}(s_{0})=\min_{\delta\in\mathbb{R}^{p}\backslash\{0\},\|\delta_{S_{0}^{C}}\|_{1}\leq3\sqrt{s_{0}}\|\delta_{S_{0}}\|_{2}}\frac{\delta^{\top}\mathbb{E}_{n}\left[\tilde{X}_{i}\tilde{X}_{i}^{\top}\right]\delta}{\|\delta_{S_{0}}\|_{2}^{2}}.
\]

 \proof 
Let $(\hat{\beta},\hat{f})$ be the solution to then minimization problem in Equation \eqref{eq:Estimator}, define 
\[t := \frac{4R}{4R+\|\hat f - f_{n0}\|_{P,2}} \]
and $\tilde \gamma_n := t\hat \gamma_n + (1-t) \gamma_{n0}$ and thus $\tilde f=\tilde{f}(z)= \psi^{k_n}(z)^{\top}\tilde \gamma_n$. By convexity,  
\[
\Vert \hat{\bm{S}}- \tilde{\bm{X}}\hat{\beta} - \Pi_{n,X|Z}\hat{\beta} - \Psi_{n}\tilde{\gamma}_{n}\Vert^{2}_n+\lambda\|\hat{\beta}\|_{1}\leq\Vert \hat{\bm{S}}- \tilde{\bm{X}}{\beta_0} - \Pi_{n,X|Z}{\beta_0} - \Psi_{n}{\gamma}_{0}\Vert^{2}_n+\lambda\|\beta_{0}\|_{1},
\]
By the definition of $\bm{S}$ and $\bm{\iota}$ and Lemma \ref{lemR},  
\begin{equation}\label{eqlem3}
\begin{split}
\Vert\tilde{\bm{X}}(\hat{\beta}-\beta_0)\Vert^{2}_n+\lambda\|\hat{\beta}\|_{1} &\leq2\left|\bm{\iota}^{\top}\tilde{\bm{X}}\left(\hat{\beta}-\beta_{0}\right)\right| + 2\left|(\hat{\bm{S}}-\bm{S})^{\top}\tilde{\bm{X}}\left(\hat{\beta}-\beta_{0}\right)\right|  \\
&+2\left|(\tilde{\gamma}_{n}-\gamma_{n0})^{\top} \Psi_{n}^{\top}\tilde{\bm{X}}\left(\hat{\beta}-\beta_{0}\right)\right| \\
&+ 2\left|(\hat \beta - \beta_0)^{\top} \Pi_{n,X|Z}^{\top} \tilde{\bm{X}}\left(\hat{\beta}-\beta_{0}\right)\right| + \lambda\|\beta_{0}\|_{1}.
\end{split}
\end{equation}
Notice that $(\tilde{\gamma}_{n}-\gamma_{n0})^{\top}\Psi_{n}^{\top}\tilde X\left(\hat{\beta}-\beta_{0}\right) = 0$ and $(\hat \beta - \beta_0)^{\top}\Pi_{n,X|Z}^{\top} \tilde X\left(\hat{\beta}-\beta_{0}\right) = 0$. On the set  $\mathcal{T}_{3} \cap \mathcal{T}_{7}$
\[\mathbb{E}_{n}\left[\tilde{X}_{i}^\top\left(\beta_{0}-\hat{\beta}\right)\right]^2+\lambda\|\hat{\beta}\|_{1}  \leq \frac{\lambda}{2}\|\hat \beta - \beta_0\|_1 + \lambda\|\beta_0\|_{1}\]
Subtract $\|\hat\beta_{S_0}\|$ and add $\lambda\|\hat \beta_{S_0} - \beta_{0S_0}\|_1$ to both sides and from Assumption \ref{asseigen} (ii), on the set $\mathcal{T}_{4}$,
\begin{align*}
2\mathbb{E}_{n}\left[\tilde{X}_{i}^\top\left(\beta_{0}-\hat{\beta}\right)\right]^2&+\lambda\|\hat \beta - \beta_0\|_1  \leq 4\lambda\|\hat \beta_{S_0} - \beta_{0S_0}\|_1 \leq 4\lambda \sqrt{s_0}\|\hat \beta - \beta_0\|_1\\
& \leq \frac{4\lambda \sqrt{s_0}}{\Lambda_{\tilde{X}}(s_0)}\mathbb{E}_{n}\left[\tilde{X}_{i}^\top\left(\beta_{0}-\hat{\beta}\right)\right] \leq \mathbb{E}_{n}\left[\tilde{X}_{i}^\top\left(\beta_{0}-\hat{\beta}\right)\right]^2 + \frac{\lambda^2 s_0}{\Lambda^2_{\tilde{X}}(s_0)}
\end{align*}
As a result, from Lemma \ref{lem_iotav}, \ref{lem:eigenX}, and \ref{lem_S},  with probability approaching one, $\|\hat \beta - \beta_0\|_1 \leq \lambda s_0/\Lambda^2_{\tilde{X}}(s_0)$ and $\mathbb{E}_{n}\left[\tilde{X}_{i}^\top\left(\beta_{0}-\hat{\beta}\right)\right]^2 \leq \lambda^2 s_0/\Lambda^2_{\tilde{X}}(s_0)$. 

Lemma \ref{lemR}, \ref{lem4}, \ref{lem:eigen}, \ref{lem_iotav} and \ref{lem_S} imply that $\|\bm{X}(\beta_{0}-\hat{\beta})+(f_{n0}-\tilde{f})\|^{2}_{P,2} \leq R^2$ with probability approaching one. By orthogonal decomposition, it implies that $\|\tilde{\bm{X}}(\beta_{0}-\hat{\beta})\|^{2}_{P,2}\leq R^{2}$ and $\|\Pi_{n,X|Z}(\beta_{0}-\hat{\beta})+f_{n0}-\tilde{f}\|^{2}_{P,2} \leq R^{2}$ with probability approaching one. Then
$\|\Pi_{n,X|Z}(\beta_{0}-\hat{\beta})\|^{2}_{P,2}\leq_{p} 3R^2$ by Assumption \ref{asseigen} (iii) and \ref{assfeigen} (iii), so we have 
\[
\|\tilde{f}-f_{n0}\|_{P,2}^{2}\leq \|\Pi_{n,X|Z}(\beta_{0}-\hat{\beta})+f_{n0}-\tilde{f}\|_{P,2}^{2}+\|\Pi_{n,X|Z}(\beta_{0}-\hat{\beta})\|_{P,2}^{2}\leq_{p} 4R^2.
\]
which implies $\|\tilde{f}-f_{n0}\|_{P,2}\leq_{p} 2R$. Combining with Assumption \ref{assum:chooselambda}(ii), it yields that
\[\|\hat{f}-f_{n0}\|_{P,2}\leq_{p} 4R.\]
\qed \\

\begin{lem}\label{lemR}Suppose that Assumptions \ref{Ass:parallel}-\ref{assum:chooselambda} are satisfied. For $\tilde f$ as defined in the proof of Theorem \ref{thm1},
\[\tau(\hat\beta - \beta_0, \tilde f - f_{n0}, R) \leq R\]
on the event $\mathcal{T}_{1} \cap \mathcal{T}_{2} \cap \mathcal{T}_{3}  \cap \mathcal{T}_{5} \cap \mathcal{T}_{6}$.\end{lem}
\proof
Define 
\[\tilde t = \frac{R}{R + \tau(\hat\beta - \beta_0, \tilde f - f_{n0}, R) }\]
Let $\tilde \beta = \tilde t\hat\beta + (1-\tilde t) \beta_0$, $\vardbtilde{f} = \tilde t\tilde f + (1-\tilde t)f_{n0}$. Notice that the definition of $\vardbtilde{f}$ implies $\vardbtilde{\gamma}_{n} = \tilde t \tilde \gamma_{n} + (1-\tilde t) \gamma_{n0}$ and $\vardbtilde{f} \in \mathcal{M}_1(R)$. Since $\tau(\tilde\beta - \beta_0, \vardbtilde{f} - f_{n0}, R) = \tilde t\tau(\hat\beta - \beta_0, \tilde f - f_{n0}, R) \leq R $. Thus $(\tilde\beta - \beta_0, \vardbtilde{f} - f_{n0}) \in \mathcal{M}_2(R)$. To show $\tau(\hat\beta - \beta_0, \tilde f - f_{n0}, R) \leq R$, it is then sufficient to show $\tau(\hat\beta - \beta_0, \vardbtilde{f} - f_{n0}, R) \leq R/2$. From the definition of $(\hat{\beta},\vardbtilde{f})$, and convexity,   
\[
\mathbb E_n \left[\left\{\hat S_i-X_i^{\top}\tilde{\beta}-\psi^{k_n}(Z_i)^{\top}\vardbtilde{\gamma}_{n}\right\}^2\right]+\lambda\|\tilde{\beta}\|_{1}\leq\mathbb E_n \left[\left\{\hat S_i-X_i^{\top}\beta_{0}-\psi^{k_n}(Z_i)^{\top}\gamma_{n0}\right\}^{2}\right]+\lambda\|\beta_{0}\|_{1},
\]
Then by the definition of $\widehat{\bm{S}}$ and $\bm{\iota}_n$, 
\begin{align}
 &\Vert\Psi_n(\vardbtilde{\gamma}_n-\gamma_{n0})+\bm{X}(\hat{\beta}-\beta_0)\Vert^{2}_n+\lambda\|\hat{\beta}\|_{1}\nonumber\\
&\leq  2(\bm{\iota}_{n}+(\hat{\bm{S}} - \bm{S}))^{\top}\left(\Psi_{n}\left(\vardbtilde{\gamma}_{n}-\gamma_{n0}\right)/n+\bm{X}\left(\hat{\beta}-\beta_{0}\right)/n\right)+\lambda\|\beta_{0}\|_{1}  \label{eq:L3_1}
\end{align}
First notice that 
\begin{align*}
	\left|(\hat{\bm{S}} - \bm{S})^{\top}\Psi_{n}\left(\vardbtilde{\gamma}_{n}-\gamma_{n0}\right)/n\right|
	&\leq \|(\hat{\bm{S}} - \bm{S})^{\top} \Psi_n /n\|_\infty \|\vardbtilde{\gamma}_{n}-\gamma_{n0}\|_1\\
	& \leq \sqrt{k_n}\|(\hat{\bm{S}} - \bm{S})^{\top} \Psi_n /n\|_\infty \|\vardbtilde{\gamma}_{n}-\gamma_{n0}\|_2 \\
	&\leq \sqrt{k_n}/\Lambda_{min,(\widehat{Q}_{z}} \cdot \|(\hat{\bm{S}} - \bm{S})^{\top} \Psi_n /n\|_\infty \|\Psi_n (\vardbtilde{\gamma}_n-\gamma_{n0})\|_n. 
\end{align*}
where $\Lambda_{min,(\widehat{Q}_{z}}$ is the minimum eigenvalue of $\mathbb{E}(\Psi_n^{\top}\Psi_n/n)$ and is bounded away from 0 by Assumption \ref{assfeigen}. On the set $\mathcal{T}_{5}$ and $\mathcal{T}_{6}$, and for $\vardbtilde{f}\in \mathcal{M}_1(R)$
\begin{align}\label{eq:L3_new}
	2\left|(\hat{\bm{S}} - \bm{S})^{\top}{\bm{v}}_{n}(\vardbtilde{f}, \tilde \beta)/n\right|\leq \lambda/4\|\tilde{\beta}-\beta_0\|_{1} +  R^2/192.
\end{align}
By the Cauchy-Schwarz inequality, 
\[
2\left|\bm{\iota}_{n}^{*\top}\tilde{\bm{v}}_{n}(\bm{Z}, \vardbtilde{f})/n\right|\leq 2\|\bm{\iota}_{n}^{*}\|_{n}\cdot\|\tilde{\bm{v}}_{n}(\bm{Z}, \vardbtilde{f})\|_{n}\leq 2\|\bm{\iota}_{n}^{*}\|_{n}^{2}+\frac{1}{2}\|\tilde{\bm{v}}_{n}(\bm{Z}, \vardbtilde{f})\|_{n}^{2}.
\]
Therefore, 
\begin{align}
& 2\bm{\iota}_{n}^{\top}\Psi_{n}(\vardbtilde{\gamma}_{n}-\gamma_{n0})/n+2\bm{\iota}_{n}^{\top}\Pi_{n,X|Z}(\tilde{\beta}-\beta_{0})/n+2\bm{\iota}_{n}^{\top}\tilde{\bm{X}}(\tilde{\beta}-\beta_{0})/n\nonumber\\
& =2\bm{\iota}_{n}^{*\top}\tilde{\bm{v}}_{n}(\bm{Z}, \vardbtilde{\gamma})/n+2\bm{\iota}_{n}^{\top}\tilde{\bm{v}}_{n}(\bm{X}, \tilde \beta)/n\nonumber\\
&\leq 2\|\bm{\iota}_{n}^{*}\|_{n}^{2}+\frac{1}{2}\|\tilde{\bm{v}}_{n}(\bm{Z}, \vardbtilde{f})\|_{n}^{2}+2\bm{\iota}_{n}^{\top}\tilde{\bm{v}}_{n}(\bm{X}, \tilde \beta)/n.\label{eq: L3_2}
\end{align}
Then, (\ref{eq:L3_1}), (\ref{eq:L3_new}),  and (\ref{eq: L3_2}) imply that 
\begin{align*}
 &\Vert\Psi_n(\vardbtilde{\gamma}_n-\gamma_{n0})+\bm{X}(\tilde{\beta}-\beta_0)\Vert^{2}_n+\lambda\|\tilde{\beta}\|_{1} =\|\bm{v}_{n}\|^{2}_{n}+\lambda\|\tilde{\beta}\|_{1} \\
&\leq  2\|\bm{\iota}_{n}^{*}\|_{n}^{2}+\frac{1}{2}\|{\bm{v}}_{n}\|_{n}^{2}+2\bm{\iota}_{n}^{\top}\tilde{\bm{v}}_{n}(\bm{X}, \tilde \beta)/n+\lambda/4\|\tilde{\beta}-\beta_0\|_{1} +  R^2/96 + \lambda\|\beta_0\|_1
\end{align*}
where the last inequality follows from the orthogonal decomposition such that
\[
\|\bm{v}_n(\vardbtilde{f}, \tilde \beta)\|^{2}_{n}=\|\tilde{\bm{v}}_n(\bm{X}, \tilde \beta)\|^{2}_{n}+\|\tilde{\bm{v}}_n(\bm{Z}, \vardbtilde{f})\|^{2}_{n}\geq \|\tilde{\bm{v}}_n(\bm{Z}, \vardbtilde{f})\|^{2}_{n},
\]
which implies that
\begin{align}
	\|\bm{v}_{n}(\vardbtilde{f}, \tilde \beta)\|^{2}_{n}+2\lambda\|\tilde{\beta}\|_{1}&\leq 4\|\bm{\iota}_{n}^{*}\|_{n}^{2}+4\bm{\iota}_{n}^{\top}\tilde{\bm{v}}_{n}(\bm{X}, \tilde \beta)/n+\lambda/2\|\tilde{\beta}-\beta_0\|_{1} \nonumber\\
	&+  R^2/96+2\lambda\|\beta_0\|_1, \label{eq: L3_3}
\end{align}
On the event $\mathcal{T}_{3}$, $|\bm{\iota}_{n}^{\top}\tilde{\bm{v}}_{n}(\bm{X}, \tilde \beta)/n|\leq \lambda/8\|\tilde{\beta}-\beta_0\|_{1}$. Thus \eqref{eq: L3_3} is equivalent to
\begin{equation}
\|\bm{v}_{n}(\vardbtilde{f}, \tilde \beta)\|^{2}_{n}+2\lambda\|\tilde{\beta}\|_{1}\leq 4\|\bm{\iota}_{n}^{*}\|_{n}^{2}+\lambda\|\tilde{\beta}-\beta_0\|_{1}+  R^2/96 + 2\lambda\|\beta_0\|_1.\label{eq:Lemma3:1}
\end{equation}
Because
\begin{align*}
\|\tilde{\beta}-\beta_{0}\|_{1}&=\|\tilde{\beta}_{S_{0}}-\beta_{0,S_{0}}\|_{1}+\|\tilde{\beta}_{S_{0}^{C}}\|_{1},\\
\|\tilde{\beta}\|_{1}&=\|\tilde{\beta}_{S_{0}}\|_{1}+\|\tilde{\beta}_{S_{0}^{C}}\|_{1}\geq\|\beta_{0,S_{0}}\|_{1}-\|\tilde{\beta}_{S_{0}}-\beta_{0,S_{0}}\|_{1}+\|\tilde{\beta}_{S_{0}^{C}}\|_{1},
\end{align*}
(\ref{eq:Lemma3:1}) further implies that
\begin{align*}
&\|\bm{v}_{n}(\vardbtilde{f}, \tilde \beta)\|_{n}^{2}+2\lambda\left(\|\beta_{0,S_{0}}\|_{1}-\|\tilde{\beta}_{S_{0}}-\beta_{0,S_{0}}\|_{1}+\|\hat{\beta}_{S_{0}^{C}}\|_{1}\right)\\
 &\leq \lambda\left(\|\tilde{\beta}_{S_{0}}-\beta_{0,S_{0}}\|_{1}+\|\tilde{\beta}_{S_{0}^{C}}\|_{1}\right)+2\lambda\|\beta_{0}\|_{1}+4\|\bm{\iota}_{n}^{*}\|_{n}^{2}+ R^2/96.
\end{align*}
Since $\left|\|\bm{v}_{n}(\vardbtilde{f}, \tilde \beta)\|_{n}^{2}-\|\bm{v}_{n}(\vardbtilde{f}, \tilde \beta)\|_{P,2}^{2}\right|\leq R^{2}/96$ on $\mathcal{T}_{1}$,
\begin{equation}
\|\bm{v}_{n}(\vardbtilde{f}, \tilde \beta)\|_{P,2}^{2}+\lambda\|\tilde{\beta}_{S_{0}^{C}}\|_{1}\leq3\lambda\|\tilde{\beta}_{S_{0}}-\beta_{0,S_{0}}\|_{1}+4\|\bm{\iota}_{n}^{*}\|_{n}^{2}+ R^2/48.\label{eq:A.5}
\end{equation}
By Assumption \ref{asseigen} (ii), and take $144\lambda^2s_0/\Lambda^{2}_{\bar{X}}(s_0)\leq R^2$,
\begin{align}
&\lambda\|\tilde{\beta}_{S_{0}}-\beta_{0,S_{0}}\|_{1}\leq \lambda\sqrt{s_0}\|\tilde{\beta}_{S_{0}}-\beta_{0,S_{0}}\|_2\leq \lambda\sqrt{s_0}\|\tilde{\bm{X}}(\tilde{\beta}-\beta_{0})\|_{P,2}/\Lambda_{\bar{X}}(s_0)\nonumber\\
&\leq\lambda^{2}s_{0}/\Lambda^2_{\bar{X}}(s_0)+\|\tilde{\bm{X}}(\tilde{\beta}-\beta_{0})\|_{P,2}^{2}/4\leq \lambda^{2}s_{0}/\Lambda^2_{\bar{X}}(s_0)+\|\bm{v}_{n}(\vardbtilde{\gamma}, \tilde \beta)\|^{2}_{P,2}/4\nonumber\\
&\leq R^{2}/144+\|\bm{v}_{n}(\vardbtilde{f}, \tilde \beta)\|^{2}_{P,2}/4, \label{eq:A.6}
\end{align}
where the third inequality follows becaue $ab<a^{2}+b^{2}/4$.

Adding $\lambda\|\tilde{\beta}_{S_{0}}-\beta_{0,S_{0}}\|_{1}$ on both sides of (\ref{eq:A.5}) yields that 
\begin{align*}
\|\bm{v}_{n}(\vardbtilde{f}, \tilde \beta)\|_{P,2}^{2}+\lambda\|\tilde{\beta}-\beta_{0}\|_{1}=&\|\bm{v}_{n}(\vardbtilde{f}, \tilde \beta)\|_{P,2}^{2}+\lambda\|\tilde{\beta}_{S_{0}^{C}}\|_{1}+\lambda\|\tilde{\beta}_{S_{0}}-\beta_{0,S_{0}}\|_{1}\\
\leq& 4\lambda\|\tilde{\beta}_{S_{0}}-\beta_{0,S_{0}}\|_{1}+4\|\bm{\iota}_{n}^{*}\|_{n}^{2}+R^2/24\\
\leq& R^2/36+\|\bm{v}_{n}(\vardbtilde{f}, \tilde \beta)\|_{P,2}^{2}+4\|\bm{\iota}_{n}^{*}\|_{n}^{2}+ R^2/48,
\end{align*}
which implies that 
\[
\lambda\|\tilde{\beta}-\beta_0\|_1\leq 4\|\bm{\iota}_{n}^{*}\|_{n}^{2}+7R^2/144.\label{eq:lambda}
\]
On the set $\mathcal{T}_{2}$, $\|\bm{\iota}_{n}^{*}\|^{2}_{n}\leq R^{2}/192$. Thus we find the bound of $\|\tilde{\beta}-\beta_0\|_1$ such that 
\[
\|\tilde{\beta}-\beta_0\|_1\leq \frac{5}{72}R^{2}/\lambda. \label{eq:beta_1_bound}
\]
On the other hand, substituting (\ref{eq:A.6}) into (\ref{eq:A.5}) and combining with $\|\bm{\iota}_{n}^{*}\|^{2}_{n}\leq R^{2}/192$ yield that 
\[
\|\bm{v}_{n}(\vardbtilde{f}, \tilde \beta)\|^{2}_{P,2}=\|\vardbtilde{f}-f_{n0}+\bm{X}(\tilde{\beta}-\beta_{0})\|^{2}_{P,2}\leq \frac{1}{16} R^{2}.
\]
As a result, 
\[\tau(\tilde\beta - \beta_0, \vardbtilde{f} - f_{n0}, R) = R^{-1}\lambda\|\tilde \beta - \beta_0\|_1+\|\vardbtilde{f}-f_{n0} + \bm{X}(\tilde\beta - \beta_0) \|_{P,2} \leq \frac{23}{72} R \leq R/2\]
\qed

\begin{lem} \label{lem4}
Suppose that Assumptions \ref{Ass:parallel}-\ref{assum:chooselambda} are satisfied, we have $\mathbb{P}(\mathcal{T}_{1} ) = 1-O(1/p \wedge 1/k_n)$
\end{lem}

\proof 
First notice that $\tau(\beta, f_{n}, R) \leq R$ implies for $(\beta, f_{n})\in \mathcal{M}_2(R)$,

\[\|X\beta +f\|_{P,2}^2 \leq R^2, \qquad \|\beta\|_1\leq R^2/\lambda,\]
which further implies $\|\bar X\beta\|_{P,2}^2 \leq R^2$ and $\|\Pi_{X|Z}^{\top}\beta+f\|_{P,2}^2 \leq R^2$.
By Assumption \ref{asseigen}, we then have
\begin{align*}
&\|X\beta\|_{P,2}\leq \|\bar X\beta\|_{P,2} + \|\Pi_{X|Z}^{\top}\beta\|_{P,2} \leq R+\Lambda_{\max}\left(\Sigma_{\Pi}\right)/\Lambda_{\min}\left(\Sigma_{\Pi}\right)\|\bar X\beta\|_{P,2}\\
& \leq \left(1+\Lambda^{2}_{\max}\left(\Sigma_{\Pi}\right)/ \Lambda^{2}_{\min}\left(\Sigma_{\bar X}\right)\right)R :=R_1,
\end{align*}
 and
\begin{align}
\|f_{n}\|_{P,2}& \leq \|f+\Pi_{X|Z}^{\top}\beta\|_{P,2}+\|\Pi_{X|Z}^{\top}\beta\|_{P,2}\leq \left(1+\Lambda^{2}_{\max}\left(\Sigma_{\Pi}\right)/ \Lambda^{2}_{\min}\left(\Sigma_{\bar X}\right)\right)R=R_1. \label{eq:f_n}
\end{align}
For any $(\beta, f_{n})\in \mathcal{M}_{2}(R)$, consider the decomposition such that
\begin{align*}
 & \left|\|X\beta+f_{n}\|_{n}^{2}-\|X\beta+f_{n}\|^{2}_{P,2}\right|\\
\leq & \underbrace{|(\mathbb{E}_{n}-\mathbb{E})[\|X_{i}^{\top}\beta\|^2]|}_{(A)}
+\underbrace{|(\mathbb{E}_{n}-\mathbb{E})[\|\psi^{k_n}(Z_i)^{\top}\gamma_{n}\|^2]|}_{(B)}
+2\underbrace{|\left(\mathbb{E}_{n}-\mathbb{E}\right)[\beta^{\top}X_i \psi^{k_n}(Z_i)^{\top}\gamma_n]|}_{(C)}
\end{align*}

For Term (A), let $\Sigma_{X}=\mathbb{E}\left[X_{i}X_{i}^{\top}\right]$. For all $(\beta,f)\in\mathcal{M}_{2}(R)$, 

\[
(A) = \left|\frac{1}{n}\sum_{i=1}^{n}\beta^{\top}X_{i}X_{i}^{\top}\beta-\beta^{\top}\Sigma_{X}\beta\right|\leq\|\left(\mathbb{E}_{n}-\mathbb{E}\right)X_{i}X_{i}^{\top}\|_{\infty}\|\beta\|_{1}^{2}.
\]
From Bernstein inequality, with probability at least $1-1/p$,  
\[\|\left(\mathbb{E}_{n}-\mathbb{E}\right)X_{i}X_{i}^{\top}\|_{\infty} \lesssim \sqrt{\frac{\log p}{n}}\]
because $X_{i}$'s are sub-Gaussian variables. Thus, from Assumption \ref{assum:chooselambda}, $ R^2 \leq \lambda$, and with probability at least $1-1/p$, 
\[
\sup_{(\beta,f)\in\mathcal{M}_{2}(R)}\left|\left(\mathbb{E}_{n}-\mathbb{E}\right)\|X_{i}^{\top}\beta\|^{2}\right|\leq CR^{4}/\lambda \leq CR^2.
\]
For Term (B), from Assumption \ref{assfeigen}, $Q_{Z}:=\mathbb{E}\left[\psi^{k_n}(Z_{i})\psi^{k_n}(Z_{i})^{\top}\right]$ can be normalized to $I_{k_{n}}.$
\begin{align*}
\sup_{(\beta,f_{n})\in \mathcal{M}_{2}(R)}\left|\left(\mathbb{E}_{n}-\mathbb{E}\right)f_{n}^{2}\right| & =\sup_{(\beta,f_{n})\in \mathcal{M}_{2}(R)}\left|\frac{1}{n}\sum_{i=1}^{n}\gamma_{n}^{\top}\psi^{k_n}(Z_{i})\psi^{k_n}(Z_{i})^{\top}\gamma_{n}-\gamma_{n}^{\top}\gamma_{n}\right|\\
& \leq \sup_{(\beta,f_{n})\in \mathcal{M}_{2}(R)} \|\left(\mathbb{E}_{n}-\mathbb{E}\right)\psi^{k_{n}}(Z_{i})\psi^{k_{n}}(Z_{i})^{\top}\|_{2}\|\gamma_{n}\|_{2}^{2}\\
\end{align*}
where $\|\gamma_{n}\|^{2}_{2}=\|f_n\|^{2}_{P,2}$ is bounded by $R^{2}_1$ following from (\ref{eq:f_n}). Moreover, by a Bernstein type inequality for random matrices (Theorem 6.1 in \cite{tropp2012user}; see also Theorem 4.3 in \cite{vandeGeer2014uniform}, Theorem 4.1 in \cite{ChenChristensen2015} or Lemma 6.2 in \cite{BCCK2015}), we have
\[
\mathbb{P}\left(\|\left(\mathbb{E}_{n}-\mathbb{E}\right)\psi^{k_{n}}(Z_{i})\psi^{k_{n}}(Z_{i})^{\top}\|_{2}>C\left(\xi_{0}(k_n)\sqrt{\frac{\log k_{n}+t}{n}}+\xi^{2}_{0}(k_n)\left(\frac{\log k_{n}+t}{n}\right)\right)\right)\leq\exp(-t).
\]
By taking $t=\log k_n$, with probability at least $1-O(1/k_n)$, 
\[
\mathbb{P}
\sup_{(\beta,f_{n})\in \mathcal{M}_{2}(R)}\left|\left(\mathbb{E}_{n}-\mathbb{E}\right)f_{n}^{2}\right|\leq CR^{2}_1
\]
by Assumption \ref{assfeigen} (ii). 
For Term (C), we have 
\begin{align*}
&\sup_{(\beta,f_{n})\in\mathcal{M}_{2}(R)}\left|\left(\mathbb{E}_{n}-\mathbb{E}\right)\left[\|\beta^{\top}X_{i}f_{n}(Z_{i})\|^{2}\right]\right| \\
& =\sup_{(\beta,f_{n})\in \mathcal{M}_{2}(R)}\left|\frac{1}{n}\sum_{i=1}^{n}\beta^{\top}X_{i}\psi^{k_n}_{n}(Z_{i})^{\top}\gamma_{n}-\mathbb{E}\left[\beta^{\top}X_{i}\psi^{k_n}_{n}(Z_{i})^{\top}\gamma_{n}\right]\right|\\
& \leq \sup_{(\beta,f_{n})\in\mathcal{M}_{2}(R)}\|\beta\|_{1} \|\gamma_{n}\|_{1}\left\Vert \left(\mathbb{E}_{n}-\mathbb{E}\right)X_{i}\psi_{n}^{k_{n}}(Z_{i})^{\top}\right\Vert _{\infty}\\
 & \leq\sup_{(\beta,f_{n})\in\mathcal{M}_{2}(R)}\|\beta\|_{1}\sqrt{k_{n}}\|\gamma_{n}\|_{2}\left\Vert \left(\mathbb{E}_{n}-\mathbb{E}\right)X_{i}\psi_{n}^{k_{n}}(Z_{i})^{\top}\right\Vert _{\infty}.
\end{align*}
Note that $\vert X_{ij}\psi^{k_n}_{m}(Z_i)\vert \leq \xi_{0}(k_{n})|X_{ij}|$ for $j=1,\dots,p$ and $m=1,\dots,k_{n}$.
Thus, Lemma 14.15 in \cite{buhlmann2011statistics} implies that given $X$,
\begin{align*}
&\mathbb{P}\left(\max_{1\leq j\leq p}\max_{1\leq m\leq k_{n}}\left|\left(\mathbb{E}_{n}-\mathbb{E}\right)X_{ij}\psi_{m}^{k_{n}}(Z_{i})\right|\geq\max_{1\leq j\leq p}\sqrt{\frac{\xi_{0}^{2}(k_{n})\sum_{i=1}^{n}X_{ij}}{n}}\sqrt{2\left(t^{2}+\frac{2\log2p}{n}\right)}\right)\\
& \leq \exp (-nt^{2}).
\end{align*}
Because $X$ is sub-Gaussian, letting $t^{2}=\log (2p)/n$ gives 
\[
\mathbb{P}\left(\max_{1\leq j\leq p}\max_{1\leq m\leq k_{n}}\left|\left(\mathbb{E}_{n}-\mathbb{E}\right)X_{ij}\psi_{m}^{k_{n}}(Z_{i})\right|\geq \sqrt{2K_{X}\xi_{0}^{2}(k_{n})}\sqrt{\frac{\log2p}{n}}\right)\leq p^{-1}. 
\]
Because for $(\beta,f_n)\in \mathcal{M}_{2}(R)$, $\|\beta\|_{1}\leq R^{2}/\lambda$ and $\|\gamma_{n}\|_2=\|f_n\|_{2}\leq R^2_{1}$, it implies that $\|\beta\|_{1}\|\gamma_{n}\|_{2}\lesssim R^{4}/\lambda\lesssim R^{2}$ by Assumption \ref{assum:chooselambda}(ii). Then combing with Assumption \ref{assfeigen}(iii), we have that with probability at least $1-O(1/p)$,
\[
\sup_{(\beta,f_{n})\in\mathcal{M}_{2}(R)}\left|\left(\mathbb{E}_{n}-\mathbb{E}\right)\left[\|\beta^{\top}X_{i}f_{n}(Z_{i})\|^{2}\right]\right| \lesssim R^{2}.
\]
The conclusion follows from Assumption \ref{asseigen}. 
\qed

\begin{lem} \label{lem:eigen}
Suppose that Assumptions \ref{Ass:parallel}-\ref{assum:chooselambda} are satisfied. For a sequence $\kappa_n \rightarrow \infty$ as $n \rightarrow \infty$ and $\kappa_n$ does not depend on $p$ or $k_n$, the set $\mathcal{T}_{2}$ has probability at least $1-O(1/\kappa_{n}+1/k_{n})$. 
\end{lem}
\begin{proof}
\[
\|\bm{\iota}_{n}^{*}\|_{n}^{2}\leq2\bm{\epsilon}^{\top}P_{z}\bm{\epsilon}/n+2\bm{r}_{n}^{\top}P_{z}\bm{r}_{n}/n.
\]
For the first term, note that 
\begin{align*}
\bm{\epsilon}^{\top}P_{z}\bm{\epsilon}/n & =\|\left(\Psi_{n}^{\top}\Psi_{n}/n\right)^{-1/2}\Psi_{n}\bm{\epsilon}/n\|_{2}^{2}=\|\widehat{Q}_{z}^{-1/2}\mathbb{E}_{n}\left[\psi^{k_{n}}(Z_{i})\epsilon_{i}\right]\|_{2}^{2}\\
 & \leq\|\widehat{Q}_{z}^{-1/2}\|_{2}^{2}\|\mathbb{E}_{n}\left[\psi^{k_{n}}(Z_{i})\epsilon_{i}\right]\|_{2}^{2},
\end{align*}
where all eigenvalues of $\widehat{Q}_{z}$ are bounded away with probability at least $1-1/k_{n}$ following the matrix Bernstein inequality. Moreover,
\begin{align*}
&\|\mathbb{E}\left[\psi^{k_n}(Z_i )\epsilon_{i}\right]\|^{2}_{2}=\mathbb{E}[\epsilon^2_{i}\psi^{k_n}(Z_i)^{\top}\psi^{k_n}(Z_i)/n]\lesssim \mathbb{E}[\psi^{k_n}(Z_i)^{\top}\psi^{k_n}(Z_i)/n]\\
&=tr(\mathbb{E}[\psi^{k_n}(Z_i)\psi^{k_n}(Z_i)^{\top}/n])=tr(I_{k_{n}}/n)=(k_{n}/n).
\end{align*}
For the second term, 
\[
\bm{r}_{n}^{\top}P_{z}\bm{r}_{n}/n=\|\widehat{Q}_{z}^{-1/2}\mathbb{E}_{n}\left[\psi^{k_{n}}(Z_{i})r_{ni}\right]\|_{2}^{2}.
\]
Because
\[
\mathbb{E}\left[\|\mathbb{E}_{n}\left[\psi^{k_{n}}(Z_{i})r_{ni}\right]\|_{2}^{2}\right]=\frac{1}{n^{2}}\sum_{k=1}^{k_{n}}\mathbb{E}\left[\psi_{k}(Z_{i})^{2}r_{ni}^{2}\right]\leq\left(\frac{\ell_{k_{n}}c_{k_{n}}}{\sqrt{n}}\right)^{2}\mathbb{E}\left[\|\psi^{k_{n}}(Z_{i})\|_{2}^{2}\right]=\frac{\ell_{k_{n}}^{2}c_{k_{n}}^{2}k_{n}}{n}
\]
or 
\[
\mathbb{E}\left[\|\mathbb{E}_{n}\left[\psi^{k_{n}}(Z_{i})r_{ni}\right]\|_{2}^{2}\right]\leq \frac{1}{n}\mathbb{E}[\xi_{0}^{2}(k_n)r^{2}_{ni}]\leq \frac{\xi_{0}^{2}(k_{n})c^{2}_{k_n}}{n}
\]
so we have 
\[
\mathbb{E}\left[\|\mathbb{E}_{n}\left[\psi^{k_{n}}(Z_{i})r_{ni}\right]\|_{2}^{2}\right] \leq \min\left(\frac{\ell_{k_{n}}^{2}c_{k_{n}}^{2}k_{n}}{n}, \frac{\xi_{0}^{2}(k_{n})c^{2}_{k_n}}{n}\right).
\]
From triangular inequality and Markov inequality,there exists a constant $C$, such that $\mathbb{E}\left[\|\mathbb{E}_{n}\left[\psi^{k_{n}}(Z_{i})r_{ni}\right]\|_{2}^{2}\right] = CR^2/\kappa_n$ and by Markov inequality, 
\[\mathbb{P}\left(\left\|\mathbb{E}_{n}\left[\psi^{k_{n}}(Z_{i})r_{ni}\right]\right\|_{2}^{2} >R^2/192\right) \leq  192C/\kappa_n\]

\end{proof}

\begin{lem}\label{lem_iotav}
Suppose that Assumptions \ref{Ass:parallel}-\ref{assum:chooselambda} are satisfied, then $\mathbb{P}(\mathcal{T}_3) = 1 - O(1/p+1/k_{n})$
\end{lem}
\proof
By definition of $\bm{\iota}_n$,
\[
|\bm{\iota}_{n}^{\top}\tilde{\bm{v}}_{n}(\bm{X}, \tilde \beta)/n|
=|(\bm{\epsilon}+\bm{r}_n)^{\top}\tilde{\bm{X}}(\tilde{\beta}-\beta_0)|\leq \|(\bm{\epsilon}+\bm{r}_n )^{\top}\tilde{\bm{X}}\|_{\infty}\|\tilde{\beta}-\beta_0\|_1.
\]
By Lemma 6.2 in \cite{BCCK2015}, all eigenvalues of $\widehat{Q}_{Z}$ are bounded away from zero with Assumption \ref{assfeigen} on the set $\mathcal{T}_{3}$ with probability at least $1-1/k_{n}$.
Notice that for $\tilde{X}_{i} = (I - P_{Z})X_i$. As $X_i$ is sub-Gaussian, its moment generating function has
\[\mathbb{E}[\exp(s((I - P_{Z}) X_{ij}))] \leq  \mathbb{E}[\exp(s(1-\Lambda^2_{\min}(\widehat{Q}_{Z})) X_{ij})]\leq \exp(K_{X}^2s^2(1-\Lambda^2_{\min}(\widehat{Q}_{Z}))^2/2). \]
Thus $\tilde{X}_{ij}$ is also sub-Gaussian so Assumption \ref{assum:error} implies that $\max_{1\leq i\leq n}, \max_{1\leq j\leq n}\mathbb{E}[(\tilde{X}_{ij}\epsilon_i)^2]$ is bounded from above. From  Lemma E.1 and E.2 of \cite{chernozhukov2017central}, 

\[\mathbb{E}\left(\max_{1\leq j \leq p}\left|\frac{1}{n}\sum_{i = 1}^n\epsilon_i \tilde X_{ij} \right| > C\sqrt{\frac{t+\log p}{n}}\right)\leq \exp(-t).\]
Because $\lambda\gtrsim 4\cdot C\sqrt{2\log p/n}$, we have with probability at least $1-1/p$ 
\[
\max_{1\leq j\leq p}\left|\frac{1}{n}\sum_{i=1}^{n}\epsilon_{i}\tilde{X}_{ij}\right|\leq \lambda/4.
\]
Furthermore, by Bernstein inequality, we have for some constant $K_{\tilde{X}}\geq 1$, 
\[\mathbb{P}\left(\max_{1\leq j\leq p}\|\tilde{X}_{j}\|_{n}^{2}\geq2K_{\tilde{X}}\right)\leq \mathbb{P}\left(\max_{1\leq j\leq p}\|\tilde{X}_{j}\|_{n}^{2}\geq\mathbb{E}\|\tilde{X}_{j}\|_{n}^{2}+K_{\tilde{X}}\sqrt{\frac{\log p}{n}}\right)\leq(2p)^{-1}.\]
Next, note that $\max_{1\leq i\leq n}\max_{1\leq j\leq p}|r_{ni}\tilde{X}_{ij}|^2\leq \max_{1\leq i\leq n}\max_{1\leq j\leq p}\ell^{2}_{k_n}c^{2}_{k_n}|\tilde{X}_{ij}|^2,$
by using Lemma 14.15 of \cite{buhlmann2011statistics}, we have 
\[
\mathbb{P}\left(\max_{1\leq j\leq p}\left|\frac{1}{n}\sum_{i=1}^{n}r_{ni}\tilde{X}_{ij}\right|\geq\max_{1\leq j\leq p}\sqrt{\ell^{2}_{k_n}c^{2}_{k_n}\sum_{i=1}^{n}\tilde{X}_{ij}^{2}/n}\sqrt{2(t^{2}+\log p/n)}\right)\leq\exp(-nt^{2}).
\]
Taking $t^2=\log p/n$ yields that with probability at least $1-1/(2p)-1/p$, 
\[
\max_{1\leq j\leq p}\left|\frac{1}{n}\sum_{i=1}^{n}r_{ni}\tilde{X}_{ij}\right|\leq 2\sqrt{K_{\tilde{X}}}\ell_{k_{n}}c_{k_n}\sqrt{\log p/n}.
\]
Thus, when $\lambda\gtrsim 8 \cdot C\sqrt{2\log p/n}$ and because $l_{k_n}c_{k_n}=o(1)$, we have 
\[
\max_{1\leq j\leq p}\left|\frac{1}{n}\sum_{i=1}^{n}(\epsilon_{i}+r_{ni} )\tilde{X}_{ij}\right|\leq \lambda/8.
\]
which further implies that
$|\bm{\iota}_{n}^{\top}\tilde{\bm{v}}_{n}(\bm{X}, \tilde \beta)/n|\leq \lambda/8\|\tilde{\beta}-\beta_0\|_{1}$.
\qed

\begin{lem} \label{lem:eigenX}
Suppose that Assumptions \ref{Ass:parallel}-\ref{assum:chooselambda} are satisfied. Then $\mathbb{P}(\mathcal{T}_{4})=1 - O(1/p)$.
\end{lem}
\begin{proof}
For $\hat{\Sigma}_{\tilde{X}}=\mathbb{E}_{n}\left[\tilde{X}_{i}\tilde{X}_{i}^{\top}\right]$
and $\Sigma_{\tilde{X}}=\mathbb{E}\left[\tilde{X}_{i}\tilde{X}_{i}^{\top}\right]$, we have 
\begin{align*}
\frac{\delta^{\top}\hat{\Sigma}_{\tilde{X}}\delta}{\|\delta_{S_{0}}\|_{2}^{2}} & =\frac{\delta^{\top}\Sigma_{\tilde{X}}\delta+\delta^{\top}\left(\hat{\Sigma}_{\tilde{X}}-\Sigma_{\tilde{X}}\right)\delta}{\|\delta_{S_{0}}\|_{2}^{2}}\\
 & \geq\Lambda_{\tilde{X}}^{2}(s_{0})-\left|\frac{\|\delta\|_{1}^{2}\|\hat{\Sigma}_{\tilde{X}}-\Sigma_{\tilde{X}} \|_{\infty}}{\|\delta_{S_{0}}\|_{2}^{2}}\right|\\
 & \geq\Lambda_{\tilde{X}}^{2}(s_{0})-16s_{0}\|\hat{\Sigma}_{\tilde{X}}-\Sigma_{\tilde{X}}\|_{\infty},
\end{align*}
where the first inequality follows by the definition of $\Lambda_{\tilde{X}}^{2}(s_{0})$
and the second inequality follows by $\|\delta\|_{1}^{2}\leq16s_{0}\|\delta_{S_{0}}\|_{2}^{2}$.
Since with probability at least $1-1/p$,
\begin{align*}
 \|\hat{\Sigma}_{\tilde{X}}-\Sigma_{\tilde{X}}\|_{\infty}\leq\|\left(\mathbb{E}_{n}-\mathbb{E}\right)\tilde{X}_{i}\tilde{X}_{i}^{\top}\|_{\infty}+\|\mathbb{E}\left[\tilde{X}_{i}\tilde{X}_{i}^{\top}-\bar{X}_{i}\bar{X}_{i}^{\top}\right]\|_{\infty}=O_{p}\left(\sqrt{\log p/n}\right)
\end{align*}
where the equality follows from Bernstein inequality applied in Lemma \ref{lem_iotav} and Assumption \ref{assfeigen}(iii). Because $s_{0}\sqrt{\frac{\log p}{n}}=o(1)$, for large enough
$n$, we have $\Lambda_{\tilde{X},n}^{2}(s_{0})\geq\Lambda_{\tilde{X}}^{2}(s_{0})/2$
with probability at least $1-1/p$.
\end{proof}
\begin{lem}\label{lem_S}
Suppose that Assumptions \ref{Ass:parallel}-\ref{assum:chooselambda} are satisfied. Then $\mathbb{P}(\mathcal{T}_{5})=1 - O(1/p)$ and $\mathbb{P}(\mathcal{T}_{7})=1 - O(1/p)$.
\end{lem}
\proof 
The proof for $\mathcal{T}_{7}$ is similar to $\mathcal{T}_{5}$ and thus we focus on the proof of $\mathcal{T}_{7}$ below. From the definition of $S_i$, 
\begin{align*}
\hat S_i - S_i = & \hat\rho\Big(\Delta Y(i) -(1-\hat{\pi}_i)\hat{\Phi}_1(W_i) - \hat{\pi}_i\hat{\Phi}_0(W_i) \Big) - \rho_i \Big(\Delta Y_i -(1-\pi_i)\Phi_1(W_i) - \pi_i\Phi_0(W_i) \Big)  \\
& =  \underbrace{D_i(\Delta Y_{i} - \Phi_1(W_i))\left(\frac{1}{\hat{\pi}_i} -  \frac{1}{\pi_i}\right)}_{E2:1} + \underbrace{(1-D_i)(\Delta Y_{i} - \Phi_0(W_i))\left(\frac{1}{1-\hat{\pi}_i} -  \frac{1}{1-\pi_i}\right)}_{E2:2}\\
& + \underbrace{(\hat{\Phi}_1(W_i) - \Phi_1(W_i))\left(1-\frac{D_i}{{\pi}_i}\right)}_{E3:1}-\underbrace{(\hat{\Phi}_0(W_i) - \Phi_0(W_i))\left(1-\frac{1-D_i}{1-{\pi}_i}\right)}_{E3:2}\\
& + \underbrace{D_i(\hat{\Phi}_1(W_i) -  \Phi_1(W_i))\left(\frac{1}{{\pi}_i}-\frac{1}{\hat{\pi}_i}\right)}_{E4:1}-\underbrace{(1-D_i)(\hat{\Phi}_0(W_i) - \Phi_0(W_i))\left(\frac{1}{1-{\pi}_i}-\frac{1}{1-\hat{\pi}_i}\right)}_{E4:2}
\end{align*}
Let $S1$ to denote Sample 1 and $S2$ to denote Sample 2. We use $S1$ to estimate all nuisance functions $\hat{\Phi}_1(\cdot)$, $\hat{\Phi}_0(\cdot)$ and $\hat\pi(\cdot)$, and we use $S2$ to find $\hat \beta$ and $\hat f$. Similar to Lemma \ref{lem_iotav}, for $(\beta,f) \in \mathcal{M}_2(R)$,
\[
|(\hat{\bm{S}} - \bm{S})^{\top}\tilde{\bm{v}}_{n}(\bm{X}, \beta)/n|
=|(\hat{\bm{S}} - \bm{S})^{\top}\tilde{\bm{X}}(\beta-\beta_0)/n|\leq \|(\hat{\bm{S}} - \bm{S})^{\top}\tilde{\bm{X}}/n\|_{\infty}\|\beta-\beta_0\|_1.
\]
For each component of $\tilde X_{i}$,
\begin{align*}
\E_n\left[\tilde X_{ij}D_i(\Delta Y_{i} - \Phi_1(W_i))\left(\frac{1}{\hat{\pi}_i} - \frac{1}{\pi_i} \right)\right] = \E_n\left[\left(\frac{1}{\hat{\pi}_i} - \frac{1}{\pi_i} \right)\tilde X_{ij}D_i\epsilon_{1,i}\right].
\end{align*}
Since $\epsilon_{1,i}$ is independent of $\left(1/\hat{\pi}_i - 1/\pi_i \right)$ due to sample split and by Assumption \ref{assum:error}, we have $\E\left[\left(1/\hat{\pi}_i - 1/\pi_i \right)\tilde X_{ij}D(i)\epsilon_{1,i}\right]= 0$. 
Define $r_{ni}^{\pi} = \left(1/\hat{\pi}_i - 1/\pi_i \right)$. Lemma 14.15 in \cite{buhlmann2011statistics}  implies that for a constant $C_\pi$,
\begin{align*}
&\mathbb{P}\left(\max_{1\leq j\leq p}\left|\frac{1}{n}\sum_{i=1}^{n}r_{ni}^{\pi}\tilde{X}_{ij}D_{i}\epsilon_{1,i}\right|\geq\max_{1\leq i\leq n}r_{ni}^{\pi}\max_{1\leq j\leq p}\sqrt{\frac{1}{n}\sum_{i=1}^{n}\tilde{X}_{ij}^{2}D_{i}\epsilon_{1,i}^{2}}\sqrt{2\left(t^{2}+\frac{\log p}{n}\right)}\right)\\
&\leq\exp(-nt^{2}),
\end{align*}
where with probability bigger than $1-1/p$, for sufficiently large $C$

\[\max_{1\leq j\leq p}\left|\frac{1}{n}\sum_{i=1}^{n}\left(\tilde{X}_{ij}^{2}D_{i}\epsilon_{1,i}^{2}-\mathbb{E}\left[\tilde{X}_{ij}^{2}D_{i}\epsilon_{1,i}^{2}\right]\right)\right|\leq C\sqrt{\frac{\log p}{n}}.\]
Thus,
\[\max_{1\leq j\leq p}\left|\frac{1}{n}\sum_{i=1}^{n}r_{ni}^{\pi}\tilde{X}_{ij}D_{i}\epsilon_{1,i}\right|\leq2\max_{1\leq i\leq n}r_{ni}^{\pi}\sqrt{\frac{\log p}{n}}\]
with probability approaching one.

From Assumption \ref{assum:splitting1}, $\max_{1\leq i \leq n}(1/\hat{\pi}_{i}-1/\pi_{i})^2=O_p(1)$. Thus, take $t^{2} = \log p$, with probability at least $1-2/p$, $\|\E_n(   (E2:1) \tilde X_{i})\|_\infty \lesssim \sqrt{\log p/n}$. 

Next define $r_{ni}^{\Phi_1} = \hat{\Phi}_1(W_i) - \Phi_1(W_i)$. From sample splitting, 
$$
\mathbb{E}\left(r_{ni}^{\Phi_1} \tilde X_{ij}\left(1-\frac{D_i}{{\pi}_i}\right)\right) = \mathbb{E}\left(r_{ni}^{\Phi_1}\tilde X_{ij}\left(1-\frac{\mathbb{E}(D_i|X_i, Z_i, S1)}{{\pi}_i}\right)\right) = 0
$$
With constant $C_{\Phi_1}$, We can then apply the same bound such that
\begin{align*}
&\mathbb{P}\left(\max_{1\leq j\leq p}\left|\frac{1}{n}\sum_{i=1}^{n}r_{ni}^{\Phi_{1}}\tilde{X}_{ij}D_{i}\epsilon_{1,i}\right|\geq C_{\Phi_1} \max_{1\leq i\leq n}r_{ni}^{\Phi_{1}}\max_{1\leq j\leq p}\sqrt{\frac{1}{n}\sum_{i=1}^{n}\tilde{X}_{ij}^{2}D_{i}\epsilon_{1,i}^{2}}\sqrt{2\left(t^{2}+\frac{\log p}{n}\right)}\right)\\
&\leq\exp(-nt^{2}).
\end{align*}
From Assumption \ref{assum:splitting1}, $\|\E_n((E3:1) \tilde X_{i})\|_\infty = O_p(\sqrt{\log p/n})$. Lastly consider $E4:1$,
\begin{align*}
& \left|\mathbb{E}_n\left(\tilde X_{ij} D_i(\hat{\Phi}_1(W_i) -  \Phi_1(W_i))\left(\frac{1}{{\pi}_i}-\frac{1}{\hat{\pi}_i}\right) \right) \right|^2\\
& \leq  \mathbb{E}_n\left((\hat{\Phi}_1(W_i) - \Phi_1(W_i))^2 \cdot \left(\frac{1}{{\pi}_i}-\frac{1}{\hat{\pi}_i}\right)^2\right) \cdot \mathbb{E}_n\left(\tilde X_{ij}^2 D_i^2 \right) = O_p(\log p/n)	
\end{align*}
The last equality follows as $\mathbb{E}_n\left(\tilde X_{ij}^2 D_i^2 \right) = O_p(1)$. Similar results can be derived for $E2:2$, $E3:2$, and $E4:2$. Thus the first statement is proved. \\
\qed
\begin{lem}\label{lem_Phi}
Suppose that Assumptions \ref{Ass:parallel}-\ref{assum:chooselambda} are satisfied. Then $\mathbb{P}(\mathcal{T}_{6})=1 - O(1/p)$.
\end{lem}
\proof 
Similar as Lemma \ref{lem_S}, consider the interaction of each component of $\hat S_i - S_i$ with $\psi^{k_n}_{j}(Z_{i})$ for $j = 1, \cdots, k_n$.
 \begin{align*}
\E_n\left[\psi^{k_n}_{j}(Z_{i})D_i(\Delta Y_{i} - \Phi_1(W_i))\left(\frac{1}{\hat{\pi}_i} - \frac{1}{\pi_i} \right)\right] = \E_n\left[\left(\frac{1}{\hat{\pi}_i} - \frac{1}{\pi_i} \right)\psi^{k_n}_{j}(Z_{i})D_i\epsilon_{1,i}\right].
\end{align*}
By sample splitting and Assumption \ref{assum:error}, we have $\E\left[\left(1/\hat{\pi}_i - 1/\pi_i \right)\psi^{k_n}_{j}(Z_{i})D(i)\epsilon_{1,i}\right]= 0$. Again, Lemma 14.15 in \cite{buhlmann2011statistics} implies that for a constant $C_\pi$,
\begin{align*}
&\mathbb{P}\left(\max_{1\leq j\leq k_n}\left|\frac{1}{n}\sum_{i=1}^{n}r_{ni}^{\pi}\psi^{k_n}_{j}(Z_{i})D_{i}\epsilon_{1,i}\right|\geq\max_{1\leq i\leq n}r_{ni}^{\pi}\max_{1\leq j\leq k_n}\sqrt{\frac{1}{n}\sum_{i=1}^{n}\psi^{k_n}_{j}(Z_{i})^{2}D_{i}\epsilon_{1,i}^{2}}\sqrt{2\left(t^{2}+\frac{\log k_n}{n}\right)}\right)\\
&\leq\exp(-nt^{2}),
\end{align*}
where

\[\max_{1\leq j\leq k_n}\left|\frac{1}{n}\sum_{i=1}^{n}\left(\psi^{k_n}_{j}(Z_{i})^{2}D_{i}\epsilon_{1,i}^{2}-\mathbb{E}\left[\psi^{k_n}_{j}(Z_{i})^{2}D_{i}\epsilon_{1,i}^{2}\right]\right)\right|=O_{p}\left(\sqrt{\frac{\log k_n}{n}}\right).\]
So we have 
\[\max_{1\leq j\leq p}\left|\frac{1}{n}\sum_{i=1}^{n}r_{ni}^{\pi}\psi^{k_n}_{j}(Z_{i})D_{i}\epsilon_{1,i}\right|\leq2\max_{1\leq i\leq n}r_{ni}^{\pi}\sqrt{\frac{\log k_n}{n}}\]
with probability approaching one. From Assumption \ref{assum:splitting1}, $\max_{1\leq i \leq n}(1/\hat{\pi}_{i}-1/\pi_{i})^2=O_p(1)$. Thus, take $t^{2} = \log p$, with probability at least $1-1/p$, $\sqrt{k_n}\|\E_n(   (E2:1) \Psi_{n})\|_\infty\lesssim \sqrt{k_n\log k_n/n}$. 
Equation $E3:1$ and $E4:1$ can be bounded under the rate $\sqrt{k_n\log k_n/n}$ similarly as shown in Lemma \ref{lem_S} with assumption \ref{assum:splitting1} and thus we omit their proof. 
\qed

\subsection{Proof of Theorem \ref{thminf1}}
\begin{proof}
Consider the following decomposition such that for $\widehat{\Sigma}_{\tilde{X}}:=\mathbb{E}\tilde{X}_{i}\tilde{X}_{i}^{\top}$, 
\begin{align*}
\hat T-\xi^{\top}\beta_0&=\xi^{\top}(\hat\beta-\beta_0)-\hat w^{\top}\mathbb{E}_n [(\hat S_i-X_i^\top\hat\beta-\hat f(Z_i)\tilde X_i]\\
&=\xi^{\top}(\hat\beta-\beta_0)-\hat w^{\top}\mathbb{E}_n [(S_i-X_i^\top\beta_0- f_{0}(Z_i))\tilde X_i]-\hat w^{\top}\mathbb{E}_n [\hat S_i-S_i]\\
&~~~+\hat w^{\top}\mathbb{E}_n [X_i^\top(\hat\beta-\beta_0)-\hat f(Z_i)+f_{0}(Z_i))\tilde X_i]\\
&=-\hat w^{\top}\mathbb{E}_n [S_i-X_i^\top\beta_0- f_{0}(Z_i))\tilde X_i]-\hat w^{\top}\mathbb{E}_n [(\hat S_i-S_i)\tilde X_i]\\
&~~~+(\hat w\widehat{\Sigma}_{\tilde{X}}+\xi)^{\top}(\hat\beta-\beta_0)+\hat w^{\top}\mathbb{E}_n[(\Pi_{n,X_i|Z_i}{(\hat\beta-\beta_0))+(\hat f(Z_i)-f_{0}(Z_i)})\tilde X_i]\\
&:=-I_1-I_2+I_3+I_4.
\end{align*}
Recall that $\hat S_i = \hat{\rho}_i (\Delta Y_i -(1-\hat{\pi}_i)\hat{\Phi}_1(W_i) - \hat{\pi}_i\hat{\Phi}_0(W_i))$ and $S_i = \rho_i (\Delta Y_i -(1-\pi_i)\Phi_1(W_i) - \pi_i\Phi_0(W_i))$. We now analyze the four terms in the last expression one by one. 

For the first term,
\begin{align*}
&I_1-w_0^{\top}\mathbb{E}_n [S_i-X_i^\top\beta_0- f_{0}(Z_i))\tilde X_i]=(\hat w-w_0)^{\top}\mathbb{E}_n[S_i-X_i^\top\beta_0- f_{0}(Z_i))\tilde X_i]\\
&\leq\|\hat w-w_0\|_1\|\mathbb{E}_n (S_i-X_i^\top\beta_0- f_{0}(Z_i))\tilde X_i\|_\infty\\
&=O_p(s_w\sqrt{\log p/n}) \times O_p(\sqrt{\log p/n})= O_p(s_w(\log p/n))
\end{align*}
where $\|\hat w-w_0\|_1 = O_p(s_w\sqrt{\log p/n})$ as a Dantzig selector as defined in Theorem 7.1 in \cite{Bickel2009}, and the fact that $\|\mathbb{E}_{n}\epsilon_{i}\tilde{X}_{i}\|_{\infty}=O_{p}\left(\sqrt{\log p/n}\right)$ by Lemma E.1 and E.2 of  \cite{chernozhukov2017central}. Then Assumption \ref{assum: rate_thminf1} (i) guarantees that the remainder term in $I_1$ is $o_p(n^{-1/2})$. 

For the second term, 
\[
I_2\leq \|\hat w\|_1\|\mathbb{E}_n (\hat S_i-S_i)\tilde X_i\|_\infty\leq \|w_0\|_1\|\mathbb{E}_n (\hat S_i-S_i)\tilde X_i\|_\infty,
\]
where $\|w_0\|_1\leq s_w$ and $\|\hat w\|_1 \leq \|w_0\|_1$ because of the definition of $\hat w$. Then Lemma \ref{lem_S2} implies that $s_{w}\|\mathbb{E}_{n}\left[(\hat{S}_{i}-S_{i})\tilde{X}_{i}\right]\|_{\infty}=o_p(n^{-1/2})$ so we have $\sqrt{n}I_2= o_p(1)$.

For the third term,
$$
I_3\leq \|\hat w\widehat{\Sigma}_{\tilde{X}}+\xi\|_\infty\|\hat\beta-\beta_0\|_1\leq \lambda'\|\hat\beta-\beta_0\|_1=O_p(s_0{\frac{\log p}{n}}),
$$
where the last step follows from Theorem \ref{thm1} and $\lambda'=O_p(\sqrt{\log p/n})$ so $\sqrt{n}I_3=o_p(1)$ because $s_0\log p/\sqrt{n}=o_{p}(1)$.  

For the last term, $I_4$, by construction, $\mathbb{E}_{n}\left[\Pi_{n,X_{i}|Z_i}^{\top}\tilde{X}_{i}\right]=0$ and $\mathbb{E}_{n}\left[\left\{ \hat{f}(Z_{i})-f_{n0}(Z_{i})\right\} \tilde{X}_{i}\right]=0$. Moreover, Assumption \ref{assum: rate_thminf1} (i) implies that $s_w \max_{1\leq j\leq p,1\leq i \leq n}|\tilde{X}_{ij}r_{ni}|=o(n^{-1/2})$, so we have 
\[
\sqrt{n}I_{4}\leq \|w_{0}\|_{1} \|\mathbb{E}_{n}[r_{ni}\tilde{X}_{i}]\|_{\infty}=o_p(1).
\]
Combining the above results for $I_1$-$I_4$ and Assumption \ref{assum: rate_thminf1} (ii), we obtain
\[
\hat T-\xi^{\top}\beta_0=-w_0^{\top}\mathbb{E}_n \epsilon_i\tilde{X}_i+o_p(n^{-1/2})=-w_0^{\top}\mathbb{E}_n \epsilon_i\bar{X}_i+o_p(n^{-1/2}).
\]

Next, for $\Omega_{\beta}=\mathbb{E}\left[\sigma_{i}^{2}\bar{X}_{i}\bar{X}_{i}^{\top}\right]$, note that $w_0=\Sigma_{\bar{X}}^{-1}\xi$, $\xi^{\top}V_{\beta}\xi=w_{0}^{\top}\Omega_{\beta}w_{0}$,
\[
\mathbb{E}\left[\left(w^{\top}_{0}\Omega_{\beta}w_0\right)^{-1/2}\frac{1}{\sqrt{n}}\sum_{i=1}^{n}w_{0}^{\top}\bar{X}_{i}\epsilon_{i}\right]=0\]
because $\mathbb{E}[\epsilon_i\bar{X}_i]=0$, and 
\[
\mathbb{E}\left[\left(w^{\top}_{0}\Omega_{\beta}w_0\right)^{-1/2}\frac{1}{\sqrt{n}}\sum_{i=1}^{n}w_{0}^{\top}\bar{X}_{i}\epsilon_{i}\right]^{2}=1.\]
We want to verify the Lyapunov's condition for CLT. In particular, we wish to show that
\begin{align}
\frac{1}{(w_{0}^{\top}\Omega_{\beta}w_{0})^{r_{\epsilon}/4}}\sum_{i=1}^{n}\mathbb{E}\left[w_{0}^{\top}\bar{X}_{i}\epsilon_{i}/\sqrt{n}\right]^{r_{\epsilon}/2}=o_p(1). \label{eq:Lyapunov}
\end{align}

First because $\|w_{0}\|_{1}\leq s_{w}$, for $r_\epsilon>4$,
\begin{align*}
&\sum_{i=1}^{n}\mathbb{E}\left[w_{0}^{\top}\bar{X}_{i}\epsilon_{i}/\sqrt{n}\right]^{r_{\epsilon}/2}\leq\sum_{i=1}^{n}\mathbb{E}\left[\|w_{0}\|_{1}\|\bar{X}_{i}\epsilon_{i}/\sqrt{n}\|_{\infty}\right]^{r_{\epsilon}/2}\\
 & \leq \sum_{i=1}^{n}\left(\frac{s_{w}}{\sqrt{n}}\right)^{r_{\epsilon}/2}\max_{1\leq k\leq p}\mathbb{E}\left[\left|\bar{X}_{ik}\epsilon_{i}\right|^{r_{\epsilon}/2}\right]=O_{p}\left(\frac{s_{w}^{r_{\epsilon}/2}}{n^{r_{\epsilon}/4-1}}\right),
\end{align*}
where the last equaltion follows because $\left[|\bar{X}_{ik}\epsilon_{i}|\right]^{r_{\epsilon}/2}<\infty$ by the Cauchy-Schwarz inequality.
Moreover, because 
\[
\left(\xi^{\top}V_{\beta}\xi\right)^{r_{\epsilon}/4}\geq\left[\|\xi\|_{2}^{2}\Lambda_{\min}(\Omega_{\beta})\Lambda_{\min}^{2}(\Sigma_{\bar{X}}^{-1})\right]^{r_{\epsilon}/4}=p^{r_{\epsilon}/2}\left[\Lambda_{\min}(\Omega_{\beta})\Lambda_{\min}^{2}(\Sigma_{\bar{X}}^{-1})\right]^{r_{\epsilon}/4}>0
\]
 is bounded away from zero so (\ref{eq:Lyapunov}) is satisfied with Assumption \ref{assum:error}. Therefore,
\[
\left(w^{\top}_{0}\Omega_{\beta}w_0\right)^{-1/2}\frac{1}{\sqrt{n}}\sum_{i=1}^{n}w_{0}^{\top}\bar{X}_{i}\epsilon_{i}\rightarrow_{d} N(0,1),
\]
which implies that 
\[
\sqrt{n}(\hat{T}-\xi^{\top}\beta_{0})\rightarrow_{d}  N(0,w^{\top}_{0}\Omega_{\beta}w_0)= N(0,\xi^{\top}V_{\beta} \xi).
\]

Finally, we show that $\hat{V}_{\beta}\overset{p}{\rightarrow}V_{\beta}$.
Let $\hat{\Omega}_{\beta}=\frac{1}{n}\sum^{n}_{i=1}\hat{\epsilon}^{2}_i\tilde{X_i}\tilde{X}_{i}^{\top}$ and $\tilde{\Omega}_{\beta}=\frac{1}{n}\sum^{n}_{i=1}\left(\epsilon_i+r_{ni}\right)^{2}\tilde{X_i}\tilde{X}_{i}^{\top}$. 
Note that 
$\left|\hat{w}^{\top}\hat{\Omega}_{\beta}\hat{w}-\hat{w}^{\top}\tilde{\Omega}_{\beta}\hat{w}\right|\leq\|\hat{\Omega}_{\beta}-\tilde{\Omega}_{\beta}\|_{\infty}\|\hat{w}\|_{1}^{2}$,
$\left|\hat{w}^{\top}\tilde{\Omega}_{\beta}\hat{w}-\hat{w}^{\top}\Omega_{\beta}\hat{w}\right|\leq\|\tilde{\Omega}_{\beta}-\Omega_{\beta}\|\|\hat{w}\|_{1}$
and $\left|\hat{w}^{\top}\Omega_{\beta}\hat{w}-w_{0}^{\top}\Omega_{\beta}w_{0}\right|\leq\|\Omega_{\beta}\|_{\infty}\|\hat{w}-w_{0}\|_{1}$. Thus,
\begin{align}
\left|\hat{V}_{\beta}-V_{\beta}\right|= &\left|\hat{w}^{\top}\hat{\Omega}_{\beta}\hat{w}-w_{0}^{\top}\Omega_{\beta}w_{0}\right| \nonumber\\
\leq &\|\hat{w}\|_{1}^{2}\left(\|\hat{\Omega}_{\beta}-\tilde{\Omega}_{\beta}\|_{\infty}+\|\tilde{\Omega}_{\beta}-\Omega_{\beta}\|_{\infty}\right)+\|\hat{w}-w_{0}\|_{1}\|\Omega_{\beta}\|_{\infty} \nonumber\\
\leq & s^{2}_w\left(\|\hat{\Omega}_{\beta}-\tilde{\Omega}_{\beta}\|_{\infty}+\|\tilde{\Omega}_{\beta}-\Omega_{\beta}\|_{\infty}\right)+o_p(1). \label{V_beta}
\end{align}
We next consider to bound $\|\hat{\Omega}_{\beta}-\tilde{\Omega}_{\beta}\|_{\infty}$.
First note that Lemma E.1 and E.2 of \cite{chernozhukov2017central} imply that
\[
\max_{1\leq j\leq p}\left|\frac{1}{n}\sum_{i=1}^{n}\left(\tilde{X}_{ij}^{2}\epsilon_{i}^{2}-\mathbb{E}\left[\tilde{X}_{ij}^{2}\epsilon_{i}^{2}\right]\right)\right|=O_{p}\left(\sqrt{\frac{\log p}{n}}\right)\]
so 
\begin{equation}
s^{2}_w \|\tilde{\Omega}_{\beta}-\Omega_{\beta}\|_{\infty}=o_p(1) \label{eq: Omega_tilde}
\end{equation}
because of Assumption \ref{assum: rate_thminf1}(i).
Moreover, because
\begin{align}
\|\hat{\Omega}_{\beta}-\tilde{\Omega}_{\beta}\|_{\infty}=& \|\frac{1}{n}\sum_{i=1}^{n}\left(\hat{\epsilon}_{i}^{2}-\left(\epsilon_{i}+r_{ni}\right)^{2}\right)\tilde{X}_{i}\tilde{X}_{i}\|_{\infty} \nonumber \\
\leq & \max_{1\leq j,k \leq p}\left|\frac{2}{n}\sum_{i=1}^{n}\tilde{X}_{ij}\tilde{X}_{ik}\epsilon_{ni}\left(X_{i}^{\top}\left(\hat{\beta}-\beta_{0}\right)+\hat{f}(Z_{i})-f_{0}(Z_{i})\right)\right| \nonumber\\
 & +\max_{1\leq j,k \leq p}\left| \frac{1}{n}\sum_{i=1}^{n}\tilde{X}_{ij}\tilde{X}_{ik}\left(X_{i}^{\top}\left(\hat{\beta}-\beta_{0}\right)+\hat{f}(Z_{i})-f_{0}(Z_{i})\right)^{2}\right| \label{eq:varX}
\end{align}
where the first term on the RHS of (\ref{eq:varX}) is bounded by 
\begin{align*}
&2\max_{1\leq i\leq n}\max_{1\leq j,k \leq p}\left|\tilde{X}_{ij}\tilde{X}_{ik}\right|\|\bm{\epsilon}_{n}^{\top}\left(\bm{X}\left(\hat{\beta}-\beta\right)+\hat{f}(Z)-f(Z)\right)\|_{n}\\
\leq & 2\max_{1\leq i\leq n}\max_{1\leq j,k \leq p}\left|\tilde{X}_{ij}\tilde{X}_{ik}\right|\left(\|\frac{1}{n}\sum_{i=1}^{n}\left(\epsilon_{i}+r_{ni}\right)X_{i}\|_{\infty}\|\hat{\beta}-\beta_{0}\|_{1}+\sup_{z}\left|\hat{f}(z)-f(z)\right|\sqrt{\frac{1}{n}\sum_{i=1}^{n}\epsilon_{ni}^{2}}\right)\\
= & O_p\left(\log(np)\left( \frac{s_{0} \log p}{n}+\sqrt{\frac{\xi_{0}^{2}(k_{n}) k_{n}}{n}}+\ell_{k_{n}}c_{k_{n}}\right) \right)
\end{align*}

Moreover, the second term is bounded by 
\begin{align*}
&\max_{1\leq i\leq n}\max_{1\leq j,k \leq p}\left|\tilde{X}_{ij}\tilde{X}_{ik}\right|^{2}\frac{1}{n}\sum_{i=1}^{n}\left(X_{i}^{\top}(\hat{\beta}-\beta_{0})+\hat{f}(Z_{i})-f_{0}(Z_{i})\right)^{2}\\
&=O_{p}\left(\log(np)\left( \frac{s_{0}\log p}{n}+\sqrt{\frac{k_{n}}{n}}+\ell_{k_{n}}c_{k_{n}}\right)^{2}\right)
\end{align*}
where we use the sub-Gaussian property to obtain that $\max_{1\leq i\leq n}\max_{1\leq j,k\leq p}\left|X_{ij}X_{ik}\right|^{2}=O_{p}\left(\left(\log\left(np\right)\right)^{2}\right)$. 
Thus, with the additional assumption in Theorem \ref{thminf1}, we have
\begin{equation}
s^{2}_w\|\hat{\Omega}_{\beta}-\tilde{\Omega}_{\beta}\|_\infty=o_p(1) \label{eq:Omega_beta}
\end{equation}
Then $\left|\hat{V}_{\beta}-V_{\beta} \right|=o_p(1)$ follows from (\ref{V_beta}), (\ref{eq:Omega_beta}),  (\ref{eq: Sigma_rate}), and Assumption \ref{assum: rate_thminf1}. 
\end{proof}

\subsection{Proof of Theorem $\ref{thminf2}$}
\begin{proof}
We consider the following decomposition such that
\begin{align}
\bar{f}(z)-f_{n0}(z)=&\psi^{k_n}(z)^{\top}(\hat{\gamma}_{n}-\gamma_{n0}) \notag\\
&-\psi^{k_n}(z)^{\top}\hat  \Sigma^{-1}_{f}\mathbb{E}_n\left[\left(\hat S_i-X_i^\top\hat\beta - \psi^{k_n}(Z_i)^{\top}\hat{\gamma}_{n} \right)(\psi^{k_n}(Z_i)-\hat MX_i)\right] \notag\\
=&-\psi^{k_n}(z)^{\top}\hat  \Sigma^{-1}_{f}\mathbb{E}_n\left[\left( S_i-X_i^\top\beta - \psi^{k_n}(Z_i)^{\top}\gamma_{n0} \right)(\psi^{k_n}(Z_i)-\hat MX_i)\right] \notag\\
&-\psi^{k_n}(z)^{\top}\hat  \Sigma^{-1}_{f}\mathbb{E}_n\left[\left(\hat S_i-S_i\right)(\psi^{k_n}(Z_i)-\hat MX_i)\right] \notag\\
&-\psi^{k_n}(z)^{\top}\hat  \Sigma^{-1}_{f}\mathbb{E}_n\left[(\psi^{k_n}(Z_i)-\hat MX_i)X_i^{\top}\right](\hat \beta-\beta) \notag\\
&-\psi^{k_n}(z)^{\top}\Big(\hat  \Sigma^{-1}_{f}\mathbb{E}_n\left[(\psi^{k_n}(Z_i)-\hat MX_i)\psi^{k_n}(Z_i)^{\top}\right]-I_{k_n}\Big)(\hat \gamma_{n}-\gamma_{0n})\notag\\
:=&II_{1}+II_{2}+II_{3}+II_{4}\label{eq:f4}
\end{align}

The first term $II_{1}$ can be further expand as  
\begin{align}
&-\psi^{k_n}(z)^{\top}\hat  \Sigma^{-1}_{f}\mathbb{E}_n\left\{\left( S_i-X_i^\top\beta - \psi^{k_n}(Z_i)^{\top}\gamma_{n0}\right)(\psi^{k_n}(Z_i)-\hat MX_i)\right\}	\notag\\
=& -\psi^{k_n}(z)^{\top}\hat  \Sigma^{-1}_{f}\mathbb{E}_n\left\{\left( r_{ni}+ \epsilon_{i} \right)(MX_i- \hat MX_i)\right\}\label{eq:f11}\\
&-\psi^{k_n}(z)^{\top}\hat  \Sigma^{-1}_{f}\mathbb{E}_n\left\{\left( r_{ni} +  \epsilon_{i} \right)(\psi^{k_n}(Z_i)-  MX_i)\right\}.\label{eq:f12}
\end{align}

We first consider the term (\ref{eq:f11}) such that 
\begin{align*}
&\|\psi^{k_{n}}(z)^{\top}\hat{\Sigma}_{f}^{-1}\mathbb{E}_{n}\left[(\hat{M}-M)X_{i}\left( r_{ni} +  \epsilon_{i} \right)\right]\|_{2}\\
&\leq\|\psi^{k_{n}}(z)^{\top}\hat{\Sigma}_{f}^{-1}\|_{2}\sqrt{k_{n}}\|\mathbb{E}_{n}\left[(\hat{M}-M)X_{i}\left( r_{ni} +  \epsilon_{i} \right)\right]\|_{\infty}\notag\\
&\leq\|\psi^{k_{n}}(z)^{\top}\hat{\Sigma}_{f}^{-1}\|_{2}\sqrt{k_{n}}\|\hat{M}-M\|_{1}\|\mathbb{E}_{n}\left[X_{i}\left( r_{ni} +  \epsilon_{i} \right)\right]\|_{\infty}\notag\\
&=O_{p}\left(\sqrt{k_{n}}\left(s_{m}\sqrt{(\log k_{n}+\log p)/n}\right)\right)\times o_{p}\left(\sqrt{k_{n}/n}\sqrt{\log p/n}\right),
\end{align*}
where  $\|\hat{M}-M\|_{1}=O_{p}\left(s_{m}\sqrt{(\log k_{n}+\log p)/n}\right)$ from Lemma \ref{lemmaM}, $\max_{1\leq j\leq p,1\leq i\leq n}|X_{ij}r_{ni}|=o_{p}(n^{-1/2})$ by Assumption \ref{assum: rate_thminf1}(i) and $\|\mathbb{E}_{n}X_{i}\epsilon_{i}\|_{\infty}=O_{p}\left(\sqrt{\log p/n}\right)$ so the above term is $o_p(n^{-1/2})$ because of Assumption (\ref{assum: rate_thminf2}) (i). 
Also note that
\begin{align*}
(\ref{eq:f12})=&-\psi^{k_n}(z)^{\top}\left(\hat{\Sigma}^{-1}_{f}-\Sigma^{-1}_{f}\right)\mathbb{E}_n\left\{\left( r_{ni} +  \epsilon_{i} \right)(\psi^{k_n}(Z_i)-  MX_i)\right\} \nonumber\\
&-\psi^{k_n}(z)^{\top}\Sigma^{-1}_{f}\mathbb{E}_n\left\{\left( r_{ni} +  \epsilon_{i} \right)(\psi^{k_n}(Z_i)-  MX_i)\right\},
\end{align*}
where 
\begin{align}
\|\hat{\Sigma}_{f}-\Sigma_{f}\|_{2}\leq&\left\Vert \left(\mathbb{E}_{n}-\mathbb{E}\right)\left[\psi^{k_{n}}(Z_{i})\psi^{k_{n}}(Z_{i})^{\top}\right]\right\Vert _{2} \label{eq:f_Sigma_1}\\
&+\left\Vert \hat{M}\mathbb{E}_{n}[X_{i}X_{i}^{\top}]\hat{M}^{\top}-M\mathbb{E}\left[X_{i}X_{i}^{\top}\right]M^{\top}\right\Vert _{2}. \label{eq: f_Sigma_2}
\end{align}
The first term (\ref{eq:f_Sigma_1}) is bounded by
\begin{align}
\|(\mathbb{E}_n-\mathbb{E})[\psi^{k_n}(Z_i)\psi^{k_n}(Z_i)^{\top}]\|_{2}=O_p\left(\sqrt{\xi^{2}_0(k_n)\log k_n/n}\right).\label{eq:A.20}
\end{align}
following Lemma 6.2 in \cite{BCCK2015}. The second term (\ref{eq: f_Sigma_2})  is bounded by 
\begin{align}
&\left\Vert \hat{M}\mathbb{E}_{n}[X_{i}X_{i}^{\top}]\hat{M}-M\mathbb{E}\left[X_{i}X_{i}^{\top}\right]M\right\Vert _{2} \nonumber\\
&\leq \sqrt{k_{n}}\left\Vert \hat{M}\mathbb{E}_{n}[X_{i}X_{i}^{\top}]\hat{M}-M\mathbb{E}\left[X_{i}X_{i}^{\top}\right]M\right\Vert _{\infty} \nonumber\\
&\leq \sqrt{k_{n}}\|\hat{M}-M\|_1\|M\|_1\|\mathbb{E}_n\left[X_iX_i^{\top}\right]\|_{\infty}+\sqrt{k_{n}}\|M\|^{2}_1\|(\mathbb{E}_n-\mathbb{E})\left[X_iX_i^{\top}\right]\|_{\infty} \nonumber\\ 
& =O_{p}(s_{m}^{2}\sqrt{k_{n}\log p/n}) \label{eq:A.21},
\end{align}
where the first and second inequalities follow from direct calculation and the last inequality follows because $\lambda ''=O_p(\sqrt{\log p/n})$, $\|M\|_{1}\leq s_m$. Moreover, because $X_{ij}X_{ik}-\mathbb{E}[X_{ij}X_{ik}]$ is sub-exponential so by Bernstein inequality, we have for some constant $K_X$, 
\begin{align*}
\mathbb{P}\left(\max_{1\leq j,k\leq p}\left|\left(\mathbb{E}_{n}-\mathbb{E}\right)X_{ij}X_{ik}\right|>K_{X}\sqrt{\log2p/n}\right) \leq 1/2p.
\end{align*}
Thus, (\ref{eq:A.20}) and (\ref{eq:A.21}) imply that
\begin{equation}
\|\hat{\Sigma}_{f}-\Sigma_{f}\|_{2}=O_p\left(\sqrt{\xi^{2}_{0}(k_{n}) \log k_{n}/n}+s^{2}_{m}\sqrt{k_{n} \log p/n}\right)=o_p(1), \label{eq: Sigma_rate}
\end{equation}
where the second equality follows from Assumption \ref{assum: rate_thminf2}(i). 
Because 
\[
\sigma^{-1}_{z}\psi^{k_{n}}(z)^{\top}\Sigma^{-1}_{f} \mathbb{G}_{n}\epsilon_{i}(\psi^{k_{n}}(Z_{i})-MX_{i})=O_p(1)
\] 
and
\[
\sigma^{-1}_{z}\psi^{k_{n}}(z)^{\top}\Sigma^{-1}_{f} \mathbb{G}_{n}r_{ni}(\psi^{k_{n}}(Z_{i})-MX_{i})=O_{p}(\ell_{k_n}c_{k_n}\sqrt{k_n}),
\]
with (\ref{eq: Sigma_rate}), we have 
\[
\sqrt{n}\sigma_{z}^{-1}II_{1}=-\sqrt{n}\sigma_{z}^{-1}\psi^{k_{n}}(z)^{\top}\Sigma_{f}^{-1}\mathbb{E}_{n}\left\{ \left(\epsilon_{i}+r_{ni}\right)\left(\psi^{k_{n}}(Z_{i})-MX_{i}\right)\right\} +o_{p}(1).
\]
Similar to the argument in Theorem 2 of \cite{Newey1997}, with Assumption \ref{assum:error} and \ref{approx error}, the Lindbergh-Feller central limit theorem gives us
\begin{align*}
-\sqrt{n}\sigma_{z}^{-1}\psi^{k_{n}}(z)^{\top}\Sigma_{f}^{-1}\mathbb{E}_{n}\left\{\epsilon_{i}\left(\psi^{k_{n}}(Z_{i})-MX_{i}\right)\right\} \rightarrow_{d} N(0,1).
\end{align*}
and Assumption \ref{assum: rate_thminf2}(i) implies that 
\[
-\sqrt{n}\sigma_{z}^{-1}\psi^{k_{n}}(z)^{\top}\Sigma_{f}^{-1}\mathbb{E}_{n}\left\{r_{ni}\left(\psi^{k_{n}}(Z_{i})-MX_{i}\right)\right\}=o_{p}(1).
\]
Next we consider the term $II_2$. Note that 
\begin{align*}
&\sqrt{n}\sigma_{z}^{-1/2}\psi^{k_{n}}(z)^{\top}\Sigma_{f}^{-1}\mathbb{E}_{n}\left[(\hat{S}_{i}-S_{i})(\hat{M}-M)X_{i}\right]\\
&\leq\sqrt{k_{n}n}\|\sigma_{z}^{-1/2}\psi^{k_{n}}(z)^{\top}\Sigma_{f}^{-1}\|_{2}\|\mathbb{E}_{n}\left[(\hat{S}_{i}-S_{i})X_{i}\right]\|_{\infty}\|\hat{M}-M\|_{1}\\
&=O_p(s_m\sqrt{k_n}\sqrt{\log p/n})=o_{p}(1)
\end{align*}
and further from Lemma \ref{lem_S2}.
\[
\sqrt{n}\sigma_{z}^{-1/2}\psi^{k_{n}}(z)^{\top}\Sigma_{f}^{-1}\mathbb{E}_{n}\left[\left(\hat{S}_{i}-S_{i}\right)\left(\psi^{k_{n}}(Z_{i})-MX_{i}\right)\right]=o_{p}(1)
\]
Next, consider equation $II_{3}$. Consider the following decomposition:
\begin{align*}
&\left|\sigma_{z}^{-1/2}\psi^{k_n}(z)^{\top}\hat \Sigma^{-1}_{f}\mathbb{E}_n\left[(\psi^{k_n}(Z_i)-\hat MX_i)X_i^{\top}\right](\hat \beta-\beta)\right|\\
& \leq \|- \sigma_{z}^{-1/2}\psi^{k_n}(z)^{\top}\hat  \Sigma^{-1}_{f}\|_2 \cdot \sqrt{k_n}\left\|\mathbb{E}_n\left[(\psi^{k_n}(Z_i)-\hat MX_i)X_i^{\top}\right](\hat \beta-\beta)\right\|_\infty\\
&\leq \|- \sigma_{z}^{-1/2}\psi^{k_n}(z)^{\top}\hat  \Sigma^{-1}_{f}\|_2 \cdot \sqrt{k_n}\left\|\mathbb{E}_n\left[(\psi^{k_n}(Z_i)-\hat MX_i)X_i^{\top}\right]\right\|_\infty\left\|(\hat \beta-\beta)\right\|_1.
\end{align*}
From the definition of $\hat M_j$, $\left\|\mathbb{E}_n\left\{(\psi_j^{K_n}(Z_i)-\hat M_j^{\top}X_i)X_i^{\top}\right\}\right\|_\infty \leq \lambda''$ for all $j$; and from Theorem \ref{thm1}, $\|(\hat \beta-\beta)\|_1=O_p(s_0\sqrt{\log p/n})$, thus by choosing $\lambda''=O(\sqrt{\log p/n})$, $II_{3}=o_p(1)$ because $\sqrt{n}s_{m}s_{0}\log p/n=o_p(1)$ by Assumption \ref{assum: rate_thminf2}(i).

Finally the last term $II_{4}$ is 0, since $\hat \Sigma^{-1}_{f}\mathbb{E}_n\left\{(\psi^{k_n}(Z_i)-\hat MX_i)\psi^{k_n}(Z_i)^{\top}\right\} = I_{k_n} $.

When $\sqrt{n}\sigma^{-1}l_{k_n}c_{k_n}=o_p(1)$, we have 
\[
\sqrt{n}\sigma^{-1}_z(\tilde{f}(z)-f_0(z))=\sqrt{n}\sigma^{-1}_z(\tilde{f}(z)-f_{n0}(z))+o_{p}(1) \rightarrow_{d} N(0,1). 
\]

Next, we show the consistency of the variance term. Similar to Theorem 4.6 in \cite{BCCK2015}, we have
\[
\|(\mathbb{E}_n-\mathbb{E})[\sigma^{2}_i(Z_i,X_i)\psi^{k_n}(Z_i)\psi^{k_n}(Z_i)^{\top}]\|_{2}=O_p\left(\sqrt{\xi^{2}_0(k_n)\log k_n/n}\right).
\]

For $V_f=\Sigma^{-1}_f\Omega_f\Sigma^{-1}_f$ and $\sigma^2_{z}=\psi^{k_n}(z)^{\top}V_f\psi^{k_n}(z)$,
\[
\|\hat{V}_f-V_f\|_{2}\lesssim \|(\hat{\Sigma}^{-1}_f-\Sigma^{-1}_f)\hat{\Omega}_f\hat{\Sigma}^{-1}_f\|_{2}+\|\Sigma^{-1}_f(\hat{\Omega}_f-\Omega_f)\Sigma^{-1}_f\|_{2}
+\|\Sigma^{-1}_f\Omega_f(\hat{\Sigma}^{-1}_f-\Sigma^{-1}_f)\|_{2}.
\]
Note that for $\tilde{\Omega}_{f}=\mathbb{E}_{n}\left[\sigma_{i}^{2}\psi^{k_{n}}(Z_{i})\psi^{k_{n}}(Z_{i})^{\top}\right]-\hat{M}\mathbb{E}_{n}\left[\sigma_{i}^{2}X_{i}X_{i}^{\top}\right]\hat{M}$, 
\begin{align*}
\|\hat{\Omega}_{f}-\Omega_{f}\|_{2}&\leq\|\hat{\Omega}_{f}-\tilde{\Omega}_{f}\|_{2}+\|\tilde{\Omega}_{f}-\Omega_{f}\|_{2}.
\end{align*}
To bound $\|\hat{\Omega}_{f}-\tilde{\Omega}_{f}\|_{2}$, note that
\begin{align*}
&\|\mathbb{E}_{n}\left[\left(\hat{\sigma}_{i}^{2}-\sigma_{i}^{2}\right)\psi^{k_{n}}(Z_{i})\psi^{k_{n}}(Z_{i})^{\top}\right]\|_{2}\\
&\leq \max_{1\leq i\leq n}\left|X_{i}^{\top}\left(\hat{\beta}-\beta\right)+\hat{f}(Z_{i})-f_{0n}(Z_{i})\right|^{2}\|\mathbb{E}_{n}\left[\psi^{k_{n}}(Z_{i})\psi^{k_{n}}(Z_{i})^{\top}\right]\|_{2}\\
&\ \ +2\max_{1\leq i\leq n}\left|\epsilon_{i}+r_{ni}\right|\max_{1\leq i\leq n}\left|X_{i}^{\top}\left(\hat{\beta}-\beta\right)+\hat{f}(Z_{i})-f_{0n}(Z_{i})\right|\|\mathbb{E}_{n}\left[\psi^{k_{n}}(Z_{i})\psi^{k_{n}}(Z_{i})^{\top}\right]\|_{2}\\
&\lesssim_{p}\|\widehat{Q}_{z}\|_{2}O_{p}\left(\sqrt{\log np}\left(n^{1/r_{\epsilon}}+c_{k_{n}}\ell_{k_{n}}\right)\left(s_{0}\sqrt{\log p/n}+\sqrt{\xi_{0}^{2}(k_{n})k_{n}/n}+c_{k_{n}}\ell_{k_{n}}\right)\right)=o_p(1),
\end{align*}
where the second inequality follows from the results in Theorem \ref{thm1} and the last equality follows from the condition in Theorem \ref{thminf2}. 
Similarly,
\begin{align*}
&\|\hat{M}\mathbb{E}_{n}\left[\left(\hat{\sigma}_{i}^{2}-\sigma_{i}^{2}\right)X_{i}X_{i}^{\top}\right]\hat{M}^{\top}\|_{2}\leq\sqrt{k_{n}}\|\hat{M}\mathbb{E}_{n}\left[\left(\hat{\sigma}_{i}^{2}-\sigma_{i}^{2}\right)X_{i}X_{i}^{\top}\right]\hat{M}^{\top}\|_{\infty}\\
&\leq\sqrt{k_{n}}\|\hat{M}\|_{1}^{2}\|\mathbb{E}_{n}\left[\left(\hat{\sigma}_{i}^{2}-\sigma_{i}^{2}\right)X_{i}X_{i}^{\top}\right]\|_{\infty}
=\sqrt{k_{n}}s_{m}^{2}\left(s_{0}\sqrt{\frac{\xi_{0}^{2}(k_{n})\log p}{n}}+\ell_{k_{n}}c_{k_{n}}\right)
\end{align*}
so we have $\|\hat{\Omega}_{f}-\tilde{\Omega}_{f}\|_{2}=o_p(1).$

To bound $\|\tilde{\Omega}_{f}-{\Omega}_{f}\|_{2}=o_p(1)$, note that 
\[
\|\left(\mathbb{E}_{n}-\mathbb{E}\right)\left[\sigma_{i}^{2}\psi^{k_{n}}(Z_{i})\psi^{k_{n}}(Z_{i})^{\top}\right]\|_{2}=O_{p}\left((1+\ell_{k_{n}}c_{k_{n}})\sqrt{\frac{\xi_{0}^{2}(k_{n})\log k_{n}}{n}}\right)
\]
by Theorem 4.6 in \cite{BCCK2015}. Because
\[
\max_{1\leq j\leq p}\left|\frac{1}{n}\sum_{i=1}^{n}\left(X_{ij}^{2}\sigma_{i}^{2}-\mathbb{E}\left[X_{ij}^{2}\sigma_{i}^{2}\right]\right)\right|=O_{p}\left(\sqrt{\frac{\log p}{n}}\right),
\]
then  
\begin{align*}
&\|\hat{M}\mathbb{E}_{n}\left[\sigma_{i}^{2}X_{i}X_{i}^{\top}\right]\hat{M}^{\top}-M\mathbb{E}\left[\sigma_{i}^{2}X_{i}X_{i}^{\top}\right]M^{\top}\|_{2}\\
&\leq\sqrt{k_{n}}\|\hat{M}-M\|_{1}\|M_{1}\|\|\mathbb{E}_{n}\left[\sigma_{i}^{2}X_{i}X_{i}^{\top}\right]\|_{\infty}+\|M\|_{1}^{2}\|\left(\mathbb{E}_{n}-\mathbb{E}\right)\left[\epsilon_{i}^{2}X_{i}X_{i}^{\top}\right]\|_{\infty}\\
&=O_{p}\left(\sqrt{k_{n}}s_{m}^{2}\sqrt{\log p/n}\right)=o_p(1).
\end{align*}
The conclusion follows from triangular inequality.

\end{proof}

\begin{lem}\label{lem_S2}
Suppose that conditions in Theorem \ref{thminf1} and \ref{thminf2} are satisfied,  then

\[s_w\left\|\E_n \left[(\hat S_i - S_i)\tilde X_{i}\right]\right\|_\infty = o_p(1/\sqrt{n}) \]
\[\left\|\E_n \left[(\hat S_i - S_i)X_{i}\right]\right\|_\infty = o_p(1/\sqrt{n}) \]
\[\sqrt{k_n}\left\|\E_n \left[(\hat S_i - S_i)(\psi^{k_n}(Z_i)-\hat MX_i)\right]\right\|_\infty = o_p(1/\sqrt{n}) \]
\end{lem}
\proof
Similar to the proof in Lemma \ref{lem_S}, we consider the interaction of function $F_1 = \tilde X_{i}$, $F_2 =  X_{i}$ and $F_3 = (\psi^{k_n}(Z_i)-\hat MX_i)$ with each component of $\hat S_i - S_i$. With Bernstein inequality and union bound,
\[\mathbb{P}\left(\max_{1 \leq j \leq p}\left|\frac{1}{n}\sum_{i = 1}^n r_{ni}^{\pi} F_{kj} D_i(\epsilon_{1,i}) \right| \geq C_{\pi}\|r_{ni}^{\pi}\|_n \sqrt{(t+\log p)/n}) \right) \leq \exp(-t), \quad \text{  for } k = 1, 2 \text{ and } 3\]
Thus, when $k = 1$, taking $t = \log p$ and from Assumption \ref{assum:splitting}, with probability at least $1-1/p$, $s_w\|\E_n (F_1 \cdot (E2:1))\|_\infty = o_p(1/\sqrt{n})$. When $k=2$, we can apply the similar argument and with probability at least $1-1/p$, $\|\E_n(F_2 \cdot (E2:1))\|_\infty = o_p(1/\sqrt{n})$. And again when $k=3$, with probability at least $1-1/k_n$, $\sqrt{k_n}\|\E_n(F_3 \cdot (E2:1))\|_\infty = o_p(1/\sqrt{n})$. The same logic in Lemma \ref{lem_S} will lead to the rest of the terms and we omitted those here.   \\
\qed

\begin{lem}\label{lemmaM}
Let $M_j^{\top} = \mathbb{E}(\psi^{k_n}_j(Z_i)X_i^{\top})\{\mathbb{E}\left[X_iX_i^{\top}\right]\}^{-1}$, suppose that conditions in Theorem \ref{thminf2} are satisfied, then 
\[\|M - \hat M\|_1 = O_p(s_m\sqrt{(\log k_n +\log p) /n}).\]	
\end{lem}
\begin{proof}
Since 
\[\psi_j^{k_n}(Z_i) = M_j^{\top}X_i + (\psi_j^{k_n}(Z_i) - \mathbb{E}(\psi^{k_n}_j(Z_i)X_i^{\top})\{\mathbb{E}X_i^{\otimes 2}\}^{-1}X_i),\]
we define $\upsilon_i := \psi_j^{k_n}(Z_i) - \mathbb{E}(\psi^{k_n}_j(Z_i)X_i^{\top})\{\mathbb{E}\left[X_iX_i^{\top}\right]\}^{-1}X_i)$ such that
\begin{align*}
v_iX_i^{\top} &= \psi_j^{k_n}(Z_i)X_i^{\top} - \mathbb{E}(\psi^{k_n}_j(Z_i)X_i^{\top})\{\mathbb{E}X_i^{\otimes 2}\}^{-1}X_iX_i^{\top} \\
& = (\psi_j^{k_n}(Z_i)X_i^{\top} - \mathbb{E}(\psi^{k_n}_j(Z_i)X_i^{\top})) + \mathbb{E}(\psi^{k_n}_j(Z_i)X_i^{\top}) (I_{p\times p} - \{\mathbb{E}X_i^{\otimes 2}\}^{-1}X_iX_i^{\top}).
\end{align*}
Thus $\mathbb{E}(\upsilon_iX_i^{\top} ) = 0$, and the event $\mathcal{B} = \|\mathbb{E}_n(\upsilon_iX_i^{\top} )\|_\infty <\lambda''$ has probability at least $1-\exp(-cn\lambda''^2)$. Equation \eqref{eqhatm} is thus a linear Dantzig selector as defined in Theorem 7.1 in \cite{Bickel2009}. Therefore
\[\mathbb{P}(\|M - \hat M\|_1 > \lambda'') \leq  \sum_{j=1}^{K_n} \mathbb{P}(\|M_j - \hat M_j\|_1 >\lambda'') \leq \exp\left(\log k_n -cn\lambda''^2\right) \]  
By choose $\lambda'' \gtrsim \sqrt{(\log p +\log k_n)/n}$, with probability at least $1 - \exp(-c_1\log(k_n) - c_2 \log(p))$,
\[\|M - \hat M\|_1 = O_p(s_m\sqrt{(\log k_n +\log p)/n}).\]	
\end{proof}

\bibliographystyle{imsart-nameyear}
\bibliography{diffindiff}
\end{document}